\documentclass[11pt]{article}
\usepackage{tikz,tikz-cd}
\usepackage{subcaption}
\usepackage{pgfplots}\pgfplotsset{compat=1.18}

\usepackage[utf8]{inputenc}
\usepackage[T1]{fontenc}
\usepackage[english]{babel}
\usepackage{amsmath,amssymb,amsthm}
\usepackage{graphicx} 
\usepackage[margin=1.in]{geometry}
\usepackage[shortlabels]{enumitem}
\usepackage{todonotes}
\usepackage[bibstyle=numeric, backend=biber, sorting=none, citestyle=numeric-comp, giveninits,url=false,maxbibnames=5]{biblatex} 
\usepackage{csquotes}\MakeOuterQuote"
\usepackage{soul}\setstcolor{red}
\usepackage[colorlinks,citecolor=blue,linkcolor=blue]{hyperref}
\usepackage[capitalize]{cleveref}
\usepackage{comment}

\newtheorem{theorem}{Theorem}
\newtheorem{lemma}[theorem]{Lemma}
\newtheorem{corollary}[theorem]{Corollary}
\newtheorem{proposition}[theorem]{Proposition}

\newtheorem{definition}[theorem]{Definition}
\newtheorem{introtheorem}{Theorem} 

\newtheorem{introdefinition}[introtheorem]{Definition}

\theoremstyle{definition}

\newtheorem*{problem*}{Problem}
\newtheorem*{assumption*}{Assumption}
\newtheorem{example}[theorem]{Example}
\newtheorem{remark}[theorem]{Remark}
\newtheorem*{warning*}{Warning}

\newcommand{\F}{{\mc F}}

\newcommand{\ip}[2]{\langle #1,#2\rangle}
\newcommand{\braket}[2]{\langle #1|#2\rangle}
\newcommand{\ketbra}[2]{|#1\rangle\langle#2|}
\newcommand{\ket}[1]{|#1\rangle}
\newcommand{\kettbra}[1]{\ketbra{#1}{#1}}
\newcommand{\bra}[1]{\langle#1|}
\DeclareMathOperator{\supp}{supp}
\newcommand{\norm}[1]{\lVert #1\rVert}
\newcommand{\oo}{\infty}
\newcommand{\ox}{\otimes}
\newcommand{\mc}{\mathcal}
\newcommand{\eps}{\varepsilon}
\newcommand{\III}{{\mathrm{III}}}
\newcommand{\II}{{\mathrm{II}}}
\newcommand{\I}{{\mathrm{I}}}
\newcommand{\e}{\mathrm{e}}

\DeclareMathOperator{\Aut}{Aut}
\newcommand{\abs}[1]{\lvert #1 \rvert}
\newcommand{\up}[1]{^{(#1)}}

\DeclareMathOperator{\tr}{Tr}
\DeclareMathOperator{\Tr}{Tr} 

\renewcommand{\tilde}{\widetilde}
\renewcommand{\hat}{\widehat}

\newcommand{\hide}[1]{}

\def\A{{\mc A}}
\def\fA{{\mathfrak A}}
\def\B{{\mc B}}
\def\CC{{\mathbb C}}

\def\H{{\mc H}}
\def\K{{\mathcal K}}
\def\M{{\mc M}}
\def\N{{\mc N}}
\def\O{{\mc O}}
\def\Q{{\mc Q}}
\def\U{{\mc U}}
\def\W{{\mc W}}
\def\RR{{\mathbb R}}
\newcommand{\MM}{\mathbb M}
\def\NN{{\mathbb N}}
\renewcommand\P{\mc P}
\def\ZZ{{\mathbb Z}}
\def\QQ{{\mathbb Q}}
\newcommand{\R}{\mc R}
\newcommand{\proj}{\mathrm{Proj}}

\DeclareMathOperator{\lin}{span}
\DeclareMathOperator{\id}{id}

\DeclareMathOperator{\Sp}{Sp}

\newcommand{\diam}{\textup{diam}}

\newcommand{\placeholder}[0]{{\,\cdot\,}}

\renewcommand{\Re}{\mathrm{Re}}

\newcommand{\qandq}{\quad\text{and}\quad}

\title{Embezzlement of entanglement, quantum fields, and the classification of von Neumann algebras}

\author{Lauritz van Luijk, Alexander Stottmeister,\\ Reinhard F.~Werner, Henrik Wilming}
\date{\normalsize Institut f\"ur Theoretische Physik, Leibniz Universit\"at Hannover, \\ Appelstraße 2, 30167 Hannover, Germany\\[2ex]\today\\[2ex]
{\it Dedicated to the memory of Uffe Haagerup\footnote{
I (RFW) began the work on this project around 2011 with Volkher Scholz, at the time my PhD student. The aim was to establish the III$_1$ factor as the ``universal embezzling algebra'' in much the same way that the hyperfinite II$_1$ factor represents the idealized entanglement resource of infinitely many singlets. We were discussing this in a lobby at the 2012 ICMP congress in Aalborg when Uffe Haagerup walked by, and I decided to ask him about our problem. In a wonderfully rich conversation of about half an hour, he convinced us that the flow of weights should be the relevant thing to look at. Volkher and I decided to produce a paper explaining this convincingly to the QI community (and ourselves) and had planned to ask Uffe to be a coauthor once we were happy with our presentation. But alas, this project got stuck, and sadly Uffe passed away in the meantime. The current team took over in 2023, going far beyond what Volkher and myself had had in mind but vindicating Uffe's intuition at every turn. We dedicate this paper to his memory.}}}

\begin{document}

\maketitle

\begin{abstract}
We study the quantum information theoretic task of embezzlement of entanglement in the setting of von Neumann algebras.
Given a shared entangled resource state, this task asks to produce arbitrary entangled states using local operations without communication while perturbing the resource arbitrarily little.
We quantify the performance of a given resource state by the worst-case error. States for which the latter vanishes are `embezzling states' as they allow to embezzle arbitrary entangled states with arbitrarily small error.
The best and worst performance among all states defines two algebraic invariants for von Neumann algebras. 
The first invariant takes only two values. Either it vanishes and embezzling states exist, which can only happen in type III, or no state allows for nontrivial embezzlement. 
In the case of factors not of finite type I, the second invariant equals the diameter of the state space.
This provides a quantitative operational interpretation of Connes' classification of type III factors within quantum information theory. 
Type III$_1$ factors are 'universal embezzlers' where every state is embezzling. Our findings have implications for relativistic quantum field theory, where type III algebras naturally appear.
For instance, they explain the maximal violation of Bell inequalities in the vacuum. Our results follow from a one-to-one correspondence between embezzling states and invariant probability measures on the flow of weights. 
We also establish that universally embezzling ITPFI factors are of type III$_1$ by elementary arguments.
\end{abstract}

\clearpage

\thispagestyle{empty}
\clearpage
\tableofcontents

\clearpage
\section{Introduction}
Entanglement is often thought of as a precious resource that can be used to fulfill certain operational tasks in quantum information processing, notably quantum teleportation and quantum computation. 
It is only natural that such a resource should be consumed when put to use. Indeed, local operations with classical communication generally decrease the entanglement of a state unless the local parties only act unitarily.
Nonetheless, the phenomenon of \emph{embezzlement of entanglement}, discovered by van Dam and Hayden \cite{van_dam2003universal}, shows that there exist families of bipartite entangled states $\ket{\Omega_n}$ (with Schmidt rank $n$) shared between Alice and Bob such that any state $\ket\Psi$ with Schmidt rank $m$ may be extracted from them while perturbing the original state arbitrarily little and by acting only with local unitaries:
\begin{align}\label{eq:intro}
    u_Au_B \big(\ket{\Omega_n}_{AB}\otimes\ket{1}_A\ket1_B\big)\approx_\eps \ket{\Omega_n}_{AB} \ox \ket\Psi_{AB}
\end{align}
where $\eps\rightarrow 0$ for any fixed $m$ as $n\to \oo$. Here, $u_A$ and $u_B$ are suitable $\eps$- and $\ket\Psi$-dependent local unitaries applied by Alice and Bob, respectively. In \eqref{eq:intro} $\ket1_A\ket1_B$ denotes the product state $\ket1_A\ox\ket1_B$, and the indices $A/B$ are for emphasis only.
The family of states $\ket{\Omega_n}$ is hence referred to as an \emph{(universal) embezzling family}. 
As the resource state $\ket{\Omega_n}$ is hardly perturbed, it takes a similar role as a catalyst.
However, embezzlement is distinct from the phenomenon of \emph{catalysis of entanglement} pioneered by Jonathan and Plenio \cite{jonathan_entanglement-assisted_1999} because catalysts are typically only required to catalyze a single-state transition. Moreover, no state change is allowed on the catalyst; see \cite{Datta2022,lipka-bartosik_catalysis_2023} for reviews.
Besides the obvious conceptual importance of embezzlement, it has also found use as an important tool in quantum information theory, for example for the Quantum Reverse Shannon Theorem \cite{Bennett_2014,berta_quantum_2011} and in the context of non-local games \cite{Leung2013coherent,regev_quantum_2013,cleve_perfect_2017,coladangelo_two-player_2020}.

An obvious question is whether one can take the limit $n\rightarrow \infty$ in \eqref{eq:intro}, resulting in a state that allows for the extraction of arbitrary entangled states via local operations while remaining invariant. 
This would violate the conception of entanglement as the property of quantum states that cannot be enhanced via local operations and classical communication (LOCC). 
It is, therefore, perhaps unsurprising that the limit can not be taken in a naive way. 
Indeed, the original construction of \cite{van_dam2003universal} is given by
\begin{align}\label{eq:vandam}
	\ket{\Omega_n} = \frac{1}{\sqrt{H_n}}\sum_{j=1}^n \frac{1}{\sqrt j} \,\ket{jj},\quad H_n = \sum_{j=1}^n \frac{1}{j},
\end{align}
where $\ket{jj}$ denotes the product basis vector $\ket j\ox\ket j$.
Since $H_n\rightarrow \infty$ as $n\rightarrow \infty$, these vectors do not converge. 
It has been shown that the asymptotic scaling of Schmidt coefficients roughly as $\lambda_j \sim \tfrac{1}{j}$ is a general feature of embezzling families of states \cite{leung_characteristics_2014,zanoni_complete_2023}.
One can furthermore show using the Schmidt-decomposition that no state $\ket{\Omega}$ on a (possibly non-separable) Hilbert space $\H_A\ox\H_B$ can fulfill \eqref{eq:intro} with equality, i.e., with $\eps=0$, for all $\ket\Psi$ \cite{cleve_perfect_2017}. 
On the other hand, \cite{cleve_perfect_2017} also showed that $\eps=0$ is possible in a \emph{commuting operator framework} if $\ket\Omega$ is allowed to depend on $\ket\Psi$.  
In this work, we explore the ultimate limits of embezzlement in commuting operator frameworks.  
Specifically, we ask and answer the following natural questions:
\begin{enumerate}
	\item\label{question1} Can there be a single quantum state from which one can embezzle, with arbitrarily small error, every finite-dimensional entangled state, no matter how large its Schmidt rank? We call such states \emph{embezzling states}. 
    
	\item\label{question2} Can there be quantum systems where \emph{all} quantum states are  embezzling? We call such systems \emph{universal embezzlers}.
    
	\item\label{question3} Is there a difference between systems with individual embezzling states and universal embezzlers, if they exist, or are they all equivalent in a suitable sense?

	\item\label{question4} Can we expect to find embezzling states, or even universal embezzlers, as actual physical systems? 
\end{enumerate}
To formulate and answer the above questions in a mathematically precise and operationally meaningful way, we study embezzlement from the point of view of von Neumann algebras, which provides the most natural way to formulate bipartite systems beyond the tensor product framework.  Our results establish a deep connection between the classification of von Neumann algebras and the possibility of embezzlement.
Moreover, they imply that relativistic quantum fields are uniquely characterized by the fact that they are universal embezzlers.
Recently, the solution of Tsirelson's problem (see \cite{tsirelsons_prob} for background) and the implied negative solution of Connes' embedding conjecture in \cite{Ji2020} showed a fascinating connection between operational tasks in quantum theory and the theory of von Neumann algebras (see also \cite{goldbring_connes_2022} for an introduction).
Recall that a factor is a von Neumann algebra with trivial center. Factors can be classified into different types ($\I$, $\II$, and $\III$) and subtypes.
Factors of type $\I$ are classified into subtypes $\I_n$, corresponding to $\B(\H)$ with $n$-dimensional Hilbert space $\H$ for $n\in\NN\cup\{\infty\}$. 
Type $\II$ has subtypes $\II_1$ and $\II_\infty$.
The term semifinite factor is used to collectively refer to types $\I$ and $\II$. 
Connes showed that type $\III$ factors can be further classified into subtypes $\III_\lambda$ with $\lambda\in[0,1]$ \cite{connes1973classIII}. 
Connes' embedding problem and, hence, Tsirelson's problem, is related to the classification of type $\II_1$ factors. 
Our results, in turn, show that Connes' classification of type $\III$ factors may be interpreted as a quantitative measure of embezzlement of entanglement; see below.

In the remainder of this introduction, we give an (informal) overview of our methods and results and discuss some of their implications (see also the brief companion paper \cite{short_paper}).
For readers not familiar with von Neumann algebras, we provide some basic material in \cref{sec:preliminaries}.
From now on, we will stop using Dirac's ket-bra notation, except for basis vectors.

\subsection*{Bipartite systems and embezzling states}
After establishing the required mathematical background, we begin in \cref{sec:embezzlingstates} by formalizing a bipartite (quantum) system as a pair of von Neumann algebras $\M_A,\M_B$ acting on a Hilbert space $\H$, so that $\M_A = \M_B'$, where $\M_B'$ denotes the commutant of $\M_B$. 
That is, Alice and Bob have access to their respective local algebras of (bounded) operators to control and measure the shared quantum state $\Omega\in \H$. We refer to the condition $\M_A=\M_B'$ as Haag duality due to its importance in quantum field theory \cite{haag_local_1996}. It is automatically fulfilled in the tensor product framework (see \cref{tab:co-vs-tp} for an overview) and can be interpreted as saying that Alice can implement any symmetry of Bob and vice-versa, see also \cite{van_luijk_schmidt_2023}. 
We call a bipartite system $(\H,\M,\M')$ \emph{standard} if $\M$ is in so-called \emph{standard representation}, see \cref{subsec:pre:weights}. For our purposes, this condition simply means that every (normal) state $\omega\in S_*(\M)$ on $\M$ arises as the marginal of some vector $\Omega\in\H$ and the same is true for $\M'$ (see \cref{lem:standard}). Thus, in a standard bipartite system, every state on $\M$ and $\M'$ has a purification, just as in standard quantum mechanics.

\begin{introdefinition}[Embezzling state] We call a unit 
vector $\Omega\in\H$ an {\bf embezzling state} if for any $n\in\NN$, any $\eps>0$ and any state vector $\Psi\in \CC^n\ox \CC^n$ there exist unitaries $u_A\in \M_A\ox M_n\ox 1$ and $u_B\in \M_B\ox 1\ox M_n$ such that
\begin{align}
 \norm{u_A u_B(\Omega\ox\ket{11}) - \Omega\ox\Psi} < \eps,	
\end{align}
where $M_n$ is the algebra of $n\times n$ complex matrices.
\end{introdefinition}

When considering entanglement theory for pure bipartite states in the tensor product framework, a crucial role is played by Nielsen's theorem \cite{nielsen_conditions_1999}. It reduces the study of transformations via LOCC to majorization theory on Alice's marginals.
Similarly, we next show that whether $\Omega$ performs well at embezzlement can equally well be discussed on the level of the induced state $\omega$ on Alice's algebra $\M=\M_A$ or the induced state $\omega'$ on Bob's algebra $\M'=\M_B$.
Similarly to the definition for biparite pure states $\Omega\in\H$, we call a state $\omega$ on $\M$ a \emph{(monopartite) embezzling} state if for every $n\in\NN$, any $\eps>0$, and any two states $\phi,\psi$ on $M_n$ there exists a unitary $u\in\M$ such that
\begin{align}
	\norm{u (\omega\ox\phi)u^* - \omega\ox\psi}<\eps.
\end{align}

\begingroup
\setlength{\tabcolsep}{10pt} 
\renewcommand{\arraystretch}{1.5} 
\begin{table}[t!]
    \centering
    \begin{tabular}{c|cc}
    &Commuting Operator& Tensor Product \\
	    \hline\hline
	    Hilbert space & $\H$ & $\H_A\ox \H_B$ \\
	    States & $\Omega \in \H$ & $\Omega\in\H_A\ox \H_B$\\
	    Alice's algebra & $\M_A\subseteq \B(\H)$ & $\B(\H_A)\ox 1$\\
	    Bob's algebra & $\M_B\subseteq \M_A'\subseteq \B(\H)$ & $1\ox\B(\H_B)$\\
	    Haag duality & $\M_A=\M_B'$ & automatic     
    \end{tabular}
    \caption{Commuting operator framework vs.\ tensor product framework. }
    \label{tab:co-vs-tp}
\end{table}
\endgroup

\begin{introtheorem}[cf.\ Thm.~\ref{thm:partite}] For any bipartite system $(\H,\M,\M')$, a state $\Omega\in\H$ is an embezzling state if and only if its induced states $\omega$ and $\omega'$ on $\M$ and $\M'$, respectively, are monopartite embezzling states.
\end{introtheorem}

Since we assume Haag duality, the study of embezzlement thereby reduces to studying monopartite embezzlement on von Neumann algebras.
All our results can, therefore, also be interpreted in this monopartite setting without reference to entanglement but rather simply as a question on the state transitions that can be (approximately) realized on a von Neumann algebra via unitary transformations.

\subsection*{Spectral properties of embezzling states}

It was shown in \cite{leung_characteristics_2014} that the Schmidt spectrum of embezzling families, like the van Dam-Hayden family \eqref{eq:vandam}, essentially behaves like $\lambda_j \sim \frac1j$ (see also \cite{zanoni_complete_2023}).
This poses the following question: What are the spectral properties of embezzling states on von Neumann algebra?

Unlike for matrix algebras, the spectrum of a state on a von Neumann algebra is not defined in general. However, one can make sense of the \emph{modular spectrum}, which is defined as the spectrum of the corresponding modular operator $\Delta_\omega$, a positive self-adjoint operator constructed in Tomita-Takesaki theory (see \cite{takesaki2} or \cref{sec:preliminaries}).
For a state $\omega$ on a matrix algebra, the modular spectrum is the set of ratios $\{\frac{\lambda_i}{\lambda_j}\}_{i,j}$ of eigenvalues $\lambda_j$ of the density operator $\omega$.
Thus, the $\frac1j$-result of \cite{leung_characteristics_2014} suggests that we should get a modular spectrum of $\RR^+$ (since ${\frac1j} \big/ {\frac1k}=\frac kj$ yields every positive rational number).
We show that this is indeed the case: If $\omega$ is a (monopartite) embezzling state on a von Neumann algebra $\M$ then 
\begin{equation}
    \Sp\Delta_\omega =\RR^+.
\end{equation}
The converse is false: there exist states on $\M=\B(\H)$ with a modular spectrum of $\RR^+$. 
The fact that embezzling states have the maximal possible modular spectrum implies that universal embezzlers, i.e., von Neumann algebras where every state is embezzling, are of type $\III_1$.

We now consider semifinite von Neumann algebras. These allow us to consider the spectrum of states directly (as opposed to the modular spectrum discussed above).
Our discussion is based on the notion of \emph{spectral scales}.
For a state $\omega$ on $M_n$ with density operator $\rho$, the spectral scale is the function $\lambda_\rho:\RR^+\to\RR^+$ given by

\begin{equation}\label{eq:intro_spectral_scale}
    \lambda_\rho(t) = \sum_i p_i \,\chi_{[d_{i-1},d_{i})}(t),
\end{equation}
where $p_i$ are the eigenvalues of $\rho$, $d_i =m_i + ... +m_1$ with $m_i$ being the multiplicity of $p_i$ and $d_0=0$.
Clearly, the spectral scale is in one-to-one correspondence with the spectrum (with multiplicities) and completely determines unitary equivalence. 
In particular, the reduced states $\omega_n$ of the van Dam-Hayden embezzling family $\Omega_n$ have spectral scale 
\begin{align}
	\lambda_{\omega_n}(t) = \frac{1}{H_n}\sum_{j=1}^n \frac{1}{j}\,\chi_{[j-1,j)}(t)	
\end{align}
resembling a step-function approximation of the function $1/(H_n t)$ up to $t=n$.
Readers familiar with majorization theory will notice that the spectral scale is essentially the \emph{decreasing rearrangement} of eigenvalues so that $\rho$ majorizes $\sigma$ \cite{bhatia_matrix_1997} if and only if
\begin{align}
    \int_0^t \lambda_\rho(s)\, ds \geq \int_0^t \lambda_\sigma(s)\, ds \qquad\forall t>0.
\end{align}
Spectral scales can also be defined for states on semifinite von Neumann algebras. Generalizing \eqref{eq:intro_spectral_scale}, the spectral scale $\lambda_\omega(t)$ of a state $\omega$ is a right-continuous decreasing probability density on $\RR^+$. 
While spectral scales require semifinite von Neumann algebras, we consider on a general von Neumann algebra right-continuous functions $\tilde \lambda_\omega:(0,\infty) \to \RR^+$ defined for all (normal) states $\omega$ on $\M$ and $M_n(\M)$ that share the most important properties of the spectral scale: a) two (approximately) unitarily equivalent states $\omega,\varphi$ have the same function $\tilde \lambda_\omega=\tilde\lambda_\varphi$ and b) the function behaves as the spectral scale under tensor products.
There exist nontrivial functions $\tilde \lambda_\omega$ on type $\III_\lambda$ factors satisfying the conditions a) and b) (see \cref{sec:FOW_type_III_lambda}) and \cite{haagerup1990equivalence}.

\begin{introtheorem}[cf.\ Prop.~\ref{thm:spectral_charac}] If $\tilde \lambda_\omega$ fulfills the above properties a) and b), then 
\begin{align}\label{eq:mbz_implies_1/t}
	\omega\ \text{embezzling}\quad\implies\quad \tilde\lambda_\omega(t) \propto \frac{1}{t}\,.
\end{align}
\end{introtheorem}
Since the spectral scale on a semifinite von Neumann algebra is guaranteed to be integrable, it follows that embezzling states cannot exist on semifinite von Neumann algebras.

\subsection*{Embezzlement and the flow of weights}

We mentioned above the crucial role that Nielsen's theorem plays in entanglement theory because it allows us to study entanglement via purely classical majorization theory.
To study embezzlement in general von Neumann algebras, we use the so-called \emph{flow of weights} \cite{takesaki2}.
The flow of weights is a classical dynamical system $(X,\mu,F)$ that can be associated with a von Neumann algebra $\M$ in a canonical way. It consists of a standard Borel space $(X,\mu)$ and a flow $F=(F_t)_{t\in\RR}$, i.e., a one-parameter group of non-singular Borel transformations.
The flow of weights captures important information about $\M$. Most importantly, it is ergodic if and only if $\M$ is a factor.
Haagerup and St\o rmer found a canonical way to associate probability measures $P_\omega$ on $X$ to normal states $\omega$ on $\M$ \cite{haagerup1990equivalence}.
In the case of semifinite factors, the flow of weights and the map $\omega\mapsto P_\omega$ can be seen as a generalization of majorization theory: the flow of weights simply yields dilations on $\RR_+$ and $P_\omega$ is a generalization of the spectrum, similar to spectral scales (cf.\ \cref{sec:FOW_semifinite}). 
The crucial property of the flow of weights for us is that two normal states $\omega_1,\omega_2$ on $\M$ are approximately unitarily equivalent if and only if $P_{\omega_1}=P_{\omega_2}$. In fact, it was shown by Haagerup and St\o rmer \cite{haagerup1990equivalence} that (see \cref{thm:distance_of_unitary_orbits})
\begin{align}
\inf_{u\in\U(\M)} \norm{u\omega_1 u^*-\omega_2} = \norm{P_{\omega_1}-P_{\omega_2}},
\end{align}
where the distance of probability measures on $X$ is measured with the $L^1$-distance of their densities with respect to $\mu$.
This property is crucial for all of our results because it allows us to reduce the problem of studying embezzlement to studying the classical dynamical systems $(X,\mu,F)$. 
The interesting objects for us on von Neumann algebras are unitary orbits of embezzling states. Natural objects to consider on classical dynamical systems are stationary (invariant) probability measures. As we will see shortly, the two are in one-to-one correspondence.

\subsection*{Quantification of embezzlement}

Our main result relates the classification of von Neumann algebras to how well a given factor performs at the task of embezzlement.
To quantify how capable a given state $\omega$ is at embezzling, we define 
\begin{equation}\label{introeq:kappa_def}
    \kappa(\omega) =  \sup_{\psi,\,\phi} \inf_{u} \norm{\omega\ox \psi - u(\omega\ox\phi)u^*},
\end{equation}
where the supremum is over all states $\psi,\phi$ on $M_n$ (and over all $n\in\NN$) and where the infimum is over all unitaries $u\in M_n(\M)$.
The quantifier $\kappa(\omega)$ measures the worst-case performance of $\omega$ in embezzling finite-dimensional quantum states. 
Clearly, an embezzling state fulfills $\kappa(\omega)=0$. 
Moreover, for any factor $\M$ we introduce the algebraic invariants
\begin{equation}
	\kappa_{\textit{min}}(\M) = \inf_{\omega\in S_*(\M)} \kappa(\omega) \qandq\kappa_{\textit{max}}(\M) = \sup_{\omega\in S_*(\M)} \kappa(\omega),
\end{equation}
which measure the best and worst worst-case performance of all normal states on a factor $\M$, respectively.
Our main technical tool now allows us to connect $\kappa(\omega)$ with the flow of weights: 
\begin{introtheorem}[cf.~Thm.~\ref{thm:kappa}]
    $\kappa(\omega)$ measures precisely how much $P_\omega$ is invariant under the flow of weights:
\begin{align}\label{eq:intro_kappa}
	\kappa(\omega) = \sup_{t\in\RR} \norm{F_t(P_\omega)- P_\omega}.
\end{align}
\end{introtheorem}
In \eqref{eq:intro_kappa}, $F_t(P_\omega)$ denotes the probability measure defined by $F_t(P_\omega)(A) = P_\omega(F_{-t}(A))$.
Since the flow is ergodic, if $\M$ is a factor, there can be, at most, one invariant measure corresponding to a single unitary orbit of embezzling states. 
Using this tool, we can now evaluate $\kappa_{\textit{min}}$ and $\kappa_{\textit{max}}$ for the different types of factors. 
On a technical level, this yields our main result:

\begin{introtheorem}[cf.\ Thms.~\ref{thm:min} and \ref{thm:max}]\label{introthm:tablekappa} Let $\M$ be a factor. The invariants $\kappa_{\textit{min}}$ and $\kappa_{\textit{max}}$ take the following values:
\begingroup
\setlength{\tabcolsep}{10pt} 
\renewcommand{\arraystretch}{1.5} 
\begin{equation}
    \centering
    \begin{tabular}{c|ccccc}
    &type $\I$& type $\II$&type $\III_0$&type $\III_\lambda$, $0<\lambda<1$&type $\III_1$\\\hline\hline
    $\kappa_{\textit{min}}(\M)$&$2$&$2$&$\in\{0,2\}$&$0$&$0$\\
    $\kappa_{\textit{max}}(\M)$&$2$&$2$&$2$&$2\frac{1-\sqrt\lambda}{1+\sqrt\lambda}$&$0$
    \end{tabular}
\end{equation}
\endgroup
Type $\III_0$ factors fall into two categories: Either there exists an embezzling state and $\kappa(\omega)<2$ for all faithful states $\omega$, or there exists no embezzling state and $\kappa(\omega)=2$ for all states $\omega$.
\end{introtheorem}

\cref{introthm:tablekappa} shows that semifinite factors not only do not admit embezzling states (as discussed above based on spectral scales) but maximally fail to do so: Even $\kappa_{\textit{min}}(\M)$ attains the maximal possible value $2$.
The same holds for a certain class of type $\III_0$ factors.
To obtain $\kappa_{\textit{max}}=2$ for the case of III$_0$ factors, we make use of Gelfand theory to reduce the problem to one on aperiodic, topological dynamical systems instead of a measure-theoretic properly ergodic ones. 
For the case of III$_\lambda$ with $\lambda>0$, we provide concrete examples of states that reach the given value of $\kappa_{\textit{max}}$. 
Remarkably, we have a complete dichotomy: Either $\kappa_{\textit{min}}(\M)=0$, in which case there exists an embezzling state, or $\kappa_{\textit{min}}(\M)=2$ meaning that every state maximally fails at the task of embezzlement. In particular, the existence of a single faithful state with $\kappa(\omega)<2$ already guarantees the existence of an embezzling state.

While semifinite factors do not admit embezzling states, the situation is very different for type $\III_\lambda$ factors with $\lambda>0$. First, every such factor admits an embezzling state, answering question \ref{question1} affirmatively. Second $\kappa_{\textit{max}}(\M)$ monotonically decreases to $0$ as $\lambda$ approaches $1$. Thus, for $\lambda\approx1$, every state is approximately embezzling. In particular, a quantum system is a universal embezzler \emph{if and only if} it is described by a type $\III_1$ factor. This answers questions \ref{question2} and \ref{question3}. 

An interesting observation is that for factors of type $\III$, we can recover their subtype from $\kappa_{\textit{max}}$. Thus, at least in principle, the operational task of embezzlement allows one to obtain Connes' classification of type $\III$ factors. The values taken by $\kappa_{\textit{max}}$ for type $\III$ factors are well-known as the \emph{diameter of the state space} \cite{connes1985diameters}. To define the diameter of the state space, one considers the quotient of the normal state space $S_*(\M)$ modulo approximate unitary equivalence:
\begin{equation}
    \omega\sim\varphi :\iff \forall_{\eps>0}\ \exists_{u\in\U(\M)} : \norm{\omega- u\varphi u^*}<\eps,\qquad\omega,\varphi\in S_*(\M).
\end{equation}
We then have 
\begin{align}\label{eq:kappa_max_diameter}
	\kappa_{\textit{max}}(\M) = \diam(S_*(\M)/\sim) = 2\, \frac{1-\sqrt\lambda}{1+\sqrt\lambda}, \qquad 0\le \lambda\le1,
\end{align}
where the diameter is measured in terms of the induced distance.
We note that the diameter of the state space for type $\II$ factors is $2$ and for type $\I_n$ factors is $2(1-\frac1n)$ \cite{connes1985diameters}. Therefore, $\kappa_{\textit{max}}$ is equal to the diameter unless $\M$ is a finite matrix algebra.

Quite remarkably, even though $\kappa$ is defined only in terms of embezzlement of states on finite-dimensional matrix algebras, it actually bounds the performance for embezzlement on factors of arbitrary type:

\begin{introtheorem}[cf.\ Thm.~\ref{thm:kappa_bound}]\label{introthm:kappa_bound}
    Let $\omega$ be a normal state on a von Neumann algebra $\M$ and $\psi,\phi$ be normal states on a hyperfinite factor $\P$. Then 
    \begin{equation}\label{introeq:kappa_bound}
        \inf_{u\in \U(\M\ox\P)} \norm{u(\omega\ox\psi)u^*-\omega\ox\phi} \le \kappa(\omega).
    \end{equation}
\end{introtheorem}

In particular, if $\omega$ is an embezzling state, it may embezzle state transitions between arbitrary states even on (hyperfinite) type $\III$ factors. We suspect that the assumption of hyperfiniteness can be dropped.
In fact, we know that \cref{introthm:kappa_bound} holds for a much larger class of factors $\P$, encompassing all those that are semifinite or type $\III_{0}$.

\subsection*{Embezzlement and infinite tensor products}

Our results about embezzlement in general von Neumann algebras rely on the flow of weights. While elegant and powerful, the flow of weights requires the full machinery of modular theory. It is, therefore, desirable to have a simpler argument to show that universal embezzlers must have type $\III_1$.
We provide such an argument for infinite tensor products of finite type $\I$ (ITPFI) factors in \cref{sec:itpfi}. Our argument is elementary in the sense that it does not rely on modular theory.

ITPFI factors are special cases of \emph{hyperfinite}, also called \emph{approximatly finite dimensional}, von Neumann algebras, which by definition allow for an (ultraweakly) dense filtration by matrix algebras. 
Therefore, the von Neumann algebras found in physics are typically hyperfinite.
Importantly, there are ITPFI factors for every (sub-)type of the classification of factors mentioned above \cite{araki1968factors}. 
Moreover, it is an important result in the classification of von Neumann algebras that every hyperfinite factor (with separable predual) of type $\III$$_\lambda$ with $\lambda>0$ is isomorphic to the respective Powers factor \cite{powers1967uhf} for $\lambda<1$ and the Araki-Woods factor $\R_\infty$ \cite{araki1968factors} for type $\III$$_1$.
Connes showed the cases $0<\lambda<1$ \cite{connes1976injective} while Haagerup proved the case $\lambda=1$ in \cite{haagerup_uniqueness_1987}.
Thus, our direct argument for ITPFI factors covers all hyperfinite factors (with separable predual) apart from those of type $\III$$_0$.

An ITPFI factor $\M=\bigotimes_{j=1}^\oo (\M_j,\rho_j)$ is specified by a sequence of finite type $\I$ factors $\M_j$ with reference states $\rho_j$. 
The type of $\M$ only depends on the asymptotic behavior of the states $(\rho_j)$: modifying or removing any finite number of them results in an algebra isomorphic to $\M$.
Our argument, which we sketch here, relies on the fact that on an ITPFI factor, every normal state $\omega$ may be approximated to arbitrary precision as a tensor product of the form
\begin{align}
    \omega \approx_\eps \omega_{\leq n}\ox \rho_{>n},
\end{align}
where $\rho_{>n} = \otimes_{j>n}^\infty\rho_j$ and $\omega_{\leq n}$ is a state on $\M_{\leq n}:=\otimes_{j=1}^n\M_j$. 
If $\M$ is a universal embezzler, then the states $\rho_{>n}$ must all be embezzling states.  
Hence $\omega_{\leq n}\ox \rho_{>n}$ is (approximately) unitarily equivalent to $\sigma_{\leq n}\ox \rho_{>n}$ for any normal state $\sigma_{\leq n}$ on $\M_{\leq n}$. Consequently, all normal states $\omega$ and $\sigma$ on $\M$ are (approximately) unitarily equivalent. Therefore, the diameter of the state space $\M$ is $0$, which happens if and only if $\M$ has type $\III$$_1$ \cite{connes_homogeneity_1978}.
Conversely, since type $\III$$_1$ factors are properly infinite, we have $\M\cong \M\otimes M_n$. Hence, the fact that the diameter of the state space is $0$ directly implies that III$_1$ factors are universal embezzlers. We thus find:

\begin{introtheorem}[cf.\ Cor.~\ref{cor:itpfi-III1}]
    Let $\M$ be an ITPFI factor. $\M$ is universally embezzling if and only if $\M$ is the unique hyperfinite factor of type $\III_{1}$, i.e., $\M\cong\R_{\infty}$.
\end{introtheorem}

Since the unique hyperfinite type $\III_1$ factor is an ITPFI factor, universal bipartite embezzling extends from all pure states to all bipartite mixed states:
\begin{introtheorem}[cf.\ Cor.~\ref{cor:itpfi-III1-mixed}]
    In a bipartite system $(\H,\M,\M')$, where $\M$ (and hence $\M'$) are hyperfinite type $\III_1$ factors, every density operator $\rho$ on $\H$ is embezzling: For every unit vector $\Psi\in\CC^n\ox\CC^n$ and every $\eps>0$, there exist unitaries $u \in \M\ox M_n\ox 1$ and $u'\in\M\ox1\ox M_n$ such that
    \begin{equation}
        \norm{uu' (\rho\ox \kettbra{11}) u^*u'^* - \rho \ox \kettbra\Psi}_1 <\eps.
    \end{equation}
\end{introtheorem}

\subsection*{Exact embezzlement}
One core motivation for our work is to understand in which sense a single quantum system may serve as a good resource for embezzlement of arbitrary pure quantum states. 
We emphasized in the beginning that \emph{exact} (i.e., error-free) embezzlement is not possible in a tensor product framework. 
We now return to the question of exact embezzlement in the commuting operator framework. 
It has been shown before \cite{cleve_perfect_2017} that for every \emph{fixed state} $\Psi \in \CC^m\ox \CC^m$ it is possible to construct a quantum state $\Omega\in\H$ in a separable Hilbert space $\H$ and commuting unitaries $u,u':\H\rightarrow \H\ox \CC^m\ox\CC^m$ such that
\begin{align}\label{eq:exact-mbz-2}
    u u'\Omega = \Omega\ox\Psi.
\end{align}
The possibility of exact $\Psi$-embezzlement raises the question of whether exact embezzlement for arbitrary states may be possible in general in the commuting operator framework. We can answer this question definitively:
\begin{introtheorem}[cf.\ Cor.~\ref{cor:exact-mbz}]
    There exists a standard bipartite system $(\H,\M,\M')$ that allows for exact embezzlement in the sense that there exist unitaries $u\in M_m(\M)$, $u'\in M_m(\M')$ such that
    \begin{align}
        uu'(\Omega\ox\Phi) = \Omega\ox \Psi,
    \end{align}
    for any two states $\Phi,\Psi\in\CC^m\ox \CC^m$ with full Schmidt rank and for any $m\in\NN$. 
    However, any such bipartite system requires that $\H$ is non-separable. 
\end{introtheorem}
The requirement that the initial state $\Phi$ has full Schmidt-rank is necessary because one clearly cannot map a non-faithful state on $\M$ to a faithful state on $\M$ via a unitary operation $u$. Alternatively, $\Omega$ can also be used for exact embezzlement in the sense of \eqref{eq:exact-mbz-2}. 
To construct such exactly embezzling bipartite systems, we can take a system with an embezzling state and pass to the ultrapower.
This technique also allows us to show that the spectrum of the modular operator of any embezzling state $\omega$ must be all of $\RR^+$, which immediately implies that universal embezzlers are of type $\III_{1}$. The reason why an exactly embezzling state cannot be realized on a separable Hilbert space is that the modular operator $\Delta_\omega$ of such a state must have every $\lambda>0$ as an eigenvalue.

\subsection*{Quantum fields as universal embezzlers}

The results on ITPFI factors already show that infinite spin systems may serve as universal embezzlers. 
Besides statistical mechanics, the arena of physics where type $\III$ factors appear most naturally is relativistic quantum field theory. From the point of view of operator algebras, a quantum field theory may be viewed as a local net of observable algebras that associates von Neumann algebras $\M(\O)$ to (open) subsets $\mc O$ of spacetime. The algebras all act jointly on a Hilbert space $\H$ with a common cyclic separating vector $\Omega\in\H$ representing the vacuum. The map $\mc O\mapsto \M(\O)$ must, of course, fulfill certain consistency conditions imposed by relativity; see \cref{sec:qft} for an overview and \cite{haag_local_1996} for a thorough introduction. 

We may interpret a unitary operator $u\in \M(\O)$ as a unitary operation that may be enacted by an agent having control over spacetime region $\mc O$. 
Suppose Alice controls $\mc O$ and Bob controls the causal complement $\mc O'$. If Haag duality holds, namely $\M(\O')=\M(\O)'$, we may thus interpret $(\H,\M(\O),\M(\O'))$ as a (standard) bipartite system and ask whether the vacuum state $\Omega$ (or any other state) is an embezzling state. 
According to the results summarized above, this amounts to deciding the type of the algebra $\M(\O)$. 
As we discuss in more detail in \cref{sec:qft}, it has been found under very general assumptions that the local algebras $\M(\O)$ have type $\III$, and, in fact, subtype $\III_1$.
Succinctly:
\begin{center}\it
Relativistic quantum fields are universal embezzlers.
\end{center} 
Besides giving an operational interpretation to the diverging entanglement (fluctuations) in relativistic quantum fields, this result also provides a simple explanation for the classic result that the vacuum of relativistic quantum fields allows for a maximal violation of Bell inequalities \cite{summers1985vacuum}: Alice and Bob can simply embezzle a perfect Bell state $\frac{1}{\sqrt 2}(\ket{11} +\ket{22})$ and subsequently perform a standard Bell test.
Indeed, we can establish a quantitative link between the degree of violation of the CHSH inequality as measured by the correlation coefficient $\beta$ (see \cref{subsec:bell}) and our embezzlement quantifier $\kappa(\omega)$:
\begin{introtheorem}[cf.\ Prop.~\ref{prop:beta}]  Consider a standard bipartite system $(\H,\M,\M')$ with state $\Omega\in \H$ and marginal $\omega\in S_*(\M)$.
Then
\begin{align}
    \beta(\Omega;\M,\M') \geq 2\sqrt 2 - 8\sqrt{\kappa(\omega)}.
\end{align}
\end{introtheorem}
In particular, whenever $\kappa(\omega)<0.01$, we find that Alice and Bob can use $\Omega$ to violate a Bell inequality. 
By \eqref{eq:kappa_max_diameter}, the embezzlement quantifier $\kappa(\omega)$ is bounded by $2(1-\lambda^{1/2})/(1+\lambda^{1/2})$ for states $\omega$ on a type $\III_\lambda$ factor. 
Thus, when $\M$ is a type $\III_\lambda$ factor with $0.99\le \lambda\le1$, every pure bipartite state is guaranteed to violate a Bell inequality.

As a cautious remark, we mention that the operational interpretation via embezzlement needs to be taken with a grain of salt as the status of "local operations" in quantum field theory needs further clarification (see \cite{fewster2020local_measurements} and reference therein). Specifically, it is not clear which unitaries in the local algebras $\M(\O)$ can serve as viable operations localized in $\O$.

Besides the bipartite setting, the monopartite interpretation of the local observables algebras $\M(\O)$ as universal embezzlers reveals that all states of any (locally) coupled quantum system (at least if it is hyperfinite or semifinite) can be locally prepared up to arbitrary precision (cf. \cref{thm:kappa_bound,thm:typeII_mbz}) -- an observation that is in accordance with previous findings concerning the local preparability of states in relativistic quantum field theory \cite{buchholz1986on_noether,werner1987local_prep,summers1990on_independence}.

\section{Conclusion and outlook}
In this work, we have comprehensively discussed the problem of embezzlement of entanglement in the setting of von Neumann algebras. Our results establish a close connection between quantifiers of embezzlement and Connes' classification of type $\III$ factors. 
In particular, we show that embezzling states and even universal embezzlers exist -- both as mathematical objects and in the form of relativistic quantum fields. In the remainder of this section, we make some additional remarks and conclusions. 

An immediate question that comes to mind is the connection between embezzling states, as discussed in this work, with embezzling families. While we plan to present the details in future work \cite{embezzling_families}, we here briefly mention some results that may be obtained in this regard:

We can show that one can interpret the van Dam-Hayden embezzling family $\Omega_{2^k}$ (see \eqref{eq:vandam}) as a family of states on $M_2^{\otimes k}$, i.e., on chains of $k$ spin-$\frac12$ particles, which converge to the unique tracial state on the resulting UHF algebra $M_2^{\ox\oo}$. 
Thus, even though we start with an embezzling family and we can take a well-defined limit, we obtain a type $\II_1$ factor (after closing); hence, the resulting state \emph{cannot} be an embezzling state.  

Conversely, however, if $\M$ is a hyperfinite factor with a dense increasing family $\{\M_k\}_k$ of finite type $\I$ factors, and $\omega\in S_*(\M)$ is a (monopartite) embezzling state, then the restrictions $\omega_k$ to $\M_k$ define a monopartite embezzling family and their purifications $\Omega_k$ yield bipartite embezzling families. Therefore, there are embezzling families that lead to type $\III$ factors. These families can be characterized through a consistency condition. 
Moreover, whenever $(\M,\omega)$ arises as an inductive limit of finite-dimensional matrix algebras, we naturally obtain embezzling families. In particular, this may be related to the construction of quantum field theories via scaling limits \cite{stottmeister2020oar,morinelli2021scaling_limits_wavelets,osborne2023cft_lattice_fermions}. This suggests that embezzling arbitrary states requires operations on arbitrarily small length and large energy scales.
From a practical point of view, the latter is, of course, infeasible. But, it poses the question of to what extent embezzlement could be quantified in terms of the energy densities at one's disposal.
More importantly, it has been hypothesized in various forms that in a quantum theory of gravity, there must exist a minimal length scale \cite{hossenfelder_minimal_2013}. This would seem to break the possibility of embezzlement in our sense. 
Indeed, recently, it was argued that in the presence of gravity, local observable algebras may be of type $\II$ instead of type $\III$ \cite{witten2022crossed,chandrasekaran2023dsobs,longo2022note,chandrasekaran2023adscft}. If true, this would rule out the possibility of having quantum fields as (even non-universal) embezzlers. 
Thus, the absence of physical embezzlers may be a decisive property of quantum gravity. However, as mentioned above, drawing such a conclusion would also require further insights into the structure of admissible local operations in quantum field theory.

An interesting result in quantum information theory is the super-additivity of quantum and classical capacities of certain quantum channels \cite{smith_quantum_2008,hastings_superadditivity_2009}, the latter being equivalent to a range of super-additivity phenomena in quantum information theory \cite{shor_equivalence_2004}. Let us mention that embezzling states show a super-additivity effect, too: It is possible to have two algebras $\M_1$ and $\M_2$ with states $\omega_1$ and $\omega_2$ so that $\kappa(\omega_i) = \diam(S_*(\M_i)/\!\sim)$, i.e., the states $\omega_i$ perform as bad as possible on the respective algebras in terms of embezzlement, but nevertheless $\kappa(\omega_1\ox\omega_2) = 0$, i.e.\ $\omega_1\otimes \omega_2$ is an embezzling state. 
To see this, choose $\M_i$ as type $\III_{\lambda_i}$ ITPFI factors such that $\frac{\log(\lambda_1)}{\log(\lambda_2)}\notin\QQ$. It is well-known that in this case $\M_1\ox\M_2$ has type $\III_1$ and hence $\kappa(\omega_1\ox\omega_2)=0$ \cite{araki1968factors}. 
It is even possible to find type $\III_0$ factors such that their tensor square is a type $\III_1$ factor \cite[Cor.~3.3.5]{connes1973classIII}.

\paragraph{Acknowledgements}
We would like to thank Marius Junge, Roberto Longo, Yoh Tanimoto, and Rainer Verch for useful discussions. We thank Amine Marrakchi for pointing out \cref{lem:amine}  to us, which allowed us to complete \cref{thm:min}. We thank Stefaan Vaes for sharing his insights on Mathoverflow as well as the construction in \cref{rem:stefaan}.
We thank Carsten Voelkmann for carefully reading an earlier version of this manuscript.
LvL and AS have been funded by the MWK Lower Saxony via the Stay Inspired Program (Grant ID: 15-76251-2-Stay-9/22-16583/2022).

\subsection*{Notation and standing conventions} 

Inner products are linear in the second entry.
The standard basis vectors of $\CC^n$ are denoted $\ket i$, $i=1,...,n$.
The product basis in $\CC^n\ox\CC^n$ will be written as $\ket{ij}=\ket i \ox\ket j$.
Positive cones of ordered vector spaces $(E,\ge)$ will be denoted $E^+$. In particular, $\RR^+=[0,\infty)$.
The unitary group of a von Neumann algebra $\M$ is denoted $\U(\M)$, the normal state space is denoted $S_*(\M)$, and the center is denoted $Z(\M)$.
The support projection of a normal state $\omega$ on a von Neumann algebra $\M$ is denoted $s(\omega)$.
The set of projections in $\M$ is denoted $\proj(\M)$.
If $h$ is a self-adjoint operator, we define the (possibly unbounded) operator $h^{-1}$ as the pseudoinverse.\footnote{Explicitly, if $p$ is the spectral measure of $h$, we define $h^{-1} = \int \lambda^+ dp(\lambda)$ where $x^+ = 1/x$ for $x>0$ and $0^+=0$.}
If $\M\subset\B(\H)$ and $\Omega\in\H$ is a vector, then $[\M\Omega]$ denotes the closure of the subspace spanned by the vectors $\M\Omega$ or the orthogonal projection onto this subspace, depending on context. 
If $\M$ acts on $\H$, matrices $[x_{ij}]\in M_{n,m}(\M)$ are identified with operators $x:\H\ox\CC^m\to\H\ox\CC^n$ via $x\Psi\ox\ket j = \sum_i x_{ij}\Psi\ox\ket i$.
If $x\in M_{n,m}(\M)$ and $\omega\in S_*(M_m(\M))$ then $x\omega x^*$ denotes the normal positive linear functional defined by $(x\omega x^*)(y) = \omega(x^*yx)$, $y\in M_m(\M)$.
We denote the von Neumann tensor product of two von Neumann algebras $\M,\N$ by $\M\ox\N$.

\section{Preliminaries}\label{sec:preliminaries}

We briefly recall the basics of von Neumann algebras and give an overview of modular theory, crossed products, and spectral scales (see, for example, \cite{bratteli1987oa1,takesaki1,takesaki2,takesaki3} for further details).
We hope that this makes our work more accessible to readers from the quantum information community.

\subsection{Hilbert spaces, von Neumann algebras, and normal states}
\label{subsec:pre:hilbert-spaces}
All Hilbert spaces are assumed complex, and we use the convention that the inner product $\ip\placeholder\placeholder$ is linear in the second entry.
If $\H$ is a Hilbert space, $\mc T(\H)$ denotes the trace class and $\B(\H)$ denotes the algebra of bounded operators on $\H$. $\B(\H)$ is equipped with the operator norm, the obvious product, and the adjoint operation $x\mapsto x^*$.
Apart from the norm topology, it also carries several operator topologies.
The weak and strong operator topologies are the topologies generated by the families of functions $\{x\mapsto\ip\Psi{x\Phi}:\Psi,\Phi\in\H\}$ and $\{x\mapsto \norm{x\Psi}:\Psi\in\H\}$, respectively.
As a Banach space, $\B(\H)$ is isomorphic with the dual space of the trace class via the pairing $(x,\rho)\mapsto \tr \rho x$, where $(x,\rho)\in\B(\H)\times\mc T(\H)$.
The weak$^*$ topology induced on $\B(\H)$ by the duality with $\mc T(\H)$ is called the $\sigma$-weak operator topology.

If $\M$ is a collection of bounded operators, its commutant, denoted $\M'$, is the subalgebra of bounded operators $x'\in\B(\H)$ commuting with all $x\in\M$.
A von Neumann algebra $\M$ on $\H$ is a weakly closed non-degenerate $^*$-invariant algebra of bounded operators.
It is actually equivalent to ask for $\M$ to be closed in the strong or $\sigma$-weak topology or to ask that $\M$ is equal to its bicommutant $\M''$.
This equivalence is the celebrated bicommutant theorem of von Neumann and lies at the heart of the theory.
It implies that von Neumann algebras always come in pairs: If $\M$ is a von Neumann algebra, then so is its commutant $\M'$.
If $\M$ is a von Neumann algebra on $\H$, then it is the Banach space dual of the space $\M_*$ of $\sigma$-weakly continuous linear functionals on $\M$.
Consequently, $\M_*$ is called the \emph{predual} of $\M$. It isometrically embeds into the dual $\M^*$ and bounded linear functionals $\varphi\in \M^*$ are called normal if they are $\sigma$-weakly continuous, i.e., if they are in $\M_*$.

As was famously shown by Sakai, von Neumann algebras can also be defined abstractly as those $C^*$-algebras $\M$ that have a predual $\M_*$.
We will mostly work with abstract von Neumann algebras from which the concrete von Neumann algebras arise via representations.
In the abstract setting, the weak, strong, and $\sigma$-weak operator topologies cannot be defined on $\M$ as above. 
However, the $\sigma$-weak topology does not depend on the representation: It is the topology induced by the predual. In the abstract setting, we will refer to it as \emph{ultraweak topology}.
A $^*$-homomorphism $\pi:\M\to\N$ between von Neumann algebras is called \emph{normal} if it is ultraweakly continuous, i.e., continuous if both $\M$ and $\N$ are equipped with the respective ultraweak topologies.
In particular, a normal representation of a von Neumann algebra $\M$ on a Hilbert space $\H$ is a unital $^*$-homomorphism $\pi:\M\to\B(\H)$ which is continuous with respect to the ultraweak topology on $\M$ and the $\sigma$-weak operator topology on $\B(\H)$.
In this work, we only consider faithful representation, which we usually just write as $\M\subset\B(\H)$.

A \emph{normal state} on $\M$ is an ultraweakly continuous positive linear functional $\omega:\M\to\CC$, such that $\omega(1)=1$.
We denote the set of normal states by $S_*(\M)$.
If $\M\subset\B(\H)$ and if $\rho$ is a density operator on $\H$, i.e., a positive trace-class operator with $\tr\rho=1$, then $\omega(x)=\tr\rho x$ defines a normal state on $\M$ and all normal states arise in this way.
If $\omega$ is a normal state on $\M$, then its \emph{support projection} $s(\omega)$ is defined as the smallest projection $p\in\M$ such that $\omega(p)=1$.\footnote{Equipped with the usual order and $p^\perp=1-p$, the projections in a von Neumann algebra $\M$ form an orthocomplete lattice $\proj(\M)$: For every family of projections $p_j$ in $\M$ there exists a least upper bound $\bigvee_j p_j$ and a greatest lower bound $\bigwedge_j p_j$, both of which are projections in $\M$, such that $1-\bigvee_j p_j = \bigwedge_j (p_j-1)$.}
If $\omega(x)=\ip\Omega{x\Omega}$ is a vector state, then $s(\omega)$ is the orthogonal projection onto $[\M'\Omega]$, the closure of $\M'\Omega$.
A normal state $\omega$ is faithful if $\omega(x^{*}x)=0$ implies $x=0$, and in this case the support projection is $s(\omega)=1$.

In physics, one often considers separable Hilbert spaces only.
Consequently, von Neumann algebras describing observables of a physical system should admit a faithful representation on a separable Hilbert space.
Such von Neumann algebras are called \emph{separable} and may be characterized by the following equivalent properties:
\begin{enumerate}
    \item $\M$ admits a faithful representation on a separable Hilbert space
    \item the predual $\M_*$ is norm-separable
    \item $\M$ is separable in the ultraweak topology.
\end{enumerate}
In particular, a separable von Neumann algebra only admits countable families of pairwise orthogonal non-zero projections.
The latter property is called \emph{$\sigma$-finiteness} (or countable decomposability), and it is equivalent to the existence of a faithful normal state \cite[Prop.~2.5.6]{bratteli1987oa1}.

If $n\in\NN$ is an integer and $\M$ is a von Neumann algebra on $\H$, then the $n\times m$ matrices $[x_{ij}]\in M_{n,m}(\M)$ with entries $x_{ij}$ in $\M$ are identified with operators 
\begin{equation}
    x: \H\ox\CC^m \to \H\ox\CC^n, \qquad x(\Psi\ox\ket j) = \sum_{i=1}^{n} x_{ij} \Psi \ox \ket i,
\end{equation}
where $\{\ket i\}_{i=1}^k$ denotes the standard basis of $\CC^k$.
In particular, $M_n(\M)\cong \M\ox M_n$ is itself a von Neumann algebra on $\H\ox\CC^n$.

\subsection{Weights and modular theory}
\label{subsec:pre:weights}
Weights can be viewed as a non-commutative analog of integration with respect to a not-necessarily finite measure in the same sense as normal states are non-commutative analogs of integration with respect to a probability measure.
A \emph{normal weight} $\varphi$ on a von Neumann algebra $\M$ is an ultraweakly lower semicontinuous map $\varphi:\M^+\to[0,\oo]$ satisfying
\begin{equation}
    \varphi(x+\lambda y) = \varphi(x)+\lambda\varphi(y), \qquad x,y\in\M^+,\ \lambda\ge0,
\end{equation}
with the convention $0\cdot\oo=0$ (see, \cite[Ch.~VII, §1]{takesaki2}, \cite[Sec.~III.2]{blackadar_operator_2006}).
For normal weights, there is a non-commutative analog of the monotone convergence theorem: If $(x_\alpha)$ is a uniformly bounded increasing net in $\M^+$ then $\lim_\alpha \varphi(x_\alpha) = \varphi(\lim_\alpha x_\alpha)$ where the limit $\lim_\alpha x_\alpha$ is in the ultraweak topology.
A normal weight is said to be semifinite if the left-ideal $\mathfrak n_\varphi=
\{x\in\M : \varphi(x^*x)<\oo\}$ is ultraweakly dense in $\M$, and it is said to be faithful if $\varphi(x^*x)=0$ implies $x=0$.
In the following, all weights are assumed to be semifinite and we will sometimes just write "weight" instead of normal semifinite weight.
A (normal semifinite) \emph{trace} is a weight $\tau$ that is unitarily invariant, i.e., satisfies $\tau(uxu^*)=\tau(x)$ for all $x\in\M^+$ and unitaries $u\in\U(\M)$.
A particularly easy class of weights are normal positive linear functionals on $\M$. These are precisely the normal weights such that $\varphi(1)<\oo$. Equivalently, these are the finite weights with $\mathfrak n_\varphi=\M$.

The GNS construction of a normal state $\varphi$ on $\M^+$ can be generalized to normal semifinite weights $\varphi$.
For each a normal semifinite weight $\varphi$ there is -- up to unitary equivalence -- a unique semi-cyclic representation $(\pi_\varphi,\H_\varphi,\Psi_\varphi)$
where $\pi_\varphi$ is a normal representation of $\M$ on $\H_\varphi$ and $\Psi_\varphi:\mathfrak n_\varphi\to\H_\varphi$ is a linear map such that
\begin{equation}
\left.\begin{aligned}
    \pi_\varphi(a)\Psi_\varphi(x) &= \Psi_\varphi(ax)\\[5pt]
    \ip{\Psi_\varphi(x)}{\Psi_\varphi(y)} &= \varphi(x^*y)
\end{aligned}\ \right\} \qquad a\in\M,\ x,y\in\mathfrak n_\varphi,
\end{equation}
where the right-hand side is defined by polarization \cite[Sec.~VII, §1]{takesaki2}.
Semi-cyclicity means that $\Psi_\varphi(\mathfrak n_\varphi)$ is dense in $\H_\varphi$.
If the weight $\varphi$ is faithful, then so is the GNS representation.

We are now going to describe the basics of modular theory.
For this, we pick a normal semifinite faithful weight $\varphi$ and consider its GNS representation $(\H_\varphi,\pi_\varphi,\Psi_\varphi)$.
Since $\varphi$ is faithful the same holds for $\pi_\varphi$, so that we may identify $\M$ with $\pi_\varphi(\M)$.
The starting point of modular theory ist the conjugate-linear operator $S_\varphi^0\Psi_\varphi(x) = \Psi_\varphi(x^*)$ defined on all vectors of the form $\Psi_\varphi(x)$ with $x\in\mathfrak n_\varphi\cap\mathfrak n_\varphi^*$.
It can be shown that $S_\varphi^0$ is closable.
The modular operator $\Delta_\varphi$ and the modular conjugation $J$ induced by the weight $\varphi$ are defined by
\begin{equation}
    \Delta_\varphi = S_\varphi^*S_\varphi \qandq S_\varphi = J\Delta_\varphi^{1/2},
\end{equation}
where $S_\varphi$ is the closure of $S_\varphi^0$.
The \emph{modular flow} of $\varphi$ is the ultraweakly continuous one-parameter group of automorphisms 
\begin{equation}
    \sigma_t^\varphi(x) := \Delta_\varphi^{it}x\Delta_\varphi^{-it}
\end{equation}
on $\B(\H_\varphi)$.  The main theorem of modular theory, due to Tomita, states that $J\M J=\M'$ and that the modular flow leaves $\M$ invariant, i.e., $\sigma_t^\varphi(x)\in\M$ for all $t\in\RR$, $x\in\M$, see \cite[Ch.~VI]{takesaki2}.
The first fact implies that $\M$ and $\M'$ are anti-isomorphic via the conjugate linear $^*$-isomorphism $j:\M\to\M'$ defined by $j(x)=JxJ$. On the center $Z(\M)=\M\cap\M'$, $j$ simply reduces to the adjoint $j(z) = JzJ = z^*$, $z\in Z(\M)$.
An additional structure that is present in the GNS representation is the positive cone 
\begin{equation}
    \P = \overline{\{j(x)\Psi_\varphi(x) : x\in\M\}}\subset \H_\varphi.
\end{equation}
One can show that $\P$ linearly spans $\H_\varphi$, is pointwise invariant under $J$, and self-dual in the sense that $\P=\P^\natural := \{\Psi\in\H_\varphi : \ip{\Psi}{\Phi}\ge0 \ \forall \Phi\in\P\}$.
Before we explain the importance of the cone $\P$, we mention that, up to unitary equivalence, the triple $(\H_\varphi,J,\P)$ does not depend on the choice of normal semifinite faithful weight. 
In fact, all triples $(\H,J,\P)$ of a Hilbert space $\H$ equipped with a faithful representation $\M\subset\B(\H)$, a conjugation\footnote{A conjugation on a Hilbert space $\H$ is a conjugate linear isometry satisfying $J^2=1_\H$.} $J$ and a self-dual closed cone $\P\subset\H$ satisfying
\begin{enumerate}
    \item\label{it:axiom1} $J\M J=\M'$,
    \item\label{it:axiom2} $J\Psi=\Psi$ for all $\Psi\in \P$,
    \item\label{it:axiom3} $xj(x)\P\subset\P$ for all $x\in\M$, where $j(x) = JxJ\in\M'$,
\end{enumerate}
are unitarily equivalent \cite{haagerup_standard_1975}.\footnote{In \cite{haagerup_standard_1975}, where the standard form was introduced by Haagerup, it is additionally assumed that $JzJ=z^*$ for all $z\in\M\cap\M'$. This assumption is shown to be redundant in \cite[Lem.~3.19]{ando2014ultraproducts}.}
Such a triple $(\H,J,\P)$ is called the \emph{standard form} of $\M$, and we saw above that the GNS construction of a normal semifinite faithful weight gives a way to construct the standard form.\footnote{The uniqueness of the standard representation only applies to the modular conjugation $J$ and the positive cone $\P$. The modular operator $\Delta_\varphi$ does depend on the chosen weight $\varphi$.}
The standard form is sometimes called standard representation. We will use the term \emph{standard representation} to mean a representation that is spatially isomorphic to the standard form.
Roughly speaking, a standard representation is the standard form where we forget about $J$ and $\P$.

The importance of the positive cone $\P$ is due to the following fact:
For every $\omega\in\M_*^+$ there is a \emph{unique} $\Omega_\omega\in\P$ such that
\begin{equation}\label{eq:unique_gns}
    \omega(x) = \ip{\Omega_\omega}{x\Omega_\omega},\qquad x\in\M.
\end{equation}
The map $S_*(\M)\ni\omega\mapsto\Omega_\omega\in\P$ is a homeomorphism. In fact, the following estimates hold
\begin{equation}\label{eq:state_vector_norms}
    \norm{\Omega_\psi - \Omega_\phi}^2 \le \norm{\psi-\phi} \le \norm{\Omega_\psi-\Omega_\phi}\,\norm{\Omega_\psi+\Omega_\phi}.
\end{equation}
Furthermore, if $\M$ is in standard form, then so is $\M'$ and $\M$ and $\M'$ are anti-isomorphic via the map $j(x)=JxJ$.
Using $j$ and \cref{it:axiom3} above, it follows that
\begin{equation}\label{eq:purification_formula}
    \Omega_{v \omega v^*} = vj(v) \Omega_\omega, \qquad v\in \M,
\end{equation}
where $v\omega v^* = \omega(v^*(\placeholder)v)$, and $S_{v\omega v^*} = v j(v) S_\omega v^* j(v)^*$. Consequently also $\Delta_{v\omega v^*} = v j(v) \Delta_\omega v^* j(v)^*$.
Given a standard representation $\M\subset\B(\H)$ and a cyclic separating vector $\Omega$, there uniquely exist a positive cone $\P$ and a conjugation $J$ turning $(\H,J,\P)$ into a standard form, such that $\Omega\in\P$ where $\omega(x)=\ip{\Omega}{x\Omega}$, $x\in\M$.
Since normal faithful states exist precisely on $\sigma$-finite von Neumann algebras, a $\sigma$-finite von Neumann algebra is in standard representation if and only if there exists a cyclic separating vector.

\begin{example}\label{exa:M_n_std_form}
    In the case $\M=M_n$, the standard form can be described as follows:
    $\H = \CC^n\ox\CC^n$ and the action of $M_n$ is given by identifiyng $x\in M_n$ with $x\ox 1\in M_n\ox1 \subset\B(\H)$.
    Following standard notation, we write $\ket{ij}$ for $\ket i \ox \ket j$, $i,j=1,...,n$.
    The conjugation is $J(\sum_{ij}\Psi_{ij}\ket{ij}) = \sum_{ij}\Psi^*_{ij}\ket{ij}$, where $\Psi^*_{ij}=\overline{\Psi_{ji}}$.
    The commutant of $\M=M_n$ is $1\ox M_n$ and $j$ is simply given by $j(x\ox 1)=1\ox \bar x$ where $\bar x$ is the entry-wise complex conjugate of $x$, i.e., $\overline{[x_{ij}]}=[\bar x_{ij}]$.
    The cone $\P$ is
    \begin{equation}\label{eq:M_n_positive_cone}
        \P 
        = \bigg\{\sum_{ij} \Psi_{ij} \ket{ij}: [\Psi_{ij}]\ge0 \bigg\} = \bigg\{(\Psi\otimes 1) \Omega: 0\leq \Psi \in M_n\bigg\},
    \end{equation}
    with $\Omega = \sum_i \ket{ii}$ and the map $\omega\mapsto\Omega_\omega$ is given by
    \begin{equation}\label{eq:cone_rep_state}
        \Omega_\omega = \sum_{ij} (\sqrt{\rho_\omega})_{ij} \ket{ij} =(\sqrt{\rho_\omega}\otimes 1)\Omega,
    \end{equation}
    where $\rho_\omega\in M_n^+$ is the density operator of $\omega$, i.e., satisfies $\omega(x) = \tr \rho_\omega x$. The modular operator of a faithful state $\omega$ takes the form $\Delta_\omega = \rho_\omega\ox(\overline \rho_\omega)^{-1}$, which indeed fulfills
    \begin{align}
    J\Delta_\omega^{1/2} (x\otimes 1) \Omega_\omega = J(\rho_\omega^{1/2}x \rho_\omega^{1/2}\ox\overline \rho_\omega^{-1/2}) \Omega 
    = (\rho_\omega^{-1/2} \ox \overline\rho_\omega^{1/2} \overline{x} \:\!\overline\rho_\omega^{1/2})\Omega = (x^* \otimes 1)\Omega_\omega,
    \end{align}
    where we used $(x^*\otimes 1)\Omega = (1\otimes \overline{x})\Omega$ and \cref{eq:cone_rep_state}.
\end{example}

More generally, the standard form of the matrix amplification $M_n(\M)$ of a von Neumann algebra is given by:

\begin{lemma}\label{lem:standard_amplification}
    If $(\H,J,\P)$ is the standard form of $\M$, we can construct the standard form $(\H\up n,J\up n,\allowbreak\P\up n)$ of $M_n(\M)$ as follows. 
    The Hilbert space is $\H\up n = \H\ox\CC^n\ox\CC^n$ with $M_n(\M)$ acting on the first two tensor factors in the obvious way, the modular conjugation is given by $J\up n (\sum_{ij} \Psi_{ij}\ox\ket{ij}) = \sum_{ij} J\Psi_{ji} \ox \ket{ij}$, and the positive cone is 
    \begin{align}\label{eq:Pn}
        \P\up n= \overline\lin\bigg\{\sum_{kl}x_kj(x_l)\Psi \ox \ket{kl} : x_1,...,x_n\in\M,\ \Psi\in\P\bigg\}.
    \end{align}
    If $\M$ is $\sigma$-finite, we have 
    $\P\up n = \overline\lin\{\sum_{kl} x_kj(x_l)\Omega\ox\ket{kl} : x_1,...x_n\in\M\}$ for some cyclic separating vector $\Omega\in\P$.
\end{lemma}
\begin{proof}
    We construct the standard form using the GNS representation of a normal semifinite faithful weight.
    Let $\varphi$ be a normal semifinite faithful weight on $\M$ and set $\varphi\up n= \varphi\ox\tr$. 
    Let $(\H,J,\P$), the standard form of $\M$, be constructed from the GNS representation of $\varphi$.
    The GNS representation of $\varphi\up n$ is $(\H\up n,\pi\up n, \eta\up n) = (\H,\pi,\eta)\ox(\CC^n\ox\CC^n,\id \ox1,(\id\ox1)\Omega)$ with $\Omega=\sum_r \ket{rr}$.
    Note that $(x\ox 1)\Omega = \sum_{kl} x_{kl}\ox\ket{kl}$.
    We now supress $\pi$ and $\pi\up n$, i.e., identify $\pi=\id_\M$ and $\id_\M\ox \id_{M_n}\ox1$.
    Since $J\up n$ is the tensor product of the modular conjugations on $\H$ and $\CC^n\ox\CC^n$, we have $J\up n(\sum_{kl}\Psi_{kl}\ox\ket{kl}) = \sum_{kl} (J\Psi_{kl})\ox \ket{kl}$.
    Therefore, $j\up n(x)$, $x\in M_n(\M)$, is given by
    $
        j\up n(x) = \sum_{kl} j(x_{kl}) \ox 1 \ox \ketbra kl \in M_n(\M)'
    $,
    and the positive cone $\P\up n$ is the closure of 
    \begin{align*}
         \bigg\{ j\up n(x) \eta\up n(x) : x\in M_n(\A_{\varphi\up n})\bigg\}
         &= \bigg\{ \sum_{klmn} (j(x_{mn}) \ox 1\ox \ketbra mn)(\eta_\varphi(x_{kl}) \ox \ket{kl}) : x\in M_n(\A_\varphi)\bigg\}\\
        &= \bigg\{ \sum_{klmn} j(x_{mn})\eta_\varphi(x_{kl}) \ox (\ket k\ox \ket m \braket nl ) : x\in M_n(\A_\varphi)\bigg\}\\
         &= \bigg\{ \sum_{klm} j(x_{ml})\eta_\varphi(x_{kl}) \ox \ket{km} : x\in M_n(\A_\varphi)\bigg\}.\\
         &={\lin}\bigg\{\sum_{km} j(x_m)\eta_\varphi(x_k) \ox\ket{km} : x_1,...,x_n\in\A_\varphi\bigg\}.
    \end{align*}
    where we used the notation $\A_\varphi := \mathfrak n_\varphi\cap\mathfrak n_\varphi^*$. Note that $\A_\varphi''=\M$.
    We note that, if $\varphi=\omega$ is a faithful normal state, $\eta_\varphi(x)= x\Omega_\omega$, $x\in\M=\A_\omega$, so that the last claim follows.
    Because of the closure, we may replace $\Psi\in\P$ in \eqref{eq:Pn} by $\Psi = j(x_0)\eta_\varphi(x_0)$, $x_0\in \A_\varphi$. Taking $y_k = x_k x_0$, $x_k\in\M$, then gives
    \begin{equation*}
        \P\up n\ni j(y_l)\eta_\varphi(y_k)\ox\ket{kl} = j(x_l) x_k j(x_0)\eta_\varphi(x_0) \ox \ket{kl} = j(x_l)x_k \Psi \ox \ket{kl}.
    \end{equation*}
    Thus, $\P\up n$ contains the right-hand side of \eqref{eq:Pn}.
    For the converse, pick a sequence $e_n\in\A_\varphi$ which converges to $1$ strongly (this is possible because $\A_\varphi''=\M$ \cite[Thm.~VII.2.6]{takesaki2}).
    Then $j(x_k)\eta_\varphi(x_l) = \lim_n j(x_ke_n)\eta_{\varphi}(x_le_n) = j(x_k)x_k j(e_n)\eta_\varphi(e_n) = j(x_k)x_l\Psi_n$, $\Psi_n = j(e_n)\eta_\varphi(e_n)\in\P$. Thus, the set in \eqref{eq:Pn} contains $\P\up n$, which finishes the proof.
\end{proof}

\subsection{Crossed products}
\label{subsec:pre:crossed-products}

Crossed products play an important role in the theory of von Neumann algebras.
Given a group action on a von Neumann algebra, the crossed product provides a way to extend the von Neumann algebra by the generators of the group action (see \cite{vandaele1978crossed_products_typeIII, takesaki2} for general accounts).

In the following, we fix a locally compact abelian group $G$ and denote by $\hat G$ its Pontrjagin dual\footnote{The Pontjagin dual $\hat G$ of a locally compact abelian group $G$ is the group of characters of $G$, i.e., continuous homomorphisms $\chi:G\to\mathbb T$ onto the circle group $\mathbb T$, with pointwise multiplication. It is locally compact as well.}. In the rest of the paper, we only need the cases $G=\ZZ, \RR, \mathbb T$ whose duals are $\hat G = \mathbb T, \RR, \ZZ$, respectively.

Let $\M$ be a von Neumann algebra of operators on $\H$ equipped with a point-ultraweak continuous $G$-action $\alpha$.\footnote{A $G$-action $\alpha:G\to\Aut\M$ is point-ultraweak continuous if for all normal states $\omega$ and all $x\in\M$, the map $g\mapsto \omega(\alpha_g(x))$ is a continuous function $G\to\CC$. See  \cite[Thm.~III.3.2.2]{blackadar_operator_2006} for equivalent characterizations.}
To construct the crossed product $\M\rtimes_\alpha G$, consider the Hilbert space $L^2(G,\H)=\H\ox L^2(G,dg)$, where $dg$ is the left Haar measure, and define operators
\begin{equation}\label{eq:cp_ops}
    \left.\begin{aligned}
    \pi(x)\Psi(h) &= \alpha_{h^{-1}}(x)\Psi(h) \\[5pt]
    \lambda(g) \Psi(h) &=\Psi(g^{-1}h)
    \end{aligned}\ \right\}\qquad x\in\M,\ g,h\in G.
\end{equation}
Note that $\lambda(g)\pi(x)\lambda(g)^* = \pi(\alpha_g(x))$, $g\in G$, $x\in\M$.
The crossed product is defined as the von Neumann algebra generated by the operators in \cref{eq:cp_ops}:
\begin{equation}
    \M\rtimes_\alpha G := \{ \pi(x),\lambda(g) : x\in\M, \,g\in G\}''.
\end{equation}
The Hilbert space $L^2(G,\H)$ carries a natural representation $\mu$ of the dual group $\hat G$ given by $\mu(\chi)\Psi (h) = \chi(h)\cdot\Psi(h)$.
These unitaries induce a point-ultraweakly continuous $\hat G$-action $\hat\alpha$ on the crossed product via $\hat\alpha_\chi(y) = \mu(\chi)y\mu(\chi)^*$, $y\in\M\rtimes_\alpha G$.
It follows from the canonical commutation relations
\begin{equation}
    \mu(\chi)\lambda(g) = \overline{\chi(g)}\lambda(g)\mu(\chi),\qquad (g,\chi)\in G\times \hat G,
\end{equation}
that the dual action is given by
\begin{equation}
    \left.\begin{aligned}
    \hat\alpha_\chi(\pi(x)) &=\pi(x) \\[5pt]
    \hat\alpha_\chi(\lambda(g)) &= \overline{\chi(g)} \lambda(g)
    \end{aligned}\ \right\}\qquad x\in\M,\ g\in G,\ \chi\in\hat G.
\end{equation}
Clearly, the map $\pi:\M\to\M\rtimes_\alpha G$ is a normal $^*$-embedding. 
In fact, $\pi(\M)\subset \M\rtimes_\alpha G$ is exactly the fixed point algebra $\N^{\hat\alpha}$.
Using the latter fact, one can associate to every weight $\varphi$ on $\M$ a so-called \emph{dual weight} $\tilde\varphi$ on $\M\rtimes_\alpha G$ via
\begin{equation}
    \tilde\varphi(y) = \varphi\bigg(\int_{\hat G} \hat\alpha_\chi(y)\,d\chi \bigg),\qquad y\in (\M\rtimes_\alpha G)^+,
\end{equation}
where $d\chi$ is the left Haar measure on $\hat G$. 
Unless $\hat G$ is compact (which is equivalent to $G$ being discrete), the dual weight will always be unbounded, i.e., $\tilde\varphi(1)=\oo$, no matter if $\varphi$ is bounded or not.
The modular flow of the dual weight $\tilde\varphi$ is an extension of the modular flow $\sigma^\varphi$ to the crossed product:
\begin{equation}
    \left.\begin{aligned}
    \sigma^{\tilde\varphi}_t(\pi(x)) &=\pi(\sigma_t^\varphi(x)) \\[5pt]
    \sigma^{\tilde\varphi}_t(\lambda(g)) &= \lambda(g)
    \end{aligned}\ \right\}\qquad x\in\M,\ g\in G,\ t\in\RR.
\end{equation}

\begin{example}\label{ex:groupalgebra}
    In the case $\M=\CC$ we have  $L^2(G,\H)=L^2(G,dg)$,
    \begin{align}
        \CC\rtimes_\mathrm{\id} G =\{\lambda(g): g\in G\}''\cong L^\infty(\hat G,d\chi),
    \end{align}
    and the dual action acts on $L^\infty(\hat G,d\chi)$ by translation. Here, $d\chi$ denotes the left Haar measure on $\hat G$. We sketch the argument:
    For abelian groups, the norm completion of basic elements of the form $\int_G f(g) \lambda(g)dg$ with $f\in C_c(G)$ yields the group $C^*$-algebra $C^*(G)$ and we have $C^*(G) \cong C_0(\hat G)$ by the Fourier transform (see, e.g., \cite{williams_crossed_2007}). Since the Fourier transform converts multiplication by characters to translation, and since
    $C_0(\hat G)$ is weak$^*$ dense in $L^\infty(\hat G,d\chi)$ the claim follows. 
\end{example}

\subsection{Spectral scales}\label{subsec:pre:spectral_scales}

Let $\M$ be a semifinite von Neumann algebra and let $\tr$ be a faithful, normal, semifinite trace.
As is common,  we denote by $L^p(\M,\tr)$ (with $0<p<\infty$) the set of densely defined closed operators $\rho$ affiliated with $\M$ such that $(\tr|\rho|^p)^{1/p}<\infty$.
To every normal, positive, linear functional $\omega\in\M_*^+$, we can associate the Radon-Nikodym derivative $\rho_\omega:=d\omega/d\!\tr$, which is the unique positive self-adjoint operator in $L^1(\M,\tr)$ such that
\begin{equation}
    \omega(x) = \tr \rho_\omega x, \qquad x\in\M.
\end{equation}
$\omega$ is a state if and only if $\tr\rho_\omega=1$. 
In the following, we apply the theory of distribution functions and spectral scales in \cite{fack1986generalized,petz1985scale,hiai1991closed,hiai1989distance} to the density operator $\rho_\omega\in L^1(\M,\tr)$, and summarize the basic facts.
The \emph{distribution function} $D_\omega$ of $\omega$ is defined by
\begin{equation}\label{eq:distribution_function}
    D_\omega(t) = \tr \chi_t(\rho_\omega) = \Tr(p_{\omega}((t,\infty))),\qquad t\ge0,
\end{equation}
where $\chi_t$ is the indicator function of $(t,\oo)$ and $p_{\omega}$ is the spectral measure of $\rho_{\omega}$. The spectral scale $\lambda_\omega$ of $\omega$ is defined as
\begin{equation}\label{eq:spectral_scale}
    \lambda_\omega(t) = \inf\{s >0 : D_\omega(s)\le t\}, \qquad t\ge0,
\end{equation}
where $\lambda_\omega(0):=\oo$ if $D_\omega(s)>0$ for all $s>0$.\footnote{The definitions here are related to those in \cite{fack1986generalized} via the density $\rho_\omega\in L^1(\M,\tr)$, e.g., the distribution function $D_{\omega}$ of the state $\omega$ is the distribution function of the positive (and "$\tr$-measurable") operator $\rho_\omega$.}
Both, $D_\omega$ and $\lambda_{\omega}$ are right-continuous, non-increasing probability densities on $\RR^+$:
\begin{align}
    \int_{0}^{\infty}D_{\omega}(t) \,dt = 1, \qquad \int_{0}^{\infty}\lambda_{\omega}(t) \,dt = 1.
\end{align}
Geometrically, \eqref{eq:spectral_scale} means that the graph of $\lambda_\omega$ is the (right-continuous) reflection of the graph of $D_\omega$ about the diagonal.
The distribution function enjoys the following properties:
\begin{align}\label{eq:D_properties}
    D_\omega(t) &\in \{\tr p : p\in \proj(\M)\}, & D_\omega(0)&=\tr s(\omega),& \supp D_\omega&=[0,\norm{\rho_\omega}],
\intertext{where the interval on the right is $[0,\oo)$ if $\rho_\omega$ is unbounded.
Similarly, the spectral scale satisfies:}
    \lambda_\omega(t) &\in \Sp(\rho_\omega),& \lambda_\omega(0) &= \norm{\rho_\omega}, &\supp\lambda_\omega &= [0,\tr s(\omega)].
\end{align}
The spectral scale and the distribution function are connected by
\begin{equation}\label{eq:D_of_lambda}
    D_\omega(t) = |\{ s >0 : \lambda_\omega(s)> t \}| = \int_0^\oo \chi_t(\lambda_\omega(s)) \,ds,
\end{equation}
which may be summarized as saying that the cumulative distribution function of the spectral scale is again $D_\omega$.
Note that the left-hand side of \eqref{eq:D_of_lambda} is, by definition, equal to $\tr\chi_t(\rho_\omega)$.
Since the sets $(t,\oo)$, $t>0$, generate the Borel $\sigma$-algebra on $\RR^+$, it follows that $\tr\chi_A(\rho_\omega)= \int_0^\oo \chi_A(\lambda_\omega(s))\,ds$ for all Borel sets $A\subset \RR^+$.
Therefore, the measure $A\mapsto \tr\chi_A(\rho_\omega)$ is the push-forward of the Lebesgue measure along the spectral scale and
\begin{equation}\label{eq:trace_formula}
    \tr f(\rho_\omega) = \int_{0}^{\infty}f(\lambda)\,\Tr dp_{\omega}(\lambda) = \int_0^\oo f(\lambda_\omega(t))\,dt
\end{equation}
holds for all bounded Borel functions $f$ on $\RR^+$. 
\cref{eq:trace_formula} summarizes many important properties of the spectral scale. It was first observed in \cite[Prop.~1]{petz1985scale}.

\begin{example}\label{exa:M_n_spectrum}
    Let $\M= M_n$. Let $\omega\in S(M_n)$ be a state with density operator $\rho = \sum_i p_i P_i$. Let the eigenvalues be ordered decreasingly $p_1\ge p_2\ge ...$.
    Then, the distribution function and the spectral scale are 
    \begin{equation}
      D_{\omega}(t) = \sum_{i}m_{i}\,\chi_{t}(p_{i}), \qquad  \lambda_\omega(t) = \sum_i p_i \,\chi_{[d_{i-1},d_{i})}(t),
    \end{equation}
    where $d_i =m_i + ... +m_1$ with $m_i =\tr P_i$ being the multiplicity of $p_i$ and $d_{0}=0$. In particular, we have:
    \begin{align}\label{eq:D_l_examples}
        D_{\bra1\placeholder\ket1}(t) = \lambda_{\bra1\placeholder\ket1}(t) = \chi_{[0,1)}(t), \qandq
        \lambda_{\frac1n\!\Tr}(t) = \tfrac1n \,\chi_{[0,n)}(t),\ \ D_{\frac1n\!\tr}(t) = n \,\chi_{[0,\frac1n)}(t).
    \end{align} 
    See \cref{fig:lambda} for a plot of the spectral scale and the distribution function of a state on a matrix algebra.
\end{example}

\begin{example}
        Let $\M=L^\oo(Y,\nu)$ and let $p \in L^1(Y,\nu)$ be a probability density (so that $p\in S_*(\M)$).
        Then, the distribution function $D_{p}(t) = \int_{Y}\chi_{t}(p)d\nu$ is the cumulative distribution function of $p$, and the spectral scale $\lambda_p$ is precisely the decreasing rearrangement $p^* \in L^1(\RR^+)$ of $p$ \cite[Rem.~2.3.1]{fack1986generalized}. 
\end{example}

\begin{figure}
\centering
    \includegraphics[width=12cm]{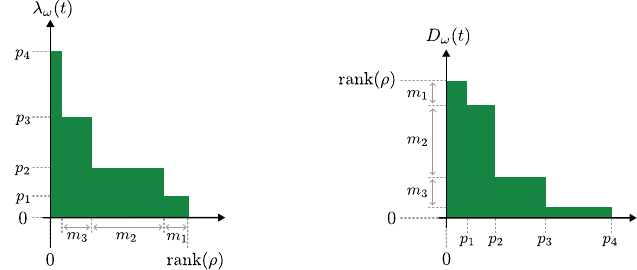}
    \caption{Spectral scale $\lambda_\omega(t)$ and distribution function $D_\omega(t)$ of a state $\omega = \tr (\rho\,\placeholder)$ on $M_n$ with density operator $\rho\in M_n$. The numbers $p_i$ are the eigenvalues of $\rho$, ordered increasingly, and $m_i\in\NN$ denotes their multiplicity.
    }
\label{fig:lambda}
\end{figure}

The following Proposition ties together spectral scales and distribution functions and shows how they relate to the distance of unitary orbits.

\begin{proposition}\label{prop:spectral_distance}
    The maps $S_*(\M)\ni\omega\mapsto D_\omega \in L^1(\RR^+)$ and $S_*(\M)\ni\omega\mapsto \lambda_\omega \in L^1(\RR^+)$ are unitarily invariant and satisfy
    \begin{equation}
        \norm{\lambda_\omega-\lambda_\varphi}_{L^1} = \norm{D_\omega-D_\varphi}_{L^1} 
    \le \inf_{u\in\U(\M)}\norm{u\omega u^*-\varphi}
    \end{equation}
    with equality if $\M$ is a factor. 
\end{proposition}

The various statements in the Proposition are contained in the works \cite{hiai1989distance,fack1986generalized,hiai1991closed,petz1985scale}. 
For the convenience of the reader, we give a short proof based on \cite{haagerup1990equivalence}:

\begin{proof}
    Unitary invariance is clear.
    $\norm{D_\omega-D_\varphi}_{L^1} \le \norm{\omega-\varphi}$ is proved in \cite[Lem.~4.3]{haagerup1990equivalence} (where the assumption of $\M$ being a factor is not used in the proof).
    By unitary invariance, the inequality on the right follows. 
    The converse inequality for $\M$ being a factor is shown in \cite[Thm.~4.4]{haagerup1990equivalence}.
    The first equality follows from the second one because the distribution function of $\lambda_\omega$ is $D_\omega$ (see \eqref{eq:D_of_lambda}) and vice versa and because conjugation by unitaries is trivial in $L^1(\RR^+)$.
\end{proof}

\section{Embezzling states}\label{sec:embezzlingstates}

This section deals with basic properties of embezzling states. We start by formally defining embezzling states and then show that several other reasonable definitions of embezzling states are equivalent to ours. We also prove the equivalence of bipartite and monopartite embezzlement for standard bipartite systems. Finally, we characterize the spectral properties of embezzling states. 

\begin{definition}\label{def:bipartite_system}
    A {\bf bipartite system} is a triple $(\H,\M,\M')$ of a Hilbert space $\H$, a von Neumann algebras $\M\subset\B(\H)$ and its commutant $\M'$.
    A (pure) {\bf bipartite state} on $(\H,\M,\M')$ is a unit vector $\Omega\in\H$ and the marginals of $\Omega$ are the states $\omega$ and $\omega'$ defined by restricting $\ip\Omega{(\placeholder)\Omega}$ to $\M$ and $\M$, respectively.
\end{definition}

Instead of $\M$ and its commutant $\M'$, one could consider the more general case of commuting von Neumann algebras $\M_A$ and $\M_B$.
Because of its importance in quantum field theory \cite{haag_local_1996}, the condition $\M_A=\M_B'$ is then called \emph{Haag duality} \cite{van_luijk_schmidt_2023}.
Operationally, Haag duality reflects that Alice can implement every unitary symmetry of Bob's observable algebra $\M_B$, i.e., every unitary $u$ on $\H$ commuting with $\M_B$ lies in $\M_A$.\footnote{We are not aware of a satisfactory interpretation of Haag duality purely in terms of correlation experiments (see \cite[Sec.~6]{van_luijk_schmidt_2023}).}

If $(\H,\M,\M')$ is a bipartite system then so is $(\H\ox\CC^n\ox\CC^n,M_n(\M),M_n(\M'))$ where we identify $M_n(\M)=\M\ox M_n\ox1$ and $M_n(\M')=\M'\ox1\ox M_n = M_n(\M)'$.

\begin{definition}
    Let $(\H,\M,\M')$ be a bipartite system.
    A pure bipartite state, i.e., a unit vector, $\Omega\in\H$ is {\bf embezzling} if for all $\Psi\in\CC^n\ox\CC^n$ and all $\eps>0$, there exist unitaries $u\in M_{n}(\M)$ and $u'\in M_{n}(\M')$ such that
    \begin{equation}\label{eq:bipartite_mbz}
        \norm{\Omega\ox\Psi - uu'\,\Omega\ox\ket{11}}<\eps.
    \end{equation}
\end{definition}

It is clear that the marginals $\omega$ and $\omega'$ of an embezzling vector state $\Omega\in\H$ satisfy the following monopartite property:

\begin{definition}
    Let $\omega$ be a normal state on a von Neumann $\M$. Then $\omega$ is {\bf embezzling}, if for all states $\psi$ on $M_n$ and all $\eps>0$ there exists unitaries $u\in M_n(\M)$ such that
    \begin{equation}\label{eq:marginal_mbz}
        \norm{\omega\ox\psi- u(\omega\ox\bra1\placeholder\ket1)u^*}<\eps.
    \end{equation}
\end{definition}

We use the following notion of approximate unitary equivalence:
\begin{definition}\label{def:approx_ue}
    Let $\omega,\varphi\in S_*(\M)$ be normal states on a von Neumann algebra $\M$. Then $\omega$ and $\varphi$ are said to be {\bf approximately unitarily equivalent}, denoted $\omega\sim\varphi$, if for all $\eps>0$ there exists a unitary $u\in\M$ such that
    \begin{equation}\label{eq:approx_ue}
         \norm{u\omega u^* -\varphi}<\eps.
    \end{equation}
\end{definition}

Clearly, approximate unitary equivalence is an equivalence relation.
Since a state $\omega$ is embezzling if and only if $\omega\ox\bra1\placeholder\ket1\sim\omega\ox\psi$ for an arbitrary state $\psi\in S(M_n)$, we have the following characterization of embezzling states:
\begin{equation}
    \omega\text{ is embezzling} \quad\iff\quad \omega\ox\psi\sim\omega\ox\phi\qquad \forall\psi,\phi\in S(M_n),\ n\in\NN.
\end{equation}
A similar statement holds for embezzling bipartite states. The main result of this section is the following:
\begin{theorem}[Equivalence of bipartite and monopartite embezzling]\label{thm:partite}
    Let $(\H,\M,\M')$ be a bipartite system, let $\Omega\in\H$ be a unit vector and let $\omega$ and $\omega'$ be the marginal states on $\M$ and $\M'$.
    The following are equivalent:
    \begin{enumerate}[(a)]
        \item\label{it:partite1} $\Omega$ is embezzling,
        \item\label{it:partite2} $\omega$ is an embezzling state on $\M$,
        \item\label{it:partite3} $\omega'$ is an embezzling state on $\M'$.
    \end{enumerate}
\end{theorem}

The proof is in several steps and will be carried out in the following subsections.

\subsection{Equivalent notions of embezzlement}

Instead of letting Alice and Bob act by unitaries on the product vector $\Omega\ox\ket{11}$, we can ask if it is possible to embezzle the state $\psi$ on $M_n$ using (partial) isometries in $M_{n,1}(\M)$ and $M_{n,1}(\M')$, respectively.

\begin{proposition}\label{thm:equiv_bipartite_mbz}
    Let $(\H,\M,\M')$ be a bipartite system with $\H$ and let $\Omega\in\H$ be a unit vector. The following are equivalent
    \begin{enumerate}[(a)]
        \item\label{it:equiv_bipartite_mbz1} $\Omega$ is embezzling
        \item\label{it:equiv_bipartite_mbz2}
            for all $\Psi\in\CC^n\ox\CC^n$, $\eps>0$ there exist isometries $v\in M_{n,1}(\M)$, $v'\in M_{n,1}(\M')$ such that
            \begin{equation}\label{eq:bipartite_mbz2}
                \norm{\Omega\ox\Psi-vv'\Omega}<\eps.
            \end{equation}
        \item\label{it:equiv_bipartite_mbz3} for all $\Psi\in\CC^n\ox\CC^n$, $\eps>0$ there exist partial isometries $v\in M_{n,1}(\M)$, $v'\in M_{n,1}(\M')$ such that \eqref{eq:bipartite_mbz2} holds, which satisfy
        \begin{equation}\label{eq:equiv_bipartite_state3}
            vv^*=s(\omega)\ox s(\psi), \quad v^*v=s(\omega),\quad v'^*v'=s(\omega'),\quad v'v'^*=s(\omega')\ox s(\psi),
        \end{equation} 
        where $\psi\in S(M_n)$ is the marginal of the vector state $\Psi$.
        \item\label{it:equiv_bipartite_mbz4} for all $\Psi\in\CC^n\ox\CC^n$, $\eps>0$ there exist contractions $v\in M_{n,1}(\M)$, $v'\in M_{n,1}(\M')$ such that \eqref{eq:bipartite_mbz2} holds,
    \end{enumerate}
\end{proposition}

The operator  $vv':\H\to\H\ox\CC^n\ox\CC^n$ in \eqref{eq:bipartite_mbz2} is defined by $vv'\Omega = \sum_{ij} v_iv'_j \Omega \ox\ket{ij}$.
We will also show the following monopartite version of \cref{thm:equiv_bipartite_mbz}:

\begin{proposition}\label{thm:equiv_marginal_mbz}
    Let $\omega$ be a normal state on a von Neumann $\M$. The following are equivalent
    \begin{enumerate}[(a)]
        \item\label{it:equiv_marginal_mbz1} $\omega$ is embezzling
        \item\label{it:equiv_marginal_mbz2} for all $\psi\in S(M_n)$, $\eps>0$, there exist isometries $v\in M_{n,1}(\M)$, such that
        \begin{equation}\label{eq:equiv_marginal_mbz}
            \norm{\omega\ox\psi-v\omega v^*}<\eps,
        \end{equation}
        \item\label{it:equiv_marginal_mbz3} for all $\psi\in S(M_n)$, $\eps>0$, there exist partial isometries $v\in M_{n,1}(\M)$ with $vv^*=s(\omega\ox\psi)$ and $v^*v=s(\omega)$ such that \eqref{eq:equiv_marginal_mbz} holds.
        \item\label{it:equiv_marginal_mbz4} for all $\psi\in S(M_n)$, $\eps>0$, there exists a contraction $v\in M_{n,1}(\M)$ such that \eqref{eq:equiv_marginal_mbz} holds.
    \end{enumerate}
\end{proposition}

An immediate consequence is the following result which often allows us to assume that embezzling states are faithful:

\begin{corollary}\label{cor:faithful}
    Let $\M$ be a von Neumann algebra with a normal state $\omega$.
    Denote by $\M_0$ the supporting corner $s(\omega)\M s(\omega)$ and by $\omega_0$ the restriction of $\M$ to $\M_0$.
    Then $\omega$ is an embezzling state on $\M$ if and only if $\omega_0$ is and embezzling state on $\M_0$.
\end{corollary}

To deduce \cref{thm:equiv_bipartite_mbz,thm:equiv_marginal_mbz}, we use a few observations about the basic fact that elements of $M_{n,1}(\M)$ are in bijection with matrices in $M_n(\M)$ with zero entries outside of the first column:

\begin{lemma}\label{lem:matrices_and_vectors}
    Let $\M$ be a von Neumann algebra. Then
    \begin{equation}
        M_{n,1}(\M) \ni v = [v_i] \longmapsto w =[w_{ij}] = [v_i \delta_{j1}] \in M_n(\M)
    \end{equation}
    is a bijection between $M_{n,1}(\M)$ and $\{[w_{ij}] \in M_n(\M) : w_{ij}=0 \ \forall j\ne 1\}$, such that $v : \H\to\H\ox \CC^n$ is a partial isometry if and only if $w$ is a partial isometry in $M_n(\M)$. Their initial and final projections are related via
    \begin{equation}
       w^*w = v^*v\ox\ketbra  11,\qquad ww^* = vv^*.
    \end{equation}
    Let $\omega$ be a normal state on $\M$. Then $v \omega v^* = w(\omega\ox\bra1\placeholder\ket1)w^*$ and
    \begin{equation}
        \Omega_{v\omega v^*} = j(w)w (\Omega_\omega \ox \ket{11}) = v'v \Omega_\omega
    \end{equation}
    in the standard form of $M_n(\M)$, where $v' := J\up nvJ$ is the partial isometry in $M_{n,1}(\M')$ corresponding to $j(w)\in M_n(\M')$ via the above bijection.
\end{lemma}
\begin{proof}
  We write $v=\sum_i v_i\otimes \ket i$ and $w=\sum_i v_i \otimes \ket i\bra1$. Hence
  \begin{align}
      w^* w = \sum_{i} v^*_i v_i\otimes \ketbra11 = v^*v\otimes\ketbra 11 ,\quad ww^* = \sum_{ij} v_i v_j^* \otimes \ket{i}\bra{j} = vv^*.
   \end{align}
  Evidently $w$ is a partial isometry if and only if $v$ is and for $x\otimes y \in \M\otimes M_n$ we have
  \begin{align}
    w(\omega\otimes \bra1\placeholder\ket1)w^*(x\otimes y) = \omega\otimes \bra1\placeholder\ket1\left(w^* x\otimes y w\right) = \sum_{i,j}\omega(v_i^* x v_j)\bra{i}y\ket{j} =  v\omega v^*(x\otimes y).
  \end{align}
  The rest follows from  and \cref{eq:purification_formula} and the standard form of $M_n(\M)$.
\end{proof}

\begin{lemma}\label{lem:simulation_by_unitaries}
    Let $\M$ be a von Neumann algebra on a Hilbert space $\H$. Let $x\in\M$ be a contraction, i.e., $\norm x\le 1$.
    If $\norm{x\Psi}=\norm\Psi$ for some $\Psi\in\H$, then for each $\epsilon>0$ there exists a unitary $u\in\M$ such that $\norm{u\Psi-x\Psi}<\eps$.
\end{lemma}

\begin{proof}
    The proof idea is taken from \cite[Lem.~2.4]{haagerup1990equivalence}. Let $0<\delta<1$.
    By the Russo-Dye Theorem \cite{kadison1985means}, $(1-\delta)x$ is a finite convex combination of unitaries $u_i\in\M$:
    \begin{equation}
        (1-\delta)x = \sum_i p_i u_i,\qquad p_i>0, \quad \sum_i p_i=1.
    \end{equation}
    Then
    \begin{align*}
        \sum p_i\norm{u_i\Psi-(1-\delta)x\Psi}^2
        &= \sum p_i (\norm{u_i\Psi}^2 + \norm{(1-\delta)x\Psi}^2 - 2\Re\ip{u_i\Psi}{(1-\delta)x\Psi})\\
        &= \norm{\Psi}^2 + (1-\delta)^2 \norm{\Psi}^2 - 2\Re\ip{(1-\delta)\Psi}{(1-\delta)\Psi} \\
        &= \norm\Psi^2 - (1-\delta)^2\norm\Psi^2 
        = \norm{\Psi}^2 (2\delta-\delta^2)
        <\norm\Psi^2 2\delta.
    \end{align*}
    Thus, for some index $i$ the unitary $u=u_i$, has to satisfy $\norm{u\Psi- (1-\delta)x\Psi}<2\delta\norm\Psi$.
    Therefore, 
    \begin{align*}
        \norm{u\Psi-x\Psi} \le \norm{u\Psi-(1-\delta)x\Psi} + \delta\norm\Psi \le \norm\Psi (\sqrt2\sqrt \delta + \delta).
    \end{align*}
\end{proof}

\begin{proof}[Proof of \cref{thm:equiv_bipartite_mbz}]
    \ref{it:equiv_bipartite_mbz3} $\Rightarrow$ \ref{it:equiv_bipartite_mbz4} is trivial.
    \ref{it:equiv_bipartite_mbz1} $\Rightarrow$ \ref{it:equiv_bipartite_mbz2}:    
    Let $\Psi\in\CC^n\ox\CC^n$ be a given unit vector and let $\eps>0$. For given unitaries $u,u'$ in $M_n(\M)$ and $M_n(\M')$, respectively, such that \eqref{eq:bipartite_mbz} holds, we define isometries $v,v'$ by $v = u (\,\placeholder\ox\ket1)$ and $v'=u'(\placeholder\ox\ket1)$. It follows then that \eqref{eq:bipartite_mbz2} holds because $vv'\Omega = uu'(\Omega\ox\ket{11})$.

    \ref{it:equiv_bipartite_mbz2} $\Rightarrow$ \ref{it:equiv_bipartite_mbz3}:
    Let $\Psi\in\CC^n\ox\CC^n$ be a given unit vector and let $\eps>0$. For given isometries $v,v'$ in $M_n(\M)$ and $M_n(\M')$, respectively, such that \eqref{eq:bipartite_mbz2} holds, we define partial isometries $w,w'$ by $w = [s(\omega)\ox s(\psi)]vs(\omega)$ and $w'=[s(\omega')\ox s(\psi)]vs(\omega')$ (note that $\psi'\equiv\psi$).
    Then the estimate follows from
    \begin{align*}
        \norm{\Omega\ox\Psi-ww'\Omega} = \norm{[s(\omega)s(\omega')\ox s(\psi)\ox s(\omega')](\Omega\ox\Psi-vv'\Omega)}
        \le \norm{\Omega\ox\Psi-vv'\Omega} <\eps.
    \end{align*}

    \ref{it:equiv_bipartite_mbz4} $\Rightarrow$ \ref{it:equiv_bipartite_mbz1}: 
    Let $\Psi\in\CC^n\ox\CC^n$ be a given unit vector and let $\eps>0$.
    Let $v\in M_n(\M)$ and $v'\in M_n(\M')$ be contractions such that \eqref{eq:bipartite_mbz2} holds with error $\frac\eps2$.
    Let $w\in M_n(\M)$ and $w'\in M_n(\M')$ be the contractions corresponding to $v$ and $v'$ via \cref{lem:matrices_and_vectors}. In particular, $ww'(\Omega\ox\ket{11})=vv'\Omega$.
    Without loss of generality we can assume that $\norm{ww'(\Omega\ox\ket{11}}=1=\norm{\Omega\ox\ket{11}}$.
    Applying \cref{lem:simulation_by_unitaries} twice lets us pick unitaries $u\in M_n(\M)$ and $u'\in M_n(\M')$ such that $\norm{uu'(\Omega\ox\ket{11})-ww'(\Omega\ox\ket{11})}<\frac\eps2$.
    This implies
    \begin{multline}
        \norm{\Omega\ox\Psi-uu'(\Omega\ox\ket{11})}
        \le \norm{\Omega\ox\Psi- ww'(\Omega\ox\ket{11})} + \norm{(uu'-ww')(\Omega\ox\ket{11})} \\
        = \norm{\Omega\ox\Psi- vv'\Omega} + \frac\eps2 <\frac\eps2+\frac\eps2=\eps.
    \end{multline}
\end{proof}

\begin{proof}[Proof of \cref{thm:equiv_marginal_mbz}]
    The implications \ref{it:equiv_marginal_mbz1} $\Rightarrow$ \ref{it:equiv_marginal_mbz2} $\Rightarrow$ \ref{it:equiv_marginal_mbz3} $\Rightarrow$ \ref{it:equiv_marginal_mbz4} are proven with the exact same techniques as we used in the proof of the bipartite-version \cref{thm:equiv_marginal_mbz}.
    The difference is that we only need to keep track of one system and that we work with states instead of vectors, e.g., $uu'(\Omega\ox\ket{11})$ is replaced by $u(\omega\ox\bra1\placeholder\ket1)u^*$ and $vv'\Omega$ is replaced by $v\omega v^*$.

    \ref{it:equiv_marginal_mbz4} $\Rightarrow$ \ref{it:equiv_marginal_mbz1}: This argument is taken from \cite[Lem.~2.4]{haagerup1990equivalence}.
    Let $\psi\in S(M_n)$ and $\eps>0$.
    Let $v\in M_{n,1}(\M)$ be as in \ref{it:equiv_bipartite_mbz4} and let $w\in M_n(\M)$ be the contraction corresponding to it via \cref{lem:matrices_and_vectors}.
    Then $w(\omega \ox\bra1\placeholder\ket1)w^*=v\omega v^*$ and, hence, $\norm{\omega\ox\psi-w(\omega\ox\bra1\placeholder\ket1)w^*}<\eps$.
    Since we can construct such a contraction $w$ for all $\eps>0$, \cite[Lem.~2.4]{haagerup1990equivalence} implies that we can also find unitaries $u\in\M$ such that $\norm{\omega\ox\psi-u(\omega\ox\bra1\placeholder\ket1)u^*}<\eps$ for all $\eps>0$.
\end{proof}

\subsection{Standard bipartite systems}

\begin{definition}
    A bipartite system $(\H,\M,\M')$ is {\bf standard} if $\M$ and, hence $\M'$, is in standard representation.
\end{definition}

In standard bipartite systems, the setup is completely symmetric for Alice  and Bob: 
The modular conjugation $J$ implements the exchange symmetry between $J\M J=\M'$.

\begin{lemma}\label{lem:standard}
    Let $(\H,\M,\M')$ be a $\sigma$-finite bipartite system, i.e., $\M$ and $\M'$ are both $\sigma$-finite. Then $(\H,\M,\M')$ is standard if and only if all normal states $\omega\in S_*(\M)$ and $\omega'\in S_*(\M')$ arise as marginals of vectors states.
\end{lemma}

\begin{proof}
    As explained in \cref{sec:preliminaries}, $\M$ is in standard representation if and only if $\M'$ is. 
    In the standard representation, all states on both algebras are implemented by vectors in the positive cone.
    For the converse, let $\Omega,\Omega'\in\H$ be vectors implementing faithful normal states on $\M$ and $\M'$ (which exist because $\M$ and $\M'$ are $\sigma$-finite). 
    Then $\Omega$ is separating for $\M$ and $\Omega'$ is separating for $\M'$, hence cyclic for $\M$.
    By \cite[Thm.~III.2.6.10]{blackadar_operator_2006}, a vector $\Omega$ exists, which is both cyclic and separating. Hence, we are in standard form.
\end{proof}

\begin{proposition}\label{prop:std_form_mbz}
    Let $(\H,J,\P)$ be the standard form of a von Neumann algebra $\M$. Let $\Omega\in\H$ be a unit vector, $\omega$ the induced normal state on $\M$, and $\Omega_\omega\in\P$ the corresponding vector in the positive cone.
    The following are equivalent:
    \begin{enumerate}[(i)]
        \item\label{it:std_form_mbz1} $\Omega$ is embezzling,
        \item\label{it:std_form_mbz2} $\omega$ is embezzling,
        \item\label{it:std_form_mbz3} $\Omega_\omega$ is embezzling.
    \end{enumerate}
\end{proposition}

For the proof, we need two Lemmas:

\begin{lemma}\label{lem:polar}
    Let $(\H,\M,\M')$ be a standard bipartite system. Let $J$, $\P$ be the modular conjugation and positive cone of the standard form.
    Then, for all $\Omega\in\H$ there exist partial isometries $u\in\M$, $u'\in\M'$ such that $u\Omega=u'\Omega\in\P$ and such that $u^*u =[\M'\Omega]$ and $u'^*u'=[\M\Omega]$.
\end{lemma}
\begin{proof}
    By the symmetry between $\M$ and $\M'$, we only need to show the claim for $\M$. By \cite[Ex.~IX.1.2)]{takesaki2}, there exists a vector $\abs\Omega\in\P$ and a partial isometry $v\in\M$ such that $vv^*=[\M'\Omega]$, $v^*v=[\M'\abs\Omega]$ and $\Omega = v\abs\Omega$. Thus, the claim holds for $u=v^*$.
\end{proof}

\begin{lemma}\label{lem:embezzling_vectors}
    Let $(\H,\M,\M')$ be a bipartite system. 
    Then the set of embezzling vectors is norm-closed and invariant under local unitaries, i.e., if $\Omega$ is embezzling, then the same is true for $uu'\Omega$ for each pair of unitaries $u,u'$ in $\M$ and $\M'$, respectively.
\end{lemma}
\begin{proof}
    Clear.    
\end{proof}

\begin{proof}[Proof of \cref{prop:std_form_mbz}]
    \ref{it:std_form_mbz1} $\Rightarrow$ \ref{it:std_form_mbz2} is clear.
    \ref{it:std_form_mbz2} $\Rightarrow$ \ref{it:std_form_mbz3}: 
    Let $\psi$ be a state on $M_n$, let $\Psi$ be the corresponding vector in the positive cone, and let $\eps>0$. 
    If we take $M_n(\M)$ to be in standard form on $\H\ox\CC^n\ox\CC^n$ (see \cref{lem:standard}) then
    \begin{equation}
        \Omega_{\omega\ox\bra1\placeholder\ket1} = \Omega_\omega\ox\ket{11},\qquad \Omega_{\omega\ox\psi}=\Omega_\omega\ox\Psi.
    \end{equation}
    Let $u\in M_n(\M)$ be a unitary such that $\norm{\omega\ox\psi-u(\omega\ox\bra1\placeholder\ket1)u^*}<\eps^2$ and set $u'=j(u)\in M_n(\M')=M_n(\M)'$. 
    Combining \eqref{eq:state_vector_norms} and \eqref{eq:purification_formula}, we get
    \begin{align}\label{eq:symmetric_choice}
        \norm{\Omega_\omega\ox\Psi - uu'\Omega_\omega\ox\ket{11}}
        =\norm{\Omega_{\omega\ox\psi}-\Omega_{u(\omega\ox\bra1\placeholder\ket1)u^*}}
        \le\norm{\omega\ox\psi-u(\omega\ox\bra1\placeholder\ket1)u^*}^{1/2} = \eps.
    \end{align}
    If $\Psi$ is not in the positive cone, the same estimate holds if $u'$ is multiplied by the adjoint of a unitary in $u_0'$ in $1\ox M_n$ such that $1\ox u_0'\Psi$ is in the positive cone.
    To see this, use the polar decomposition of the matrix $[\Psi_{ij}]$ such that $\Psi=\sum_{ij} \Psi_{ij}\ox\ket{ij}$. 
    Therefore, $\Omega_\omega$ is embezzling.

    \ref{it:std_form_mbz3} $\Rightarrow$ \ref{it:std_form_mbz1}:
    By the polar decomposition of vectors in standard form \cite[Ex.~IX.1.2)]{takesaki2}, there exists a partial isometry $v'\in\M'$ such that $\Omega = v'\Omega_\omega$.
    Since $\Omega$ and $\Omega_\omega$ are unit vectors,  \cref{lem:simulation_by_unitaries} implies that for each $k>0$, we can find a unitary $v_k'\in\M'$ such that $\norm{(v_k'-v')\Omega_\omega}=\norm{\Omega_k-\Omega}<k^{-1}$, where $\Omega_k=v_k'\Omega$.
    By \cref{lem:embezzling_vectors}, $\Omega_k$ an embezzling vector.
    Therefore, $\Omega=\lim_k\Omega_k$ is the limit of a sequence of embezzling vectors and, hence, an embezzling vector (see \cref{lem:embezzling_vectors}).
\end{proof}

\begin{lemma}\label{lem:projection}
    Let $\M$ be a von Neumann algebra on $\H$ and let $p'\in\proj(\M')$ be a projection in the commutant.
    If $\omega$ is an embezzling state on $\M$, then $p'\omega p'$ is an embezzling state on $p'\M p'$.
\end{lemma}
\begin{proof}
    Clear.
\end{proof}

\begin{proposition}\label{thm:WLOG_std}
    Let $(\H,\M\,\M')$ be a bipartite system, and let $\Omega\in\H$ be a unit vector. Let $p\ge s(\omega)$, $p'\ge s(\omega')$ be projections in $\M$ and $\M'$, respectively, and let $p_0=pp'$.
    Set $\H_0 = p_0\H$, $\M_0=p_0\M p_0$ and $\Omega_0 \equiv \Omega\in\H_0\subset\H$.
    Then $\M_0' = p_0\M'p_0$ and the following are equivalent
    \begin{enumerate}[(a)]
        \item\label{it:WLOG_std1} $\Omega$ is embezzling for $(\H,\M,\M')$
        \item\label{it:WLOG_std2} $\Omega_0$ is embezzling for $(\H_0,\M_0,\M_0')$.
    \end{enumerate}
    If $p=s(\omega)$ and $p'=s(\omega')$ then $\Omega_0$ is cyclic and separating for $\M_0$ and, hence, $(\H_0,\M_0,\M_0')$ is a standard bipartite system.
\end{proposition}

We remark that $\M_0=p_0\M p_0$ is naturally von Neumann algebra acting on $\H_0$ because $p_0$ is the product of a projection in $\M$ and a projection in $\M'$.

\begin{proof}
    First recall  $s(\omega)$ and $ps(\omega')$ are the projection onto $[\M'\Omega]$ and $[\M\Omega]$, respectively.
    Clearly, the two projections $p$ and $p'$ commute so that $p_0 = pp'$ is a projection as well.
    By \cite[Cor.~5.5.7]{KadisonRingrose1}, we have $(q\N q)' = q\N'q$ for a von Neumann algebra $\N$ and a projection $q\in\proj(\N)$ or a projection $q\in\proj(\N')$.
    Therefore $p_0 = pp'p\in p\M'p=(p\M p)'$, $p_0=p'pp'\in p'\M p'$, and, hence 
    \begin{equation}
        (p_0\M p_0)' = [p' (p\M p) p' ]' = p'(p\M p)' p' = p'p\M'pp' = p_0\M'p_0.
    \end{equation}
    If $p$ and $p'$ are the support projections then, by construction of $\H_0$, $\Omega_0$ is cyclic for $\M_0$ and $\M_0'$ and, hence, cyclic and separating for $\M_0$. 
    We remark that a similar setup was considered in \cite[Prop.~39]{van_luijk_schmidt_2023}.

    \ref{it:WLOG_std1} $\Rightarrow$ \ref{it:WLOG_std2}: 
    Denote by $\omega_0$ the state induced by $\Omega_0$ on $\M_0$. 
    We will show that $\omega_0$ is embezzling.
    Let $\Psi\in\CC^n\ox\CC^n$ be a unit vector, let $\eps>0$, and let $u\in M_n(\M)$, $u'\in M_n(\M')$ be unitaries such that \eqref{eq:bipartite_mbz} holds.
    We define contractions $a =(p_0\ox1)u(p_0\ox1)\in M_n(\M_0)$ and $a'=(1\ox p_0)u'(p_0\ox1)\in M_n(\M')$.
    Using that $p_0\Omega_0=\Omega_0=\Omega$, we get
    \begin{multline}
        \norm{\Omega_0\ox\Psi-aa'(\Omega_0\ox\ket{11})} 
        = \norm{(p_0\ox1\ox1) (\Omega\ox\Psi- uu'(\Omega\ox\ket{11}))} \\
        \le \norm{\Omega\ox\Psi- uu'(\Omega\ox\ket{11})} <\eps.
    \end{multline}
    Since we can construct such a contraction for all $\eps>0$ and since $\Psi$ was arbitarary, \cref{it:equiv_bipartite_mbz4} of \cref{thm:equiv_bipartite_mbz} holds for $\Omega_0$ and, hence, $\Omega_0$ is embezzling.

    \ref{it:WLOG_std2} $\Rightarrow$ \ref{it:WLOG_std1}: 
    Let $\Psi\in\CC^n\ox\CC^n$ be a unit vector, let $\eps>0$, and let $u_0\in M_n(\M_0)$, $u_0'\in M_n(\M'_0)$ be unitaries such that \eqref{eq:bipartite_mbz} holds for $\Omega_0$.
    Let $w\in M_n(\M)$ and $w'\in M_n(\M')$ be contractions with $p_0wp_0=u_0$ and $p_0w'p_0=u_0'$.
    Define $v\in M_{n,1}(\M)$ and $v'\in M_{n,1}(\M')$ to be elements corresponding to $w,w'$ via \cref{lem:matrices_and_vectors}.
    Then $vv' \Omega = ww'(\Omega\ox\ket{11})= u_0u_0(\Omega\ox\ket{11})$ and, hence, $\norm{\Omega\ox\Psi-vv'\Omega}<\eps$.
    Therefore $\Omega$ satisfies \cref{it:equiv_bipartite_mbz4} of \cref{thm:equiv_bipartite_mbz} and, hence, is embezzling.
\end{proof}

\begin{proof}[Proof of \cref{thm:partite}]
    It is clear that \ref{it:partite1} implies \ref{it:partite2} and \ref{it:partite3}.
    We only need to show \ref{it:partite2} $\Rightarrow$ \ref{it:partite1}.
    Let $p=s(\omega)$, $p'=s(\omega')$, $p_0=pp'$, $\H_0=p_0\H$, and $\M_0=p_0\M p_0$.
    By \cref{cor:faithful}, $p\omega p$ is an embezzling state on $p\M p$ which, by \cref{lem:projection}, implies that $p' p\omega pp' = \omega_0$ is an embezzling state on $\M_0$.
    Since $(\H_0,\M_0,\M_0')$ is in standard form, \cref{prop:std_form_mbz} implies that $\Omega_0$ is embezzling.
    It now follows from \cref{thm:WLOG_std}, that $\Omega$ is embezzling for $(\H,\M,\M')$
\end{proof}

\subsection{Spectral properties}

We now discuss spectral properties of embezzling states $\omega$. This involves two topics: The behavior of the distribution function $D_\omega$ and the spectral scale $\lambda_\omega$ (see \cref{subsec:pre:spectral_scales}) as well as the spectrum of the modular operator of embezzling states (see \cref{subsec:pre:weights}). We begin with the former, which yields a simple proof that semifinite factors cannot host embezzling states.

The distribution function behaves naturally with respect to tensor products: Let $\M$ and $\P$ be semifinite von Neumann algebras with faithful normal traces $\tr_\M$ and $\tr_\P$, respectively.
We equip the tensor product $\M\ox\P$ with the product trace $\tr_{\M\ox\P}=\tr_\M\ox\tr_\P$.
Given two normal states $\omega$ and $\varphi$ on $\M$ and $\P$ respectively, the spectral measure $p_{\omega\ox\varphi}$ of the density operator $\rho_{\omega\ox\varphi}=\rho_\omega\ox\rho_\varphi$ is given by the tensor convolution of the individual spectral measures
\begin{align}\label{eq:spectral_convolution}
    p_{\omega\ox\varphi} & = p_{\omega}\circledast p_{\varphi},
\end{align}
which is the $\proj(\M\ox\P)$-valued Borel measure on $\RR^+$ defined by $p_\omega\circledast p_\varphi (A) = \int_{0}^{\infty}\!\!\int_{0}^{\infty}\!\!\chi_{A}(ts)\, dp_{\omega}(t)\ox dp_{\varphi}(s)$.\footnote{Equivalently, $p_\omega\circledast p_\varphi$ the pushforward of the tensor-product measure $p_{\omega}\ox p_{\varphi}$ along the multiplication-map $\RR^{+}\times\RR^{+}\ni(t,s)\mapsto t\cdot s\in\RR^{+}$.}
This entails the following convolution formula for the distribution function of product states:

\begin{lemma}\label{lem:dist_convolution}
    Let $\M$ and $\P$ be semifinite von Neumann algebras. Given two normal states $\omega\in S_{*}(\M)$ and $\varphi\in S_{*}(\P)$, the distribution function of the product state $\omega\ox\varphi$ is given by:
    \begin{align}\label{eq:dist_convolution}
        D_{\omega\ox\varphi}(t) & = \Tr_{\M}((D_{\varphi}\ast p_{\omega})(t)) = \Tr_{\M}(D_{\varphi}(t \rho_{\omega}^{-1})).
    \end{align}
    Equivalently, we get $D_{\omega\ox\varphi}(t) = \Tr_{\P}(D_{\omega}(t \rho_{\varphi}^{-1}))$.

    In terms of the spectral scales of $\omega$ and $\psi$, we can write:
    \begin{align}\label{eq:scale_convolution}
    D_{\omega\ox\varphi}(t) & = \int_{0}^{\infty}\int_{0}^{\infty}\chi_{t}(\lambda_{\omega}(r)\lambda_{\psi}(s))\,dr\,ds = \int_{0}^{\infty}D_{\varphi}(\lambda_{\omega}(r)^{-1}t)\,dr = \int_{0}^{\infty}D_{\omega}(\lambda_{\psi}(s)^{-1}t)\,ds
\end{align}
\end{lemma}
\begin{proof}
Using \cref{eq:spectral_convolution}, a direct computaton yields:
\begin{align} \label{eq:dist_convolution_calc} \nonumber
D_{\omega\ox\varphi}(t) & = \Tr_{\M\ox\P}(p_{\omega\ox\varphi}((t,\infty))) = (\Tr_{\M}(p_{\omega})\ast\Tr_{\P}(p_{\varphi}))((t,\infty)) \\ \nonumber
& = \int_{0}^{\infty}\int_{0}^{\infty}\chi_{t}(rs)\,\Tr_{\M}(dp_{\omega}(r))\,\Tr_{\P}(dp_{\varphi}(s)) \\ 
& = \int_{0}^{\infty}D_{\varphi}(r^{-1}t)\,\Tr_{\M}(dp_{\omega}(r)) = \Tr_{\M}((D_{\varphi}\ast p_{\omega})(t)) = \Tr_{\M}(D_{\varphi}(t \rho_{\omega}^{-1})),
\end{align}
where we used the convolution of Borel functions and measures on $\RR^{+}$ in the last two lines:
\begin{align}\label{eq:f_m_conv}
    (f\ast\mu)(t) & = \int_{0}^{\infty}f(s^{-1}t)\, d\mu(s).
\end{align}
Reversing the roles of $\omega$ and $\varphi$, we find the corresponding statement involving $D_{\omega}$ and $\rho_{\varphi}$. The formula involving the spectral scales follows from the second line of \cref{eq:dist_convolution_calc} and the fact that the trace of the spectral measure of $\rho_{\omega}$ (or $\rho_{\varphi}$) is the pushforward of the Lebesgue measure by the spectral scale $\lambda_{\omega}$ (respectively $\lambda_{\varphi}$).
\end{proof}
\begin{remark}\label{rem:random_variables}
The distribution function of a state behaves like the cumulative distribution function of a $\RR^{+}$-valued random variable such that the tensor product of two states translates into the multiplication of the associated random variables. This is illustrated by the special case in which $\Tr_{\M}(p_{\omega})$ (or similarly $\Tr_{\P}(p_{\varphi})$) admits a distributional Radon-Nikodym derivative with respect to the Lebesgue measure $dt$ on $\RR$. Under said assumption, \cref{eq:dist_convolution} yields:
\begin{align}\label{eq:dist_convolution_cont}
    D_{\omega\ox\varphi}(t) & = -t\tfrac{d}{dt}\int_{0}^{\infty}D_{\omega}(r)D_{\varphi}(r^{-1}t)\, \tfrac{dr}{r}
\end{align}
which follows from the distributional identity $D'_{\omega}(t) = -\tfrac{\Tr(dp_{\omega})}{dt}$.
\end{remark}

We now consider the matrix amplification $M_n(\M)$ and equip it with the trace $\tr\ox\tr_n$ where $\tr_n$ is the standard trace on $M_n$.
It follows immediately from the definition of the distribution function in \cref{eq:distribution_function}, the convolution formula \cref{eq:dist_convolution}, and \cref{exa:M_n_spectrum} that
\begin{align}
    &&D_{\omega\ox\bra1\placeholder\ket1}(t) &= D_\omega(t),&D_{\omega\ox\frac1n\!\tr}(t)&= nD_\omega(nt).&&\label{eq:D_amplification}
\intertext{Hence, the spectral scales are given by}
    &&\lambda_{\omega\ox\bra1\placeholder\ket1}(t)&= \lambda_\omega(t),&\lambda_{\omega\ox\frac1n\!\tr}(t) &= \tfrac1n \lambda_\omega(\tfrac tn).&&\label{eq:l_amplification}
\end{align}

\begin{proposition}\label{thm:spectral_charac}
    Let $\M$ be a von Neumann algebra and let $\tilde\lambda_\omega:(0,\oo)\to\RR^+$ (resp.\ $\tilde D_\omega:(0,\oo)\to\RR^+$) be a non-zero right-continuous function defined for all normal states $\omega$ on $\M$ and $M_n(\M)$ for all $n$, such that 
    \begin{itemize}
        \item if states $\omega,\varphi$ on $M_n(\M)$ are approximately unitarily equivalent, then $\tilde \lambda_{\omega}(t)=\tilde \lambda_\varphi(t)$ (resp.\ $\tilde D_\omega(t)=\tilde D_\varphi(t)$),
        \item $\tilde \lambda_\omega(t)$ satisfies formula \eqref{eq:l_amplification} (resp.\ $\tilde D_\omega(t)$ satisfies \eqref{eq:D_amplification}) holds.
    \end{itemize}
    If $\omega$ is embezzling, then 
    \begin{equation}
        \tilde \lambda_\omega(t) \propto \frac1t \qquad \Big(\text{resp.}\ \tilde D_\omega(t)\propto \frac1t\ \Big)\qquad t>0.
    \end{equation}
\end{proposition}
\begin{proof}
    Since $\omega$ is embezzling, $\omega\ox\bra1\placeholder\ket1$ and $\omega\ox\frac1n\!\tr$ are approximately unitarily equivalent for all $n\in\NN$. 
    Therefore, $\tilde\lambda_\omega(t) =\lambda_{\omega\ox\bra1\placeholder\ket1}(t)= \tilde\lambda_{\omega\ox\frac1n\tr}(t) = \frac1n \tilde\lambda_\omega(\frac tn)$ for all $n$. Consequently,
    \begin{equation}
        \tilde\lambda_\omega(t) =\tfrac1n\tilde\lambda_\omega(t\tfrac 1n)= \tfrac 1n \tilde\lambda_\omega(t\tfrac mn\cdot\tfrac1m) = \tfrac mn \tilde\lambda_\omega(t\tfrac mn),\qquad n,m\in\NN.
    \end{equation}
    This shows $\tilde\lambda_\omega(t) = q\tilde\lambda_\omega(tq)$ for all rational numbers $q>0$.
    In combination with right-continuity, this gives $\tilde\lambda_\omega(t) = \frac1t\tilde\lambda_\omega(1)$ for all $t>0$.
\end{proof}

As mentioned, the spectral scale $\lambda_\omega(t)$ and the distribution function $D_\omega(t)$ for states $\omega$ on semifinite factors satisfy the criteria in \cref{thm:spectral_charac}.
Since they are both right continuous probability distribution on $\RR^+$, no embezzling states can exist on semifinite von Neumann algebras because $\frac1t$ is not integrable.
We will see later that nontrivial solutions to the assumptions of \cref{thm:spectral_charac} exist even for type $\III$ factors (see \cref{sec:FOW_type_III_lambda}). 
We also remark that, if $\tilde D_\omega(t)$ satisfies the assumption of \cref{thm:spectral_charac}, then so does $\tilde\lambda_\omega(t) = \inf\{ s>0 : \tilde D_\omega(t)\le t\}$ (cp.\ \eqref{eq:spectral_scale}).

Before we continue the study of spectral properties of embezzling states, we show how \cref{thm:spectral_charac} rules out the existence of embezzling states on semifinite von Neumann algebras.
We start with the following Lemma:

\begin{lemma}
    Let $\M=\bigoplus_{i\in I}\M_i$ be a direct sum of von Neumann algebras and let $\omega=\oplus \omega_i$ be a normal state on $\M$.
    Then $\omega$ is embezzling if and only if for each $i\in I$ we have either $\omega_i=0$ or $\omega_i$ is proportional to an embezzling state.
\end{lemma}
\begin{proof}
    Unitaries $u\in\M$ decompose as direct sums of unitaries $u_i\in\M_i$, which can be chosen independently. Moreover, for any state $\psi$ on $M_n$ we have
    \begin{align}
        \norm{\omega\otimes \psi - u(\omega\otimes \bra 1\cdot\ket 1)u^*} = \sum_i \norm{\omega_i\otimes \psi - u_i(\omega_i \otimes\bra1\cdot \ket 1)u_i^*},
    \end{align}
    which implies the claim.
\end{proof}

\begin{corollary}\label{cor:no_semifinite_embezzling}
    Let $\M$ be a von Neumann algebra and $\M=\P\oplus\R$ the decomposition into a semifinite and a type $\III$ von Neumann algebra.
    A normal state $\omega\in S_*(\M)$ is embezzling if and only if $\omega=0\oplus\phi$ with $\phi$ an embezzling state on $\R$.
\end{corollary}

Recall that a general von Neumann algebra $\M$ has a (unique) direct sum decomposition $\M = \P\oplus\R$ such that $\P$ semifinite and $\R$ is type $\III$.
If $\M=\P\oplus\R$, as above, acts on $\H$, then we can decompose $\H$ as $\mc J\oplus\K$ such that $\P$ only acts on $\mc J$ and $\R$ only acts on $\K$.
It follows that the commutant is $\M'=\P'\oplus\R'$, which is the direct sum decomposition of $\M'$ into a semifinite and a type $\III$ algebra.
If $(\H,\M,\M')$ is a bipartite system this gives us bipartite systems $(\mc J,\P,\P')$ and $(\K,\R,\R')$.
Therefore,
\begin{equation}\label{eq:direct_sum}
    (\H,\M,\M')= (\mc J,\P,\P')\oplus (\K,\R,\R')
\end{equation}
is the unique decomposition of the bipartite system $(\H,\M,\M')$ as a direct sum of a semifinite bipartite system and a type $\III$ bipartite system (with the obvious definitions).

\begin{corollary}\label{thm:no_semifinite_embezzling}
    Let $(\H,\M,\M')$ be a bipartite system and consider the direct sum decomposition into a semifinite and and a type $\III$ bipartite system.
    Let $\Omega \in\H$ be a unit vector.
    Then $\Omega$ is embezzling if and only if  $\Omega= 0\oplus\Phi\in\H$ with $\Phi$ embezzling for $(\K,\R,\R')$.
\end{corollary}
\null

We now return to the study of spectral properties of embezzling states.
We will show that the modular spectrum of an embezzling state is always the full positive real line $\Sp\Delta_\omega=\RR^+$.
Let us briefly give some intuition for the modular spectrum: For  a faithful state $\omega$ on $M_n$ represented by the density matrix $\rho_\omega$, we have seen in \cref{exa:M_n_std_form} that $\Delta_\omega = \rho_\omega\ox (\overline{\rho}_{\omega})^{-1}$. Therefore the modular spectrum $\Sp\Delta_\omega$ consists of all ratios of eigenvalues: $\Sp\Delta_\omega = \{\frac{p}{q}: p,q\in\Sp(\rho_\omega)\}$.

\begin{theorem}\label{thm:modular_spectrum}
    If $\omega$ is an embezzling state on a $\sigma$-finite von Neumann algebra $\M$, then its modular spectrum is
    \begin{equation}
        \Sp\Delta_\omega =\RR^+.
    \end{equation}
\end{theorem}

Note that since $\Delta_\omega$ is always positive, the theorem asserts that the modular spectrum of embezzling states is maximal.
The converse to this is false: For example, the density operator $\rho = 6\pi^{-2}\sum n^{-2}\kettbra n$ on $\ell^2(\NN)$ determines a normal state $\omega = \tr\rho(\placeholder)$ with modular spectrum $\RR^+$. Since $\omega$ is a state on a type $\I$ factor, it cannot be embezzling (see \cref{cor:no_semifinite_embezzling}).

Foreshadowing our discussion of universal embezzlement in \cref{sec:universal}, we immediately infer from \cref{thm:modular_spectrum}:
\begin{corollary}\label{cor:universal_mbz_spectrum}
Let $\M$ be a $\sigma$-finite von Neumann algebra such that every normal state $\omega\in S_{*}(\M)$ is embezzling. Then, $\M$ is of type $\III_{1}$.
\end{corollary}
\begin{proof}
    Since $\M$ is $\sigma$-finite, the Connes invariant $\mathrm S(\M)$ is given by \cite{connes1973classIII,stratila2020modular}
    \begin{align}\label{eq:S_state}
        \mathrm S(\M) & = \bigcap_{\substack{\omega\in S_{*}(\M)\\ \textup{faithful}}}\Sp\Delta_{\omega} = \RR^{+}.
    \end{align}
    Thus, $\mathrm S(\M)= \RR^+$ and $\M$ is of type $\III_{1}$.
\end{proof}

The basic idea of the proof of \cref{thm:modular_spectrum} will be to pass to an ultrapower of $\M$ to turn approximate unitary equivalence into an exact one. We refer to \cite{ando2014ultraproducts} for details on ultrapowers and explain the basics in the following. 
To start, we fix a free ultrafilter $\F$ on $\NN$.
The (Ocneanu) ultrapower $\M^\F$ of $\M$ is defined as the quotient $\Q_\F/\mc I_\F$, where 
\begin{equation}
\begin{aligned}
     \mc I_\F&:= \{ (x_n)\in \ell^\oo(\NN,\M) : \text{$s^*$-}\!\!\lim_{n\to\mc F} x_n=0\}, \\
     \mc Q_\F&:= \{ (x_n) \in\ell^\oo(\NN,\M) : (x_n) \mc I_\F,\,\mc I_\F(x_n) \subset \mc I_\F\}
\end{aligned}
\end{equation}
with "$\text{$s^*$-}\!\lim_{n\to\mc F}$" denoting the strong$^*$-limit along the ultrafilter $\mc F$ \cite{ando2014ultraproducts}.
$\Q_\F$ as defined above is a $C^*$-algebra, $\mc I_\F$ is a closed two-sided $^*$-ideal in $\Q_\F$ and their quotient, the ultrapower $\M^\F$ is an abstract von Neumann algebra, i.e., a $C^*$-algebra with predual.
We denote the equivalence class of $(x_n)\in\Q_\F$ in $\M^\F$ as $x_\oo$. Note that $(x_\oo)^* = x_\oo^*$.
All aithful normal states $\omega$ on $\M$ determine a faithful ultrapower state $\omega^\F\in S_*(\M^\F)$ which is defined via 
\begin{equation}
    \omega^\F(x_\oo) = \lim_{n\to\F} \omega(x_n), \qquad (x_n)\in\Q_\F.
\end{equation}
It is clear from the construction of the ultrapower that the matrix amplification $M_n(\M^\F)$ is naturally isomorphic to the ultrapower $M_n(\M)^\F$.
In the following, we identify the two.

\begin{lemma}\label{lem:u_in_Q}
    Let $(u_n)$ be a sequence of unitaries in $\M$ and let $\norm{u_n\phi u_n^*-\psi}\to0$ for two faithful normal states $\psi,\phi$ on $\M$.
    Then $(u_n)\in\Q_\F$, $u_\oo$ is a unitary in $\M^\F$, and 
    \begin{equation}\label{eq:exact_ue}
        u_\oo \phi^\F u_\oo^* = \psi^\F.
    \end{equation}
\end{lemma}

\begin{proof}
    Recall that a uniformly bounded net $(x_\alpha)$ converges to $0$ in the strong$^*$-topology if and only if $\phi(x_n^*x_n)$ and $\phi(x_nx_n^*)$ both converge to zero for some, hence all, faithful normal states $\phi$ \cite[Prop.~III.5.3]{takesaki1}.
    Let $(x_n)\in\mc I_\F$. Then $\phi((x_nu_n)^*(u_nx_n)) = \phi(x_n^*x_n)\to0$, $\phi((u_nx_n)(u_nx_n)^*)=\phi(x_nx_n^*)\to0$, and
    \begin{equation}
        \phi((u_nx_n)^*(u_nx_n) ) = (u_n\phi u_n^*)(x_n^*x_n) \le \norm{u_n\phi u_n^*-\psi} \norm{x_n}^2 + \psi(x_n^*x_n) \to 0.
    \end{equation}
    Analogously, one sees $\phi((u_nx_n)(u_nx_n)^*)\to 0$. Together these imply $(u_nx_n),\, (x_n u_n)\in\mc I_\F$ showing $(u_n)\in \Q_\F$.
    Finally, \eqref{eq:exact_ue} follows from 
    \begin{align*}
        \abs{(u_\oo \phi^\F u_\oo^*-\psi^\F)(x_\oo )}= \lim_{n\to\F} \abs{(u_n\phi u_n^*-\psi)(x_n)}
        \le \lim_{n\to\F} \norm{u_n\phi u_n^*-\psi} \norm{x_n} =0.
    \end{align*}
\end{proof}

As an immediate consequence of the preceding result, we note that the modular spectrum $\Sp\Delta_{\omega}$ is invariant under approximate unitary equivalence.

\begin{corollary}\label{eq:approx_inv_mod_spec}
    Let $\omega,\varphi\in S_{*}(\M)$ be normal states on a von Neumann algebra $\M$. If $\omega$ and $\varphi$ are approximately unitarily equivalent, i.e., $\omega\sim\varphi$, then $\Sp\Delta_{\omega}=\Sp\Delta_{\varphi}$.
\end{corollary}
\begin{proof}
    Recall that $\Delta_{\omega}$ is defined with respect to support corner $\M_{s(\omega)}=s(\omega)\M s(\omega)$ (and analogously for $\Delta_{\varphi}$). According to \cite[Lem.~2.3]{haagerup1990equivalence}, we have $s(\omega)\sim s(\varphi)$ in the Murray-von Neumann sense. Let $v\in\M$ be a partial isometry realizing the equivalence, i.e., $v^{*}v=s(\omega)$ and $vv^{*}=s(\varphi)$. Then, $v^{*}\varphi v$ is faithful on $\M_{s(\omega)}$, and $\omega\sim v^{*}\varphi v$ by \cite[Lem.~2.4]{haagerup1990equivalence}. By \cref{lem:u_in_Q}, the ultrapower states $\omega^{\F}$ and $(v^{*}\varphi v)^{\F}$ are unitarily equivalent with respect to ultrapower $(\M_{s(\omega)})^{\F}$. As the modular spectrum of a state and its ultrapower agree \cite[Cor.~4.8 (3)]{ando2014ultraproducts}, we find $\Sp\Delta_{\omega}=\Sp\Delta_{v^{*}\varphi v}$. Since $v:s(\omega)\H\to s(\varphi)\H$ is a unitary that takes $\M_{s(\omega)}$ to $\M_{s(\varphi)}$, the claim follows.
\end{proof}
    
\begin{corollary}\label{cor:error_free}
    If $\omega$ is a faithful embezzling state on $\M$, then $\omega^\F$ admits exact embezzlement in the sense that for all $n\in\NN$ and all faithful states $\psi,\phi\in S(M_n)$, there exists a unitary $u\in M_{n}(\M^\F)= M_n(\M)^\F$ such that 
    \begin{equation}
        \omega^\F\ox\psi =u(\omega^\F\ox\phi)u^*.
    \end{equation}
    \begin{proof}
    This follows from \cref{lem:u_in_Q} because ultrapowers behave well under matrix amplification, i.e., because $(\omega\ox\psi)^\F = \omega^\F\ox\psi$ (using the identification $M_n(\M)^\F= M_n(\M^\F)$).
    \end{proof}    
    \end{corollary}
\begin{proposition}\label{prop:exact_mbz_new}
Let $\omega \in S_*(\M)$ and $\psi\in S(M_n)$ be faithful states and suppose that any of the following assertions is true:
\begin{enumerate}
    \item $u(\omega \ox \frac{1}{n}\!\Tr) u^* = \omega\ox \psi$ for some unitary $u\in M_n(\M)$,
    \item $u\omega u^* = \omega\ox \psi$ for some unitary $u\in M_{n,1}(\M)$.
\end{enumerate}
Then every $\lambda\in \Sp(\Delta_\psi)$ is an eigenvalue of $\Delta_\omega$.
    \begin{proof}
    Let $\rho = \mathrm{diag}(p_1,\ldots, p_n)$ with $p_i>0$ and $\sum p_i=1$, and let $\psi= \tr\rho(\placeholder)$ be the implemented state on $M_n$.
    Then the spectral values of $\Sp(\Delta_\psi)$ are given by the ratios $p_i/p_j$, because $\Delta_\psi = \rho\ox (\overline\rho)^{-1}$ and therefore
    $\Delta_\psi \ket{ij} = p_i/p_j \ket{ij}$.
    Assume that the first item holds and let $u\in M_n(\M)$ be a unitary such that $u(\omega\ox\psi)u^* = \omega \ox\frac1n\!\tr_n$.
    Set $v=j\up n(u)$. Then  $uv(\Delta_{\omega\ox\psi})u^*v^* = uv(\Delta_\omega\ox \Delta_\psi)u^*v^* = \Delta_\omega\ox \Delta_{\frac{1}{n}\!\Tr} = \Delta_\omega \ox 1$. 
    Therefore
    \begin{align}
    (\Delta_\omega\ox1) uv(\Omega\ox\ket{ij}) = \frac{p_i}{p_j}\, uv(\Omega\ox\ket{ij}),
    \end{align}
    because $\Delta_\omega\Omega = \Omega$.
    The claim follows analogously from the second item.
    \end{proof}
\end{proposition}
\begin{corollary}\label{cor:exact-mbz-eigenvalues}
    Let $\omega\in S_*(\M)$ be a state that admits exact embezzlement, in the sense that for all faithful $\varphi,\psi \in S(M_n)$ (for all $n\in\NN)$ there exists a unitary $u\in \U(M_n(\M))$ such that
    \begin{align}
        u\omega\ox\varphi u^* = \omega\ox\psi.
    \end{align}
    Then every $\lambda>0$ is an eigenvalue of the modular operator $\Delta_\omega$.
    \begin{proof}
    Any $\lambda>0$ appears as an eigenvalue of $\Delta_\psi$ for suitable $\psi\in S(M_n)$.
    \end{proof}
\end{corollary}

\begin{corollary}
    Let $\M$ be a von Neumann algebra and let $\omega$ be a state which admits exact embezzlement in the sense of the previous corollary.
    Then $\M$ is not a separable von Neumann algebra, i.e., admits no faithful representation on a separable Hilbert space.
    \begin{proof}
    An uncountable set of distinct eigenvalues implies an uncountable set of (pairwise orthogonal) eigenvectors. Thus the claim follows from the previous corollary.
    \end{proof}
\end{corollary}

\begin{proof}[Proof of \cref{thm:modular_spectrum}]
    By \cite[Cor.~4.8 (3)]{ando2014ultraproducts}, a faithful normal state and its ultrapower have the same modular spectrum, i.e., $\Sp\Delta_\omega=\Sp\Delta_{\omega^\F}$. 
    If $\omega$ is embezzling, then $\omega^\F$ admits exact embezzlement by \cref{cor:error_free} and hence $\Sp\Delta_\omega = \RR^+$ by \cref{cor:exact-mbz-eigenvalues}, which finishes the proof.
\end{proof}

Let us note that non-separable Hilbert spaces also allow for exact embezzling states in the bipartite sense:

\begin{corollary}\label{cor:exact-mbz}
    There exists a standard bipartite system $(\H,\M,\M')$ and a unit vector $\Omega\in\H$ such that for all states $\Psi,\Phi\in\CC^n\ox\CC^n$ with marginals of full support, there exist unitaries $u\in M_n(\M)$ and $u'\in\M_n(\M')$ such that 
    \begin{equation}
        \Omega\ox\Psi = uu'\, \Omega\ox\Phi.
    \end{equation}
    However, this forces $\H$ to be a non-separable Hilbert space.
\end{corollary}

\begin{proof}
    Let $\omega$ be an exact embezzling state as in \cref{cor:exact-mbz-eigenvalues} and let $(\H,J,\P)$ be its standard form.
    Set $\Omega=\Omega_\omega$.
    Let $\psi,\phi\in S(M_n)$ and pick a unitary $u\in M_n(\M)$ such that $\omega\ox\psi=u(\omega\ox\phi)u^*$.
    Without loss of generality, we may assume that $\Psi$ and $\Phi$ are in the positive cone of $\CC^n\ox\CC^n$.
    The claim then result then follows with $u' = J\up n u J\up n$ where $J\up n$ is the modular conjugation of the standard form of $M_n(\M)$ as in \cref{lem:standard_amplification}.
\end{proof}

\section{Embezzlement and the flow of weights}\label{sec:mbz_fow}

\begin{assumption*}[Separability]
    From here on, we only consider separable von Neumann algebras, i.e., von Neumann algebras admitting a faithful representation on a separable Hilbert space $\H$ (see \cref{subsec:pre:hilbert-spaces}).
\end{assumption*}

The flow of weights assigns to a von Neumann algebra $\M$ a dynamical system $(X,\mu,F)$, where $(X,\mu)$ is a standard Borel space and $F=(F_t)_{t\in\RR}$ is a one-parameter group of non-singular Borel transformations, in a canonical way.\footnote{To be precise, the measure $\mu$ is determined by $\M$ only up to equivalence of measures.}
It owes its name to the first construction using equivalence classes of weights on $\M$.
The flow of weights encodes many properties of the algebra $\M$, e.g., $\M$ is a factor if and only if the flow is ergodic.
Furthermore, there is a canonical map taking normal states $\omega$ on $\M$ to absolutely continuous probability measures $P_\omega$ on $X$, which captures exactly the distance of unitary orbits
\begin{equation}\label{eq:dinstance_Uorbits}
    \inf_{u\in\U(\M)} \norm{u\omega_1 u^*-\omega_2} = \norm{P_{\omega_1}-P_{\omega_2}}, \qquad \omega_1,\omega_2\in S_*(\M),
\end{equation}
where the distance of probability measures on $X$ is measured with the $L^1$-distance of their densities with respect to $\mu$.
Additionally, the flow of weights behaves well if $\M$ is replaced by $M_n(\M)$. The probability measure $P_{\omega\ox\psi}$, where $\psi\in S(M_n)$, is the convolution of $P_\omega$ with the spectrum of the density operator $\rho_\psi\in M_n$ along the flow $F$.
These two properties make the flow of weights a perfect tool to study embezzlement!

Before we explain how the flow of weights can be constructed, we recall how it can be used to classify type $\III$ factors \cite[Def.~XII.1.5]{takesaki2}:
Let $\M$ be a type $\III$ factor with flow of weights $(X,\mu,F)$. Then 
\begin{itemize}
    \item $\M$ is type $\III_0$ if the flow of weights is aperiodic, i.e., no $T>0$ exists such that $F_T=\id_X$,
    \item $\M$ is type $\III_\lambda$, $0<\lambda<1$, if the flow of weights is periodic with period $T=-\log\lambda$,
    \item $\M$ is type $\III_1$ if the flow of weights is trivial, i.e., $X=\{*\}$ and, hence, $F_s=\id_X$ for all $s$.
\end{itemize}

To be precise, the periodicity and aperiodicity of $F$ are only required almost everywhere.
Since every nontrivial ergodic flow is either periodic or aperiodic, this definition covers all type $\III$ factors.
Another equivalent way to obtain this classification is through the \emph{diameter of the state space} \cite{connes1985diameters}. 
For this, one considers the quotient of the state space $S_*(\M)$ modulo \emph{approximate unitary equivalence}
\begin{equation}
    \omega\sim\varphi :\iff \forall_{\eps>0}\ \exists_{u\in\U(\M)} : \norm{\omega- u\varphi u^*}<\eps,\qquad\omega,\varphi\in S_*(\M).
\end{equation}
It turns out that a type $\III$ factor $\M$ has type $\III_\lambda$, $0\le\lambda\le1$ if and only if
\begin{equation}\label{eq:diameter}
    \diam(S_*(\M)\big/\!\sim) = 2\frac{1-\sqrt\lambda}{1+\sqrt\lambda},
\end{equation}
where the diameter is measured with the quotient metric
\begin{equation}
    d([\omega],[\varphi])=\inf_{u\in\U(\M)} \norm{u\omega u^*-\varphi},\qquad \omega,\varphi\in S_*(\M).
\end{equation}
\cref{eq:diameter} was shown in \cite{connes1985diameters} for factors with separable predual and extended to the general case in \cite{haagerup1990equivalence}.

We will often identify the flow of weights $(X,\mu,F)$ of a von Neumann algebra $\M$ with the induced one-parameter group of automorphisms $\theta$ on the abelian von Neumann algebra $L^\oo(X,\mu)$, i.e., we identify
\begin{equation}\label{eq:fow_2sides}
    (X,\mu,F) \equiv (L^\oo(X,\mu),\theta),\qquad \theta_s(f)(x) =f(F_s(x)),\quad s\in\RR.
\end{equation}

We now describe a way to construct the flow of weights for a given von Neumann algebra $\M$.
The construction requires the choice of a normal semifinite faithful weight $\phi$ on $\M$ which is used to construct the crossed product
\begin{equation}
    \N = \M \rtimes_{\sigma^\phi} \RR
\end{equation}
of $\M$ by the modular flow $\sigma^\phi$ generated by $\phi$ (see \cite{takesaki2}).
The crossed product $\N$ is generated by an embedding $\pi:\M\to\N$ and a one-parameter group of unitaries $\lambda(t)$, $t\in\RR$, implementing the modular flow
\begin{equation}
    \lambda(t)\pi(x)\lambda(t)^* = \pi(\sigma_t^\phi(x)), \qquad t\in\RR.
\end{equation}
We denote the dual action $\widehat{\sigma^\omega}$ by $\tilde\theta$ (see \cref{subsec:pre:crossed-products}). I.e., $\tilde\theta$ is the $\RR$-action  given by
\begin{equation}
    \tilde\theta_s(\pi(x))=\pi(x)\qandq\tilde\theta_s(\lambda(t)) = e^{-its} \lambda(t).
\end{equation}
In the following we supress the embedding $\pi$ and identify $\M$ and $\pi(\M)\subset \N$.
To every weight $\varphi$ on $\M$, a dual weight $\tilde\varphi$ on $\N$ is associated by averaging over the dual action:
\begin{equation}
    \tilde\varphi(y) = \varphi\bigg(\int_{-\oo}^\oo \tilde\theta_s(y)\,ds\bigg),\qquad y\in\N^+.
\end{equation}
Clearly, the dual weight is invariant under the dual action, i.e., $\tilde\varphi\circ\tilde\theta_s=\tilde\varphi$ for all $s\in\RR$.
The centralizer $\N^{\tilde\theta}$ of the dual action is exactly $\M$, and the relative commutant $\M'\cap\N$ is exactly the center $Z(\N)$ of $\N$.
Let $h\ge0$ be the positive self-adjoint operator affiliated with $\N$ such that $h^{it}=\lambda(t)$ and set
\begin{equation}
    \tau(y) = \tilde\phi(h^{-1/2}yh^{-1/2}),\qquad y\in\N^+.
\end{equation} 
 On $\N$, the modular flow of the dual weight $\tilde\phi$ is implemented by $\lambda(t)$ as $\sigma^{\tilde \phi}_t(y) = \lambda(t)y\lambda(-t)$ for $y\in\N$. In particular, $\sigma^{\tilde\phi}_s(\lambda(t))=\lambda(t)$. Due to $\tilde\theta_s(\lambda(t))=e^{-its}\lambda(t)$,  we further have $\tilde\theta_s(h)=e^{-s}h$.
As a consequence, $\tau$ is a normal semifinite faithful trace on $\N$ which is scaled by the dual action\footnote{This follows from $(D\tau:D\tilde\phi)_t = \lambda(-t)$: Since $\sigma^{\tilde\phi}_s(\lambda(t)) =\lambda(t)$, $h$ is affilated with the centralizer $\N^{\tilde\phi}$. Therefore, \cite[Thm.~2.11]{takesaki2} implies that $\sigma^\tau_t = h^{-it} \sigma^{\tilde\phi}_t(\placeholder)h^{it} = \lambda(-t)\sigma^{\tilde\phi}_t(\placeholder)\lambda(t)$.  We get $\sigma^\tau_s(x\lambda(t)) = \lambda(-s) \sigma^\phi_s(x) \lambda(s) \lambda(t) = x\lambda(t)$, where we used $\sigma^{\tilde \phi}_s(x)=\sigma^\phi_s(x)$ for $x\in\M$. Thus, the modular flow $\sigma^\tau$ is trivial and $\tau$ is a trace.}
\begin{equation}\label{eq:scaling}
    \tau\circ\tilde\theta_s = e^{-s} \tau,\qquad s\in\RR.
\end{equation}
It can be shown that $(\N,\tau,\tilde\theta)$ does not depend, up to isomorphism, on the choice of normal semifinite faithful weight $\phi$ on $\M$ \cite{takesaki2}.\footnote{Since the triple $(\N,\tau,\tilde\theta)$ is unique only up to isomorphism, the scaling of the trace $\tau$ is not canonical, indeed, $(\N,\tau',\tilde\theta)$, where $\tau'=k\cdot \tau$ for any $k>0$, is an isomorphic triple which results from the construction above if the weight $\phi$ is replaced by $\phi'=k\cdot \phi$.}
From this triple, the flow of weights is constructed as follows:
\begin{center}\it
    The flow of weights of $\M$ is the center $Z(\N)$ equipped with \\the restriction $\theta_s = \tilde\theta_{s\,|Z(\N)}$ of the dual action.
\end{center}
Using the correspondence between abelian von Neumann algebras and measure spaces, one can rephrase this as a dynamical system $(X,\mu,F)$ of non-singular transformations $F=(F_s)$ on a standard Borel space $(X,\mu)$ such that $L^\oo(X,\mu)\cong Z(\N)$.
As mentioned earlier, the flow of weights is ergodic, i.e., the fixed point subalgebra algebra of $Z(\N)$ is trivial, if and only if $\M$ is a factor \cite{takesaki2}.

\subsection{The spectral state}\label{sec:spectral_state}

We will now describe how one associates to each normal state $\omega$ on a von Neumann algebra $\M$ a normal state $\hat\omega$ on $Z(\N)$, or, equivalently, an absolutely continuous probability measure $P_\omega$ on $X$ where $L^\oo(X,\mu)=Z(\N)$.
This association was discovered by Haagerup and St\o rmer in \cite{haagerup1990equivalence}, and we will refer to $\hat\omega$ (or $P_\omega$) as the spectral state of $\omega$.

We begin by noting that we can write $h$ (the affiliated positive operator such that $h^{it}=\lambda(t)$) as the Radon-Nikodym derivative $d\tilde\phi/d\tau$ using the trace $\tau$, i.e., $\tilde\phi(y) = \tau(h y)$.
Generalizing this, we associate to a normal, semifinite weight $\varphi$ on $\M$, the positive, self-adjoint operator 
\begin{equation}
    h_\varphi = \frac{d\tilde\varphi}{d\tau}
\end{equation}
affiliated with $\N$ such that $\tilde{\varphi} = \tau(h_{\varphi}\placeholder)$. In particular, this can be done if $\varphi$ is a normal state. One can show that $h_{u\varphi u^*} = uh_\varphi u^*$ for all $u\in\U(\M)$ and \eqref{eq:scaling} translates to
\begin{equation}
    \tilde\theta_s(h_\varphi) = e^{-s} h_\varphi,\qquad s\in\RR.
\end{equation}

\begin{lemma}\label{thm:hat_normalization}
    Let $\omega$ be a positive normal functional on $\M$. Consider the spectral projection $e_\omega=\chi_1(h_\omega)\in\N$, where $\chi_1$ is the indicator function of $(1,\oo)$. Then 
    \begin{equation}
        \tau(e_\omega x) = \omega(x),\qquad x\in\M.
    \end{equation}
\end{lemma}
\begin{proof}[Proof sketch]
    For simplicity, we assume $\omega$ to be faithful. The general case can be proven similarly (see \cite[Lem.~3.1]{haagerup1990equivalence}).
    Recall that $\tilde\omega\circ \tilde\theta_s=\tilde\omega$.
    Setting $g(t)=t^{-1}\chi_1(t)$ for $t\in[0,\infty)$, we have
    \begin{align*}
        \tau(xe_\omega) = \tilde\omega(xg(h_\omega))
        =\omega\bigg(x\int_{-\oo}^\oo \theta_s(g(h_\omega))ds\,\bigg)
        = \omega\bigg(x\int_{-\oo}^\oo g(e^{-s}h_\omega) ds\bigg)
        = \omega(xs(h_\omega) )=\omega(x),
    \end{align*}
    where we used that $\int_{-\oo}^\oo g(e^{-s}t)ds = \chi_{(0,\infty)}(t)$ and that $s(h_\omega)=s(\omega) =1$ if $\omega$ is faithful.
\end{proof}

\begin{definition}[Spectral states, cp.~{\cite{haagerup1990equivalence}}]
    For a normal state $\omega\in S_*(\M)$ on a von Neumann algebra $\M$, its {\bf spectral state} is the normal state $\hat\omega$ on $Z(\N)$ given by
    \begin{equation}\label{eq:hat}
        \hat\omega(z)  = \tau(e_\omega z),\qquad  z\in Z(\N).
    \end{equation}
\end{definition}

For technical reasons, we sometimes need the spectral functional $\hat\varphi$ of a non-normalized positive linear functional $\varphi\in \M_*^+$, also by \eqref{eq:hat}, which satisfies $\norm{\hat\varphi}=\norm\varphi$.
We make the cautionary remark that the mapping $\M_*^+\ni\omega\mapsto\hat\omega\in Z(\M)_*^+$ is not affine and, in fact, not even homogeneous.
Instead, it follows from $h_{\lambda\omega} = \lambda h_\omega = \tilde{\theta}_{-\log\lambda}(h_\omega)$ that 
\begin{equation}\label{eq:hat_for_nonstates}
    \hat{\lambda\omega} = \lambda \hat\omega \circ\theta_{\log\lambda}, \qquad \lambda>0.
\end{equation}
More generally, the left-hand side of \eqref{eq:hat} defines a normal state on the full crossed product $\N$ from which one obtains $\hat\omega$ by restricting to the center. Interestingly, the restriction to $\M\subset\N$ recovers $\omega$.

In a concrete realization $(X,\mu,F)$ of the flow of weights as a standard Borel space with a flow $F$,\footnote{To be precise, a concrete realization of the flow of weights means a triple $(X,\mu,F)$ of a standard measure space $(X,\mu)$ and a one-parameter group.}
we will denote the $\mu$-absolutely continuous probability measures implementing the states $\hat\omega$ by $P_\omega$, i.e., 
\begin{equation}\label{eq:P_omega}
    \hat\omega(z) = \int_X z(x)\,dP_\omega(x),\qquad z\in Z(\N)=L^\oo(X,\mu).
\end{equation}
Despite the map $\omega\mapsto \hat\omega$ not being affine, it is extremely useful as the following result, which is the main theorem of \cite{haagerup1990equivalence}, shows:

\begin{theorem}[{\cite{haagerup1990equivalence}}]\label{thm:distance_of_unitary_orbits}
    Let $\omega_1,\omega_2$ be normal states on a von Neumann algebra $\M$, then
    \begin{equation}
        \inf_{u\in\U(\M)} \norm{u\omega_1 u^*-\omega_2} = \norm{\hat\omega_1-\hat\omega_2} = \norm{P_{\omega_1}-P_{\omega_2}}.
    \end{equation}
In particular, $\hat\omega_1=\hat\omega_2$ if and only if they are approximately unitarily equivalent.
\end{theorem}

The study of the map $\omega\mapsto\hat\omega$ and the flow of weights in general can often be reduced to the case where $\M$ is a factor.
The reason is the following observation from \cite[Sec.~8]{haagerup1990equivalence} (used in the proof of \cref{thm:distance_of_unitary_orbits}):
Let $\M = \int^\oplus_Y \M_y\, d\nu(y)$ be a direct integral representation of von Neumann algebras $\M(y)$ over a measure space $(Y,\nu)$.
Typically, but not always, we will assume that ($\nu$-almost) each $\M_y$ is a factor which implies that $Z(\M) = L^\oo(Y,\nu)$.
The direct integral implies that the triple $(\N,\tilde\theta,\tau)$ associated to $\M$ can be obtained from the triples $(\N_y,\tilde\theta_y,\tau_y)$ of $\M_y$, $y\in Y$, via direct integration
\begin{equation}
    (\N,\tilde\theta,\tau)
    = \int^\oplus_Y (\N_y,\tilde\theta_y,\tau_y)\,d\nu(y) 
\end{equation}
where the direct integral is understood component-wise.
This can be seen easily by constructing the crossed product $\N=\M\rtimes_{\sigma^\phi}\RR$ with a weight $\phi = \int^\oplus_Y\phi_y\,d\nu(y)$ and by using $\phi_y$ to construct $\N_y=\M_y\rtimes_{\sigma^{\phi_y}}\RR$.
With this, it is clear that the flow of weights $(Z(\N),\theta)$ decomposes as a direct integral of $Z((\N_y),\theta_y)$.
Any state $\omega\in S_*(\M)$ decomposes into a direct integral $\omega=\int_Y^\oplus \omega_y\,d\nu(y)$ of positive linear functionals $\omega_y\in \M_*^+$.
The spectral state on the flow of weights is then given by
\begin{equation}
    \hat\omega = \int_Y^\oplus \hat\omega_y d\nu(y).
\end{equation}

To see this, observe that $h_\omega = d\tilde\omega/d\tau = \int^\oplus_Y (d\tilde\omega_y/d\tau_y)\,d\nu(y) = \int_Y^\oplus h_{\omega_y}\,d\nu(y)$ gives $e_\omega = \chi_1(h_\omega)=\int_Y^\oplus \chi_1(h_{\omega_y})d\nu(y)=\int^\oplus_Y e_{\omega_y}d\nu(y)$ from which \eqref{eq:hat_direct_integral} follows directly.
We summarize these findings in the following equation
\begin{equation}\label{eq:FOW_direct_integral}
    (Z(\N),\theta,\,\hat\ \,) = \int_Y^\oplus (Z(\N_y),\theta_y,\,\hat\ \,)\,d\nu(y).
\end{equation}

We remark that another natural decomposition of a state $\omega$ on a direct integral as above is $\omega = \int_Y^\oplus \varphi_y\,p(y)d\nu(y)$ where $p(y)=\omega_y(1)$ and $\varphi_y$ is a measurable state-valued map such that $p(y)\varphi_y = \omega_y$.
It follows that the spectral state on the flow of weights is
\begin{equation}\label{eq:hat_direct_integral}
    \hat\omega= \int_Y^\oplus \hat\varphi_y \circ \theta_{\log p(y)} \,p(y)d\nu(y).
\end{equation}

We continue by discussing the flow of weights and, specifically, the spectral state construction for semifinite von Neumann algebras and Type $\III_{\lambda}$ factors ($0<\lambda<1$).

\subsubsection{Semifinite von Neumann algebras}\label{sec:FOW_semifinite}
Let $(\M,\tr)$ be a semifinite von Neumann algebra. 
Since the modular flow of a trace is trivial, we can identify the triple $(\N,\tau,\tilde\theta)$ as
\begin{equation}\label{eq:triple_semifinite}
    (\N,\tau,\tilde\theta) = \Big(\M\ox L^\oo(\RR),\,\tr\ox \int_{-\oo}^\oo\placeholder\,e^{-\gamma}d\gamma,\,\id \ox\, \theta\Big),\qquad \theta_s \Psi(\gamma)=\Psi(\gamma-s).
\end{equation}
Therefore, the flow of weights of $\M$ is simply $Z(\M)\ox L^\oo(\RR)$ with $\theta$ acting as translation on $\RR$ (cf.\ \cref{ex:groupalgebra}).
In this representation, the unitaries $\lambda(t)$ act by mulitplication with $e^{it\gamma}$. Since $\theta_s$ acts according to \eqref{eq:triple_semifinite}, this is consistent with $\tilde\theta_s(\lambda(t)) = e^{-its}\lambda(t)$ and results from a Fourier transformation.

The dual weight $\tilde\varphi$ of a weight $\varphi$ on $\M$ is
\begin{equation}
    \widetilde\varphi = \varphi\bigg( \int_{-\oo}^\oo \tilde\theta_s(\placeholder)\,ds\bigg) = \varphi \ox \int_{-\oo}^\oo \placeholder\,ds.
\end{equation}
Since the operator $h$ such that $h^{it}=\lambda(t)$ is the multiplication operator $h(\gamma) = e^\gamma$, the trace $\tau$ on $\N$ is indeed given by 
\begin{equation}
    \tau = \widetilde\tr(h^{-1/2}(\placeholder)h^{-1/2}) = \tr \ox \int_{-\oo}^\oo \placeholder e^{-\gamma}d\gamma.
\end{equation}
We see that the dual action indeed scales the trace: $\tau \circ \tilde\theta_{s} = \tr\ox\int \placeholder\,e^{-(\gamma+s)}d\gamma = e^{-s} \tau$.
For a state $\omega\in S_*(\M)$ denote the density operator by $\rho_\omega = (d\omega/d\!\tr)$, it follows 
\begin{equation}
    h_\omega = \frac{d\tilde\omega}{d\tau} = \rho_\omega \ox \exp \equiv \int_\RR^\oplus \rho_\omega\,e^{\gamma}d\gamma.
\end{equation}
The benefit of the direct integral representation is that it lets us compute 
\begin{equation}
    \chi_t(h_\omega) = \int_\RR^\oplus \chi_t(e^{\gamma} \rho_\omega) d\gamma = \int_\RR^\oplus \chi_{te^{-\gamma}}(\rho_\omega)\,d\gamma, \qquad t>0,
\end{equation} 
easily.
We now specialize to the case where $\M$ is a factor.
In this case, we have $Z(\N)=Z(\M)\ox L^\oo(\RR)=L^\oo(\RR)$.
The state $\hat\omega$ on $L^\oo(\RR)$ is thus given by
\begin{equation}\label{eq:hat_omega_semifinite}
    \hat\omega(g) = \tau(\chi_1(h_\omega)g) 
    =\int_{-\oo}^\oo \tr(\chi_{e^{-\gamma}}(\rho_\omega)) g(\gamma) e^{-\gamma}d\gamma
    =\int_0^\oo D_\omega(t) g(-\log t) dt,
\end{equation}
where $D_\omega(t) = \tr(\chi_t(\rho_\omega))$ is the distribution function of $\omega$ (see \cref{sec:preliminaries}).
We further identify $Z(\N)=L^\oo(\RR)$ with $L^\oo(\RR^+)$ via the logarithm.
This gives us the following geometric realization $(X,\mu,F)$ of the flow of weights
\begin{equation}
    X = (0,\oo),\qquad \mu = dt, \qquad F_s(t) = e^{s}t,
\end{equation}
and, by \eqref{eq:hat_omega_semifinite}, $\hat\omega$ is implemented by the probability measure
\begin{equation}
    dP_\omega(t) = D_\omega(t)\,dt
\end{equation}
Using direct integration, as in \eqref{eq:FOW_direct_integral}, we can lift this result to general von Neumann algebras:

\begin{proposition}[Flow of weights for semifinite von Neumann algebras]\label{prop:FOW_semifinite}
    Let $\M$ be a semifinite von Neumann algebra with separable predual and let $\M=\int_Y^\oplus\M_y\,d\nu(y)$ be the direct integral decomposition over the center $Z(\M) = L^\oo(Y,\nu)$.
    Then the flow of weights $(X,\mu,F)$ is $X=Y\times(0,\oo)$, $\mu = \nu\times dt$ and $F_s(y,t) = (y,e^{-s}t)$. The spectral state $\hat\omega$ of a state $\omega=\int^\oplus_Y\omega_yd\nu(y)$ is implemented by the probability measure
    \begin{equation}\label{eq:FOW_semifinite}
        dP_\omega(y,t) = D_{\omega_y}(t)\,d\nu(y)\,dt.
    \end{equation}
\end{proposition}

We remark that the cumulative distribution function $D_{\omega_y}(t) = \tr_y \chi_t(\rho_{\omega_y})$ is not the distribution function of a state but of a subnormalized positive linear function $\omega_y\in (\M_y)_*^+$. 
If $\M = L^\oo(Y,\nu)$ is abelian, we have $\omega_y=p(y)$ is and $D_{\omega_y}(t) = \chi_t(p(y))$, so that
\begin{equation}
    dP_\omega(y,t) = \chi_{[0,p(y))}(t) \,d\nu(y) dt.
\end{equation}
Thus, the flow of weights of $L^\oo(Y,\nu)$ is the space $Y\times (0,\oo)$ with the flow $F_s(y,t)=(y,e^{-s}t)$ and the spectral state of a probability measure on $Y$ is the uniform distribution on the area under the graph of its probability density function.

\subsubsection{Type III{$_{\lambda}$} factors}\label{sec:FOW_type_III_lambda}

Let $\M$ be a factor of type $\III_\lambda$, $0<\lambda<1$. We will follow the construction of the flow of weights for type $\III_\lambda$ factors in \cite[Sec.~5]{haagerup1990equivalence}.
For this, we need to pick a generalized trace\footnote{Generalized traces were introduced on type $\III_\lambda$ factors by Connes in \cite[Sec.~4]{connes1973classIII}.}, i.e., a normal strictly semifinite faithful weight $\phi$ with $\phi(1)=\oo$, satisfying one, hence, all, of the following equivalent properties (see, \cite[Thm.~4.2.6]{connes1973classIII}):
\begin{enumerate}[(a)]
    \item $\sigma_{t_0}^\phi=\id$, where $t_0 = \frac{2\pi}{-\log\lambda}$,
    \item $\Sp(\Delta_\phi) = \{0,\lambda^n : n\in\ZZ\}$,
    \item the centralizer $\M^\phi=\{ x \in \M : \sigma^\phi_t(x)=x\ \forall t\ge0\}\subset\M$ is a (type $\II_\oo$) factor.
\end{enumerate}
We briefly explain how a generalized trace $\phi$ can be constructed from the so-called discrete decomposition of a type $\III_\lambda$ factor \cite[Thm.~VII.2.1]{takesaki2}. The discrete decomposition allows one to write the factor $\M$ as the crossed product
\begin{equation}\label{eq:discrete_decomposition}
    \M = \P \rtimes_\alpha \ZZ
\end{equation}
where $\P$ is a type $\II_\oo$ factor and $\alpha$ is a $\ZZ$-action on $\P$ which scales the trace $\tr$ on $\P$ by $\tr\circ\alpha_n = \lambda^n\,\tr$, $n\in\ZZ$.
The generalized trace $\phi$ is now obtained as the normal, semifinite, faithful weight dual to the trace $\tr$ on $\P$, i.e., $\phi=\tilde\tr$.
This is indeed a generalized trace because $\sigma^\phi_{t}$ is $t_0$-periodic by construction (we identify the dual group $\hat\ZZ$ with $\RR/t_0\ZZ$).

With $\phi$ being a generalized trace, we have a periodic modular flow so that we may consider the crossed product
\begin{equation}
    \N_0=\M\rtimes_{\sigma^\phi} (\RR/t_0\ZZ) 
\end{equation}
in addition to the full crossed product $\N=\M\rtimes_{\sigma^\phi}\RR$.
We realize $\N_0$ on $L^2(\RR/t_0\ZZ,\H)\equiv\int_{[0,t_0)}^\oplus\H\,dt$ where it is the von Neumann algebra generated by the operators $\pi_0(x)$, $x\in\M$ and $\lambda_0(t)$, $t\in\RR/t_0\ZZ$, given by
\begin{equation}
     \pi_0(x)\Psi(s) = \sigma^\phi_{-s}(x)\Psi(s)\quad\text{and}\quad \lambda_0(t)\Psi(s) = \xi(s-t),\qquad \Psi\in L^2(\RR/t_0\ZZ,\H).
\end{equation}
Let $h_0$ be the self-adjoint positive operator affiliated with $\N_0$ such that $h_0^{it}=\lambda_0(t)$.
Denote by $\theta_0\in\Aut(\N_0)$ the automorphism generating the $\ZZ$-action dual to $\sigma^\phi$, which is determined by
\begin{equation}
    \begin{aligned}
        \theta_0(\pi_0(x)\lambda_0(t)) = e^{-i\gamma_0t} \pi_0(x) \lambda_0(t), \qquad \gamma_0:=-\log\lambda \equiv \frac{2\pi}{t_0}.
    \end{aligned}
\end{equation}
In particular, it follows that $\theta_0(h_0)=\lambda h_0$.
The weight $\bar\varphi$ on $\N_0$ dual to a weight $\varphi$ on $\M$ is
\begin{equation}
    \bar\varphi(a) = \varphi\bigg(\sum_{n\in\ZZ}\theta_0^n(a)\bigg),\qquad a\in \N_0^+.
\end{equation}
We get a faithful normal trace $\tau_0$ on $\N_0$ from 
\begin{equation}\label{eqref:trace-discrete}
    \tau_0(a) = \bar\phi(h_0^{-1/2}ah_0^{-1/2}),\qquad a\in\N_0^+,
\end{equation}
or in other words $d\bar\phi/d\tau_{0}=h_{0}$.
Crucially, this trace is scaled by the dual action: $\tau_0\circ\theta_0=\lambda\tau_0$.

It is shown in \cite[Prop.~5.6]{haagerup1990equivalence} that the full crossed product $\N=\M\rtimes_{\sigma^\phi}\RR$ can be identified with $\N_0\ox L^\oo[0,\gamma_0)$ via
\begin{equation}
    \pi(x) = \pi_0(x)\ox 1,\quad\text{and}\quad \lambda(t) = \lambda_0(t) \ox e^{it}, \quad t\in\RR,
\end{equation}
where it is understood that $\lambda_0(t)$, $t\in\RR$, means $\lambda_0([t])$ with $[t]$ denoting the equivalence class in $\RR/t_0 \ZZ$.\footnote{The dual action $\tilde\theta$ on $\N$ can be spelled out explicitly using the automorphism $\theta_0$ and the translations mod $\gamma_0$ on $L^\oo[0,\gamma_0]$. As it becomes a bit cumbersome to write down and will not be needed in the following, we refer the interested reader to the proof of \cite[Prop.~5.7]{haagerup1990equivalence}.}
Since $\N_0$ is a factor, the center $Z(\N)$ is exactly $L^\oo[0,\gamma_0)$.
With this identification, the dual weights are given by  
\begin{equation}
   \tilde\varphi = \bar\varphi\ox \int_0^{\gamma_0}\placeholder\,d\gamma,
\end{equation}
and we have
\begin{equation}
    \tau = \tau_0 \ox\int_0^{\gamma_0} \placeholder\,e^{-\gamma}d\gamma, \qandq 
    h_\varphi:= \frac{d\tilde\varphi}{d\tau}= \frac{d\bar\varphi}{d\tau_0} \ox \exp,
\end{equation}
where $\bar{h}_{\varphi}:=\tfrac{d\bar\varphi}{d\tau_0}$ is a positive, self-adjoint operator affiliated with $\N_{0}$.
The flow of weights, i.e., the restriction of the dual action $\tilde\theta_s$ to the center $Z(\N)=L^\oo[0,\gamma_0)$, is given by translations
\begin{equation}
    \theta_sz(\gamma) = z(\gamma-s), \qquad s\in\RR,\ z\in Z(\N)=L^\oo[0,\gamma_0),
\end{equation}
where the subtraction is modulo $\gamma_0$.
To obtain an explicit description of the spectral states, we follow \cite[Sec.~5]{haagerup1990equivalence} and use the crossed product $\N_0$ to define for each normal state $\omega\in S_*(\M)$, a function $f_\omega:(0,\oo)\to(0,\oo)$ by
\begin{equation}\label{eq:f}
    f_\omega(t) = \tau_0(\chi_t(\bar{h}_{\omega})),\qquad t>0.
\end{equation}
The spectral state $\hat{\omega}$ on the flow of weights $Z(\N)=L^\oo[0,\gamma_0)$ is given by
\begin{align}
    \hat\omega(z) 
    =\tau(z \chi_1(h_\omega)) 
    = \int_0^{\gamma_0} z(\gamma) \tau_0(\chi_1(e^\gamma \bar{h}_\omega))\,e^{-\gamma}d\gamma 
    = \int_0^{\gamma_0} z(\gamma) f_{\omega}(e^{-\gamma}) e^{-\gamma}d\gamma.
\end{align}

\begin{definition}\label{def:spectral_distribution}
    The function $f_{\omega}$ defined by \cref{eq:f} is called the {\bf spectral distribution function} of the normal state $\omega$.
\end{definition}

We conclude:

\begin{proposition}[Flow of weights for $\III_\lambda$ factors]\label{prop:FOW_of_type_III_lambda}
    Let $\M$ be a type $\III_\lambda$, $0<\lambda<1$ factor. Then the flow of weights $(X,\mu,F)$ is given by the periodic translations $F_s(\gamma) = \gamma-s$ mod $\gamma_0$ on $X= [0,\gamma_0)$ equipped with the Lebesgue measure $\mu=dt$ and the spectral state $\hat\omega$ of $\omega\in S_*(\M)$ is implemented by the probability measure
    \begin{equation}\label{eq:FOW_of_type_III_lambda}
        dP_\omega(t) = f_\omega(e^{-t})e^{-t} dt,
    \end{equation}
    associated with the spectral distribution function $f_{\omega}$ of $\omega$.
\end{proposition}

We collect some properties of the spectral distribution functions, which are analogous to those of distribution functions of states on semifinite von Neumann algebras (cp.~\cref{prop:spectral_distance}).
By \eqref{eq:FOW_of_type_III_lambda}, these directly transfer to properties of the probability measures $P_\omega$, $\omega\in S_*(\M)$.

\begin{lemma}[{\cite[Lem.~5.1, Thm.~5.5]{haagerup1990equivalence}}]\label{lem:f}
    For every $\omega\in S_*(\M)$, the spectral distribution function $f_\omega$ is a right-continuous, non-increasing and satisfies
    \begin{equation}\label{eq:f_props}
        f_\omega(\lambda t) = \tfrac1\lambda f_\omega(t),\quad t>0, \qandq 
        \int_\lambda^1 f_\omega(t) dt =1.
    \end{equation}
    Conversely, if $f:(0,\oo)\to(0,\oo)$ is right-continuous, non-increasing and satisfies \eqref{eq:f_props} then $f= f_\omega$ for some $\omega\in S_*(\M)$.
    Furthermore, it holds that
\begin{equation}\label{eq:hat_distance_with_f}
    \norm{\hat\omega_1-\hat\omega_2}
    = \norm{P_{\omega_1}-P_{\omega_2}} = \int_\lambda^1 \abs{f_{\omega_1}(t)-f_{\omega_2}(t)} dt.
\end{equation}
\end{lemma}

We conclude this subsection by analyzing the behavior of the spectral distribution function under a change of the generalized trace $\phi$.

\begin{lemma}
\label{lem:f_invariance}
Let $\phi,\psi$ be generalized traces on a type $\III_{\lambda}$ factor $\M$. If $\omega\in S_{*}(\M)$ is normal state with associated spectral distribution functions $f\up\phi_{\omega}$ and $f\up\psi_{\omega}$ respectively, then there exists a unique $\mu\in(\lambda,1]$ such that
\begin{equation}
    \label{eq:f_invariance}
    f\up\psi_{\omega}(t) = \mu f\up\phi_{\omega}(\mu t) \qquad \forall t>0.
\end{equation}
\end{lemma}
\begin{proof}
    We consider the crossed products $\N\up\phi_{0}$ and $\N\up\psi_{0}$ relative to the two generalized traces.
    By \cite[Thm.~XII.2.2]{takesaki2}, $\psi$ is unitarily equivalent to $\mu \phi$ for a unique $\mu\in (\lambda,1]$, i.e., there exists a unitary $u\in\M$ such that $\psi = \mu\phi_{u}$ for  $\phi_{u}=u\phi u^{*}$.
    
    Let us first analyze the effect of unitarily rotating the generalized trace $\psi$ by $u\in\U(\M)$.
    Following \cite[Prop.~3.5]{takesaki_duality_1973}, there is a *-isomorphism $\chi_{0}:\N\up{\phi_{u}}_{0}\to\N\up\phi_{0}$ induced by the Connes cocycle derivative $(D\phi_{u}:D\phi)_{t} = u\sigma^{\phi}_{t}(u^{*})$ because the latter is $t_{0}$-periodic. $\chi_{0}$ is explicitly given in terms of the generators of the crossed products by
    \begin{equation*}
        \chi_{0}(\pi\up{\phi_{u}}_{0}(a)) = \pi\up\phi_{0}(a), \qquad \chi_{0}(\lambda_{0}(t)) = \pi\up\phi_{0}((D\phi_{u}:D\phi)_{t})\lambda_{0}(t),
    \end{equation*}
    for all $a\in\M$ and $t\in\RR/t_{0}\ZZ$. In analogy with the reasoning in \cite[Sec.~3]{haagerup1990equivalence}, we find $\tau\up\phi_{0}\circ\chi_{0}=\tau\up{\phi_{u}}_{0}$ and $\bar{h}\up\phi_{\omega} = \chi_{0}(\bar{h}\up{\phi_{u}}_{\omega})$. Therefore, we have
    \begin{equation*}
        f\up{\phi_{u}}_{\omega} = f\up\phi_{\omega}.
    \end{equation*}
    As a second step, we analyze the effect of rescaling $\phi$ by $\mu$.
    Since the Connes cocycle derivative is given by $(D(\mu\phi):D\phi)_{t} = \mu^{it}$, the modular flow is scale-invariant:
    \begin{equation*}
    \sigma^{\mu\phi}_{t}(a) = (D(\mu\phi):D\phi)_{t}\sigma^{\phi}_{t}(a)(D(\mu\phi):D\phi)_{t}^{*} = \sigma^{\phi}_{t}(a), \qquad \forall a\in\M,\ t\in\RR/t_{0}\ZZ.
    \end{equation*}
    Therefore, the constructions of $\N\up\phi_{0}$ and $\N\up{\mu\phi}_{0}$ as well as the associated dual actions $\theta\up\phi_{0}$ and $\theta\up{\phi_{u}}_{0}$ agree, but the canonically associated traces differ by the scaling factor $\mu$ (cf.~\cref{eqref:trace-discrete}):
    \begin{equation*}
        \tau\up{\mu\phi}_{0} = \mu\tau\up\phi_{0}.
    \end{equation*}
    As direct consequence, we find $\bar{h}\up{\mu\phi}_{\omega} = \mu^{-1}\bar{h}\up\phi_{\omega}$, which leads to
    \begin{equation*}
        f_\omega \up{\mu\phi} = \mu f_\omega\up\phi(\mu t).
    \end{equation*}
    This implies the claim because $f_\omega\up{\psi}(t) = f_\omega\up{\mu\phi_{u}}(t) = \mu f_\omega\up{\phi_{u}}(\mu t) = \mu f_\omega\up\phi(\mu t)$ for all $t>0$.
\end{proof}

\subsection{The flow of weights and semifinite amplifications}\label{sec:fow_semifinite}

The goal of this subsection is to understand the flow of weight of $M_n(\M)$ in terms of the flow of weights on $\M$ and to obtain formulae for the spectral states of product states.
For later usage, we will consider the more general case of the flow of weights on $\M\ox\P$ where $\P$ is any semifinite factor.

Let $\P$ be a semifinite factor with trace $\tr$.
Let us briefly go through the construction of the flow of weights on $\M\ox\P$.
Given a normal semifinite weight $\phi$ on $\M$, we consider on $\M\up1=\M\ox\P$ the weight $\phi\up1=\phi\ox\tr$.
Then $\sigma^{\phi\up1} = \sigma^\phi\ox\id$ and, hence,
\begin{equation}
    \N\up1 = \M\up1\rtimes_{\sigma^{\phi\up1}}\RR = (\M\rtimes_{\sigma^\phi}\RR)\ox\P=\N\ox\P.
\end{equation}
The dual action and the trace on $\N\up1$ are given by $\tilde\theta\up1 = \tilde\theta\ox \id_\P$ and $\tau\up1=\tau\ox\tr$.
Let us summarize this as
\begin{equation}
    (\N\up1,\tau\up1,\tilde\theta\up1) = (\N\ox\P,\tau\ox\tr,\tilde\theta\ox\id).
\end{equation}
Since $\P$ is assumed to be a factor, it follows that the flow of weights of $\M\ox\P$ and $\M$ are equal
\begin{equation}
    (Z(\N\up1),\theta\up1) = (Z(\N),\theta).
\end{equation}
It is natural to ask how the spectral state of a product state $\omega\ox\psi$ relates to the spectral states of $\omega$ and $\psi$.
With the notion of distribution functions and spectral scales of states on semifinite von Neumann algebras, we have:

\begin{proposition}\label{prop:dual_state}
    Let $\omega$ be a normal state on a von Neumann algebra $\M$ and let $\psi$ be a normal state on a semifinite factor $\P$ with trace $\tr$.
    Then 
    \begin{itemize}
    \item[1.] 
    \begin{align}\label{eq:semifinite_conv}
        (\omega\ox\psi)^\wedge & = \tau((D_{\psi}\ast p_{\omega})(1)\placeholder) = \tau(D_{\psi}(h_{\omega}^{-1})\,\placeholder\,),
    \end{align}
    which is the analog of the convolution formula \cref{eq:dist_convolution}. Here, $D_{\psi}$ denotes the distribution function of $\psi$.
    \item[2.] 
    \begin{align}\label{eq:functional_calc}
        (\omega\ox\psi)^\wedge & = \int_0^\oo \lambda_\psi(t)\,\hat\omega\circ \theta_{\log\lambda_\psi(t)}\, dt =  \int_0^\oo \widehat{\lambda_\psi(t)\omega}\, dt,
    \end{align} 
    where $\lambda_{\psi}$ is the spectral scale of $\psi$. 
    \end{itemize}
\end{proposition}
Both sides of equation \eqref{eq:functional_calc} depend on the scaling of the trace $\tr$ on $\P$. The left-hand side depends on the trace via the construction of the spectral states (see the proof), and the right-hand side depends on the trace through the spectral scale $\lambda_\omega$.
Of course, we pick the same scaling on both sides.

\begin{proof}
We denote the density operator of $\psi$ by $\rho_\psi$ and its spectral measure by $p_\psi$.
With the above construction of the crossed products $\N_1$ and $\N$, the dual weight is
\begin{equation}
    \widetilde{\omega\ox\psi} = \tilde\omega\ox\psi
\end{equation}
and its Radon-Nikodym derivative with respect to $\tau\up1=\tau\ox\tr$ is
\begin{align}
    h_{\omega\ox\psi} := \frac{d\widetilde{\omega\ox\psi}}{d\tau\up1} & = h_\omega \ox \rho_\psi = \int_{\RR^{+}}dp_{\omega}(\mu)\ox\mu\,\rho_{\psi} \nonumber \\
    & = \int_{\RR^+} \lambda \,h_\omega \ox dp_\psi(\lambda) = \int_{\RR^+} \tilde\theta_{-\log\lambda}(h_\omega)\ox dp_\psi(\lambda),
\end{align}
where $h_\omega = d\tilde\omega/d\tau$, and $p_{\omega}$ is its spectral measure.
Therefore, we find
\begin{align}
    e_{\omega\ox\psi} := \chi_1(h_{\omega\ox\psi}) & = \int_{\RR^{+}}dp_{\omega}(\mu)\ox\chi_{1}(\mu\,\rho_{\psi}) = \int_{\RR^{+}}dp_{\omega}(\mu)\ox\chi_{\mu^{-1}}(\rho_{\psi})\nonumber \\
    & = \int_{\RR^+} \chi_1(\tilde\theta_{-\log\lambda}(h_\omega)) \ox dp_\psi(\lambda) = \int_{\RR^+} \tilde\theta_{-\log\lambda}(e_\omega) \ox dp_\psi(\lambda), 
\end{align}
where $e_\omega = \chi_1(h_\omega)$. Now, let $z\in Z(\N)$. Then \cref{eq:semifinite_conv} follows
\begin{align}
    (\omega\ox\psi)^\wedge(z) = \tau\up1(z e_{\omega\ox\psi}) & = \int_{\RR^{+}}(\tau\ox\Tr)\Big(dp_{\omega}(\mu)z\ox\chi_{\mu^{-1}}(\rho_{\psi})\Big)\nonumber \\
    & = \int_{\RR^{+}}\tau(dp_{\omega}(\mu)z)\Tr(\chi_{\mu^{-1}}(\rho_{\psi})) \nonumber \\
    & = \int_{\RR^{+}}\tau(dp_{\omega}(\mu)z)D_{\psi}(\mu^{-1}) \nonumber \\ 
    & = \tau\Big(\Big(\int_{\RR^{+}}dp_{\omega}(\mu) D_{\psi}(\mu^{-1})\Big)z\Big) \nonumber \\
    & = \tau((D_{\psi}\ast p_{\omega})(1)z) = \tau(D_{\psi}(h_{\omega}^{-1})z),
\intertext{and similarly \cref{eq:functional_calc}}
    (\omega\ox\psi)^\wedge(z) = \tau\up1(z e_{\omega\ox\psi}) 
    &= \int_{\RR^+}(\tau\ox\tr)\Big(z \tilde\theta_{-\log\lambda}(e_\omega)\ox dp_\psi(\lambda)\Big)\nonumber\\
    &= \int_{\RR^+} (\tau\circ\theta_{-\log\lambda})(\theta_{\log\lambda}(z)e_\omega) \, \tr dp_\psi(\lambda) \nonumber \\
    &= \int_{\RR^+} \lambda \hat\omega\circ \theta_{\log\lambda}(z) \, \tr dp_\psi(\lambda) = \int_{\RR^+} \widehat{\lambda\omega}(z) \, \tr dp_\psi(\lambda)\nonumber\\
    &= \int_{\RR^+} \lambda_\psi(t) \,\hat\omega\circ\theta_{\log\lambda_\psi(t)}(z)\,dt = \int_{\RR^+} \widehat{\lambda_\psi(t)\omega}(z)\,dt.
\end{align}
where we used \cref{eq:hat_for_nonstates} and, in the last line, the fact that $\tr dp_\psi(\lambda)$ is the push-forward of the Lebesgue measure $dt$ along the spectral scale $\lambda_\psi$ by \cref{eq:trace_formula}. 
\end{proof}

\begin{corollary}\label{cor:projections}
    Let $\omega$ be a normal state on a von Neumann algebra $\M$ and let $\P$ be a semifinite factor with trace $\tr$.
    For a finite projection $p\in \proj(\P)$ consider the normal state $\pi = (\tr p)^{-1} \tr p(\placeholder)$ on $\P$.
    Then
    \begin{equation}
        (\omega\ox\pi)^\wedge = \hat\omega \circ \theta_{\log(\tr p)}.
    \end{equation}
\end{corollary}
\begin{proof}
    Since the spectral scale of $\pi$ is given by $\lambda_\pi(t) = (\tr p)^{-1} \chi_{[0,\tr p)}(t)$, the result follows from the previous one.
\end{proof}

Since we are mostly interested in the case where $\P=M_n$, we explicitly state the formulae for matrix algebras in terms of eigenvalues:

\begin{corollary}\label{lem:dual_state_Mn}
    Let $\omega$ be a normal state on a von Neumann algebra $\M$ and $\psi$ be a state on $M_n$.
    Then
    \begin{equation}\label{eq:dual_state_Mn}
        (\omega\ox\psi)^\wedge = \sum_i p_i \,\hat\omega\circ\theta_{\log p_i},
    \end{equation}
    where $(p_i)$ are the eigenvalues of the density operator of $\psi$ (repeated according to their multiplicity).
    In particular, we have the following special cases:
    \begin{equation}
        (\omega\ox\bra1\placeholder\ket1)^\wedge=\hat\omega, \qquad (\omega\ox\tfrac1n\!\tr)^\wedge= \hat\omega\circ\theta_{-\log n}.
    \end{equation}
\end{corollary}

\begin{proof}
    This follows immediately from \cref{prop:dual_state} and \cref{exa:M_n_spectrum}.
\end{proof}

If $\M$ is semifinite, we recover the convolution formula \cref{eq:dist_convolution} for the distribution function $D_{\omega\ox\psi}$ from \cref{prop:dual_state}, and the transformation formulas \cref{eq:D_amplification,eq:l_amplification} are implied by \cref{lem:dual_state_Mn}.

If $\M$ is a type $\III_\lambda$ factor, we can derive analogous formulas for the spectral distribution function $f_{\omega\ox\psi}$ since $\M\ox\P$ is again a $\III_{\lambda}$ factor \cite[Prop.~28.4]{stratila2020modular}. We may pick a generalized trace $\phi$ on $\M$ and choose $\phi\ox\Tr$ as a generalized trace on $\M\ox\P$.
Similar to \cref{prop:dual_state}, we find:
\begin{align}\label{eq:spectral_dist_conv_1}
f_{\omega\ox\psi}(t) & = (\tau_{0}\ox\Tr)(\chi_{t}(\bar{h}_{\omega\ox\psi})) =  (\tau_{0}\ox\Tr)(\chi_{t}(\bar{h}_{\omega}\ox\rho_{\psi})) \nonumber \\
& = \int_{\RR^{+}}\!\!\!\tau_{0}(\chi_{t}(s\, \bar{h}_{\omega}))\,\Tr dp_{\psi}(s) = \int_{\RR^{+}}\!\!\!f_{\omega}(s^{-1}t)\,\Tr dp_{\psi}(s) \nonumber \\
& = \Tr((f_{\omega}\!\ast\!p_{\psi})(t)) = \Tr(f_{\omega}(t\,\rho_{\psi}^{-1})),
\intertext{and equivalently}
f_{\omega\ox\psi}(t) & = \!\int_{\RR^{+}}\!\!\!\tau_{0}(d\bar{p}_{\omega}(\mu))\Tr(\chi_{t}(\mu\,\rho_{\psi})) = \!\int_{\RR^{+}}\!\!\!\tau_{0}(d\bar{p}_{\omega}(\mu))D_{\psi}(\mu^{-1}t)\nonumber \\ \label{eq:spectral_dist_conv_2}
& = \tau_{0}((D_{\psi}\!\ast\!\bar{p}_{\omega})(t)) = \tau_{0}(D_{\psi}(t\bar{h}_{\omega}^{-1})),
\end{align}
where $\bar{p}_{\omega}$ and $p_{\psi}$ are the spectral measures of $\bar{h}_{\omega}$ and $\rho_{\psi}$. By \cref{eq:functional_calc}, we may write \cref{eq:spectral_dist_conv_1} in terms of the spectral scale $\lambda_{\psi}$ of $\psi$:
\begin{align}\label{eq:spectral_dist_conv_3}
    f_{\omega\ox\psi}(t) & = \int_{\RR^{+}}f_{\omega}(\lambda_{\psi}(s)^{-1}t)ds.
\end{align}
It follows immediately that the spectral distribution function $f_{\omega}$ satisfies the conditions of \cref{thm:spectral_charac}. Due to \cref{lem:f} we conclude:
\begin{corollary}\label{cor:mbz_f}
    Let $\M$ be a type $\III_{\lambda}$ factor and $\omega\in S_{*}(\M)$ a normal state. Then, $\omega$ is embezzling if and only if $f_{\omega}(t)\propto\tfrac{1}{t}$ or, equivalently, $dP_{\omega}(t) \propto dt$ (i.e.~$P_{\omega}$ is translation invariant). 
\end{corollary}

\subsection{Quantification of embezzlement}\label{sec:quantification}

To quantify how good a state $\omega$ is at the task of embezzling, we define
\begin{equation}\label{eq:kappa_def}
    \kappa(\omega) =  \sup_{\psi,\,\phi}\, \inf_{u}\, \norm{\omega\ox \psi - u(\omega\ox\phi)u^*},
\end{equation}
where the supremum is over all states $\psi,\phi$ on $M_n$ (and over all $n\in\NN$) and where the infimum is over all unitaries $u\in M_n(\M)$. 

\begin{theorem}\label{thm:kappa}
    Let $\omega$ be a normal state on a von Neumann algebra $\M$. 
    Then
    \begin{equation}
        \kappa(\omega) = \sup_{s\in\RR}\, \norm{\hat\omega\circ\theta_s - \hat\omega}.
    \end{equation}
\end{theorem}

\begin{proof}
    We denote the right-hand side by $\nu(\omega)$. By \cref{thm:distance_of_unitary_orbits}, we have $\kappa(\omega) = \sup_{\psi,\phi} \norm{(\omega\ox\psi)^\wedge-(\omega\ox\phi)^\wedge}$.
    Fix some $n\in\NN$. Let $m\le n$, let $p_m\in M_n$ be an $m$-dimensional projection and set $\pi_m = \frac1m\tr(p_m\placeholder)$.
    Note that $\pi_n=\frac1n\!\tr$.
    By \cref{lem:dual_state_Mn}, we have
    \begin{align*}
        \kappa(\omega) \ge \norm{(\omega\ox\pi_n)^\wedge-(\omega\ox\pi_m)^\wedge} = \norm{\hat\omega\circ\theta_{-\log n}-\hat\omega\circ\theta_{-\log m}}
        =\norm{\hat\omega\circ\theta_{\log\frac mn}-\hat\omega}.
    \end{align*}
    Since $\log(\mathbb Q^+)$ is dense in $\RR$ and since $\theta$ is continuous, we obtain $\nu(\omega)\le\kappa(\omega)$ by taking the supremum over all $n,m$.
    Conversely, let $\psi,\phi$ be states on $M_n$ with eigenvalues $(p_i)$ and $(q_i)$, respectively (repeated according to their multiplicity). Then
    \begin{align*}
        \norm{(\omega\ox\psi)^\wedge-(\omega\ox\phi)^\wedge}
        &= \norm{\sum_ip_i\hat\omega\circ\theta_{\log p_i}-\sum_j q_j\hat\omega\circ\theta_{\log q_j}}\\
        &= \norm{\sum_{ij} p_iq_j (\hat\omega\circ\theta_{\log p_i}-\hat\omega\circ\theta_{\log q_j})}\\
        &\le \sum_{ij} p_iq_j\norm{\hat\omega\circ\theta_{\log p_i}-\hat\omega_{\log q_j}}\\
        &= \sum_{ij}p_iq_j \norm{\hat\omega \circ \theta_{\log\frac{p_j}{q_i}}-\hat\omega}
        \le \sum_{ij} p_iq_j \nu(\omega) = \nu(\omega).
    \end{align*}
    Since $\psi,\phi$ were arbitrary, this shows $\nu(\omega)\ge \kappa(\omega)$.
\end{proof}

\begin{corollary}\label{cor:invariant_iff_mbz}
    A normal state $\omega$ on a von Neumann algebra $\M$ is embezzling if and only if $\hat\omega$ is invariant under the flow of weights, i.e., $\hat\omega\circ\theta_s=\hat\omega$ for all $s\in\RR$.
    Conversely, if $Z(\N)$ has a normal state $\chi$ that is invariant under the flow of weights, then there exists an embezzling state $\omega\in S_*(\M)$ with $\chi=\hat\omega$.
\end{corollary}

\begin{proof}
    The first statement is a direct consequence of \cref{thm:kappa}. 
    We can decompose $\M = \M_0 \oplus \M_1$ with $\M_0$ semifinite and $\M_1$ type $\III$.
    By \eqref{eq:FOW_direct_integral}, the flow of weights of $\M$ is the direct sum of the flow of weights of $\M_0$ and $\M_1$.
    By \cref{prop:FOW_semifinite}, the flow of weights of $\M_0$ does not admit an invariant state.
    Hence, an invariant state $\chi\in S_*(Z(\N))=S_*(Z(\N_0)\oplus Z(\N_1))$ is of the form $\chi=0\oplus\chi_1$.
    Since $\theta_1$-invariance trivially implies the inequality $\chi_1\circ(\theta_1)_s \ge e^{-s} \chi_1$, $s>0$, the main theorem of \cite{haagerup1990equivalence} implies that there exists an $\omega_1\in S_*(\M_1)$ such that $\hat\omega_1=\chi_1$.
    Setting $\omega=0\oplus\omega_1\in S_*(\M)$, we obtain a state whose spectral state is invariant under the flow of weights. Thus, $\omega$ is embezzling.
\end{proof}

\begin{corollary}\label{cor:embezzlers_are_unitarily_equivalent}
    If $\M$ is a factor, any two embezzling states $\omega_1,\omega_2$ are (approximately) unitarily equivalent.
\end{corollary}

\begin{proof}
    Since $\M$ is a factor, the flow of weights $(Z(\N),\theta)$ is ergodic. Hence, $S_*(Z(\N))$ contains at most one state which is invariant under the flow of weights. 
    Therefore, $Z(\N)$ can have at most one invariant normal state.
    By \cref{thm:distance_of_unitary_orbits}, this implies the claim: if $\omega_1$ and $\omega_2$ are embezzling, then they have the same spectral state on the flow of weights and, hence, are unitarily equivalent.
    Therefore, $\hat\omega_1=\hat\omega_2$ if both $\omega_1$ and $\omega_2$ are embezzling states.
\end{proof}

\begin{corollary}\label{cor:embezzling_needs_typeIII}
    If $\M$ is a von Neumann algebra and $\omega\in S_*(\M)$ is an embezzling state, then $s(\omega)\M s(\omega)$ is a type $\III$ von Neumann algebra.
\end{corollary}

\begin{proof}
    Let $\M=\M_0\oplus\M_1$ be the direct sum decomposition into a semifinite and a type $\III$ von Neumann algebra. 
    As argued in \cref{cor:invariant_iff_mbz}, an embezzling state $\omega$ is of the form $\omega=0\oplus\omega_1$ with $\omega_1\in S_*(\M)$ being embezzling. Since corners of type $\III$ algebras are type $\III$, the result follows.
\end{proof}

In particular, there are no embezzling states on semifinite von Neumann algebras.
This can also be seen from the explicit realization of the flow of weights of a semifinite von Neumann algebra $\M$ (see \cref{sec:FOW_semifinite}), which shows that the flow of weights cannot admit an invariant state.
With this approach, we can show that more is true:
Not only do semifinite von Neumann algebras admit no embezzling states, they also do not admit any form of approximate embezzlement:

\begin{corollary}\label{cor:kappa_semifinite}
    If $\M$ is a semifinite von Neumann algebra with separable predual, we have $\kappa(\omega) = 2$ for all normal states $\omega$ on $\M$.
\end{corollary}

\begin{proof}
    We only prove the case where $\M$ is a factor. 
    The same argument can be lifted to von Neumann algebras by using direct integration \eqref{eq:FOW_direct_integral}.
    First, note that \cref{thm:kappa} implies 
    \begin{align}
        2\geq\kappa(\omega) \ge
         \norm{\hat\omega\circ\theta_{-\log n}-\hat\omega}, \qquad n\in\NN,\ \omega\in S_*(\M).
    \end{align}
    We will show that the right-hand side converges to $2$ for all normal states $\omega$ on $\M$.
    For this, we use the explicit description of the flow of weights $(X,\mu,F)$ of a semifinite factor as $X=(0,\oo)$, $F_s(t)=e^{-s}t$ and $\mu=dt$ (see in \cref{sec:FOW_semifinite})
    \begin{align}
        \norm{\hat\omega\circ\theta_{-\log n}-\hat\omega} = \int_0^\infty | nD_\omega(nt) - D_\omega(t)|\,dt.
    \end{align}
    In fact, the right-hand side converges to $2$ for any probability density $g\in L^1(0,\oo)$.
    To see this, we may assume that $g$ is supported on $(0,b]$.
    Then $ng(tn)$ is supported on the set $(0,b/n]$ whose measure relative to $g(t)dt$ goes to zero as $n\to\oo$. 
    Set $f_n = \chi_{(0,b/n]}-\chi_{(b/n,\oo)}\in L^\oo(0,\oo)$ and note that $\norm f=1$.
    Then 
    \begin{align}
        \int_0^\oo \abs{ng(nt)-g(t)}\,dt 
        &= \int_0^{b/n}  \abs{ng(nt)-g(t)}\,dt + \int_{b/n}^\oo + \abs{ng(nt)-g(t)}\,dt\nonumber \\
        &\ge \int_0^{b/n} ng(nt)\,dt -\int_0^{b/n}g(t)\,dt + \int_{b/n}^\oo g(t)dt -\int_{b/n}^\oo ng(nt)\,dt\nonumber \\
        &\xrightarrow{n\to\oo} 1-0+1-0 =2.
    \end{align}
\end{proof}

\begin{corollary}\label{cor:kappa_diameter}
    Let $\omega$ be a state on a von Neumann algebra $\M$ which is not a finite type $\I$ factor, then $\kappa(\omega)$ is bounded by
    \begin{equation}
        \kappa(\omega)\le \diam(S_*(\M)/\!\sim).
    \end{equation}
\end{corollary}
\begin{proof}
    If $\M$ is not a factor, the right-hand side is equal to $2$, which is trivially an upper bound on $\kappa(\omega)$.
    If $\M$ is a semifinite factor, $\kappa(\omega)=2$ by \cref{cor:kappa_semifinite} and $\diam(S_*(\M)/\sim)=2$ by \cite{connes1985diameters}.
    If $\M$ is a type $\III$ factor, then we can pick unitaries $v_n \in M_{n,1}(\M)$ for each $n$ showing that 
    \begin{align}
        \kappa(\omega) &= \sup_n \sup_{\psi,\phi} \inf_{u\in\U(M_n(\M))} \norm{\omega\ox\psi - u(\omega\ox\phi)u^*}\nonumber\\
        &=\sup_n \sup_{\psi,\phi} \inf_{u\in\U(\M)} \norm{v_n(\omega\ox\psi)v_n^* - u(v_n(\omega\ox\phi)v_n^*)u^*} \nonumber\\
        &\le \diam(S_*(\M)/\!\sim).
    \end{align}
\end{proof}

\begin{corollary}\label{cor:kappa_disintegration}
    Let $\M$ be a von Neumann algebra with separable predual. Let $\M = \int^\oplus_Y\M_y\,d\nu(y)$ be a disintegration into factors $\M_y$.
    Let $\omega = \int_Y^\oplus p(y) \omega_y\,d\nu(y)$ be a normal state on $\M$ with $p(y)$ a $\nu$-absolutely continuous probability density and $Y\ni y\mapsto \omega_y\in S_*(\M_y)$ a measurable state-valued map. Then:
    \begin{enumerate}[(i)]
        \item\label{it:kappa_disintegration1}
            $\omega$ is embezzling if and only if $\omega_y\in S_*(\M_y)$ is embezzling for $\nu$-almost all $y$ with $p(y)>0$.
        \item\label{it:kappa_disintegration2}
            $y\mapsto \kappa(\omega_y)$ is measurable and $\kappa(\omega)$ is bounded by
            \begin{equation}\label{eq:kappa_non_factors}
                \kappa(\omega) \le \int_Y p(y) \kappa(\omega_y)\,d\nu(y).
            \end{equation}
        \item\label{it:kappa_disintegration3}
            If $\M$ has a discrete center,  then \eqref{eq:kappa_non_factors} holds with equality.
    \end{enumerate}
\end{corollary}
\begin{proof}
    Recall that the direct integral decomposition of von Neumann algebras implies the  direct integral decomposition of the flow of weights of general von Neumann algebras (see \cref{eq:FOW_direct_integral}). 

    \ref{it:kappa_disintegration1}:
    By \cref{cor:invariant_iff_mbz}, $\omega$ is an embezzling state if and only if $\hat\omega$ is invariant.
    Set $\varphi_\lambda = p(y) \omega_y$.
    Since $\hat\omega = \int_Y^\oplus \hat\varphi_y\,d\nu(y)$ and since $\omega\circ\theta_s = \int_Y^\oplus \hat\varphi_y\circ(\theta_y)_s \,d\nu(y)$, the spectral state $\hat\omega$ can only be invariant if for $\nu$-almost all $y$ the positive linear functional $\varphi_y$ is invariant.
    The latter is indeed equivalent to $\omega_y$ being an embezzler of $p(y)=0$ for $\nu$-almost all $y$.

    \ref{it:kappa_disintegration2}:
    We show that $y\mapsto \kappa(\omega_y)$ is a measurable function. 
    By \cite[Prop.~8.1]{haagerup1990equivalence}, $y\mapsto (\theta_y)_s$ is a measurable field of automorphisms. Therefore, $y\mapsto \norm{\hat\omega_y - \hat\omega_y\circ\theta_s}$ is measurable.
    By continuity of each $(\theta_y)_s$, we can write $\kappa$ as the pointwise-supremum of countably many measurable functions:
    \begin{equation}
        \kappa(\omega_y) = \sup_{s\in\RR} \norm{\hat\omega_y - \hat\omega_y\circ\theta_s} = \sup_{s\in\mathbb Q} \norm{\hat\omega_y - \hat\omega_y\circ\theta_s}.
    \end{equation}
    This shows that $y\mapsto \kappa(\omega_y)$ is measurable.
    The inequality \eqref{eq:kappa_non_factors} follows from dominated convergence.

    \ref{it:kappa_disintegration3}:
    This is straightforward from the definition \eqref{eq:kappa_def} of $\kappa(\omega)$ because unitaries $u = \oplus_y u_y\in\M$ are direct sums of unitaries $u_y$ which can be chosen independently.
\end{proof}

\begin{theorem}\label{thm:type_III_lambda}
    Let $\M$ be a type $\III_\lambda$ factor, $0<\lambda<1$, with separable predual.
    If $\omega$ is a faithful normal state whose modular flow $\sigma^\omega$ is periodic with minimal period $t_0=\frac{2\pi}{-\log\lambda}$, then
    \begin{equation}\label{eq:kappa_powers}
        \kappa(\omega) = 2\frac{1-\sqrt\lambda}{1+\sqrt\lambda} = \diam(S_*(\M)/\!\sim).
    \end{equation}
    Furthermore, $\M$ always admits embezzling states.
\end{theorem}

\begin{remark}\label{rem:powers_state}
    States with $t_0$-periodic modular flow exist on all type $\III_\lambda$ factors, $0<\lambda<1$ (see the introduction of \cite{haagerup1989injective}).
    An example is the Powers state 
    \begin{equation}
        \omega_\lambda = \otimes_{n=1}^\oo \varphi_\lambda, \qquad \varphi_\lambda([x_{ij}])=\frac1{1+\lambda} (\lambda x_{11} + x_{22}),
    \end{equation}
    on the hyperfinite type $\III_\lambda$ factor $\R_\lambda = \bigotimes_{n=1}^\oo (M_2,\varphi_\lambda)$.
    Since $\kappa(\omega)$ is bounded by the diameter $\diam(S_*(\M)/\!\sim)=2\frac{1-\sqrt\lambda}{1+\sqrt\lambda}$ for type $\III_\lambda$ factors, it follows that states with $t_0$-periodic modular flow, such as the Powers state, perform the worst at embezzling among all normal states.
\end{remark}

\begin{lemma}\label{lem:type_III_lambda}
    Let $\M$ be a type $\III_\lambda$ factor with separable predual where $0<\lambda\le1$.
    Let $\omega$ be a normal state on $\M$ and let $f_\omega:(0,\oo)\to(0,\oo)$ be its spectral distribution function relative to some generalized trace on $\M$ (see \cref{sec:FOW_type_III_lambda}). Then
    \begin{enumerate}
        \item $\omega$ is embezzling if and only if
        \begin{equation}\label{eq:f_embezzler}
            f_\omega(t) = \frac1{-\log\lambda} \,\frac1t,\qquad t>0.
        \end{equation}
        \item Suppose that $\omega$ faithful and that its modular flow has minimal period $t_0 = -\tfrac{2\pi}{\log\lambda}$, then $f_\omega(t) = \mu F(\mu t)$, for some $\lambda<\mu\le 1$, where
        \begin{equation}\label{eq:f_Powers}
            F(t)= \frac1{1-\lambda}\sum_{n\in\ZZ} \lambda^n \chi_{[\lambda^{n+1},\lambda^n)}(t) = \frac{\lambda^{n_\lambda(t)}}{1-\lambda} , \qquad n_\lambda(t)=\lfloor \log_\lambda(t)\rfloor,\  t>0.
        \end{equation}
    \end{enumerate}
\end{lemma}

\begin{proof}
    The first item follows from \cref{prop:FOW_of_type_III_lambda}: The unique probability distribution on $[0,\gamma_0)$, $\gamma_0=-\log\lambda$, which is invariant under the periodic shift is the uniform distribution $dP(t) = \frac1{\gamma_0}dt$.
    Thus, \eqref{eq:FOW_of_type_III_lambda} implies that if $\omega$ is an embezzling state, $f_\omega(e^{-t})= e^{-t}\frac1{\gamma_0}$.
    This forces $f_\omega(t)$ to be proportional to $\frac1t$. The proportionality factor $(-\log\lambda)^{-1}$ is determined by the normalization condition
    \begin{equation}
        1 = \int_0^{\gamma_0} dP_\omega(t) = \int_\lambda^1 f_\omega(t) dt.
    \end{equation}
    
    For the second item, we use the notation of \cref{sec:FOW_type_III_lambda}.
    The strategy is as follows: First, we construct a generalized trace $\psi$ from the state $\omega$, and show that the spectral distribution function of the state $\omega$ is given by \eqref{eq:f_Powers} if $\N_{0}=\M\rtimes_{\sigma^\psi} (\RR/t_0\ZZ)$ is constructed relative to this generalized trace.
    Second, we conclude using \cref{lem:f_invariance} that the spectral distribution function of $\omega$ will pick up a dilation if we use a different generalized trace $\phi$ for the construction.
    
    Let $u_k \in \M$, $k\in\NN$, be a realization of the $\mathcal O_\oo$-Cuntz relations, i.e., operators such that 
    \begin{equation}
        u_k^*u_l=\delta_{kl}, \qquad \sum_{k=1}^\oo u_ku_k^* =1,
    \end{equation}
    and let $u : \H\to \H \ox \ell^2(\NN)$ be the resulting unitary operator $u\Psi = \sum_i u_i \Psi \ox \ket i$.
    We consider the normal (strictly) semifinite faithful weight $\psi = u^*(\omega\ox\tr)u = \sum_k u_k^*\omega u_k$. By construction, $\sigma^{\psi}_t$ is $t_0$-periodic and $\psi(1)=+\infty$.
    Therefore, $\psi$ is a generalized trace \cite[Thm.~4.2.6]{connes1973classIII}. Since the dual action $\theta_{0}$ leaves $\M$ inside $\N_{0}$ pointwise invariant, we have $\bar\psi = u^{*}(\bar\omega\ox\tr)u$.
    This yields
    \begin{equation}\label{eq:h_0_and_Powers}
        h_0=\frac{d\bar\psi}{d\tau_0} = \sum_k u_k^* \frac{d\bar\omega}{d\tau_0} u_k = u^*\Big(\frac{d\bar\omega}{d\tau_0}\ox1\Big)u,
    \end{equation}
    which, since $h_0^{it}=\lambda_0(t)$ is $t_0$-periodic, implies
    \begin{equation}
        \Sp\big(\tfrac{d\bar\omega}{d\tau_0}\big)=\Sp h_0 \subset\{0,\lambda^n : n\in\ZZ\}.
    \end{equation}
    Consequently $f_\omega(t) = \tau_0(\chi_t(d\bar\omega/d\tau_0))$ is locally constant with jumps at the eigenvalues $t=\lambda^n$.
    Among all right-continuous non-increasing functions $f:(0,\oo)\to(0,\oo)$ such that $\int_\lambda^1 f(t)dt=1$ and $f(\lambda t)=\lambda^{-1}f(t)$, the right-hand side of \eqref{eq:f_Powers} is the unique one which is locally constant with discontinuities at the points $t=\lambda^n$, $n\in\ZZ$ (see \cref{lem:f}).
\end{proof}

\begin{proof}[Proof of \cref{thm:type_III_lambda}]
By \cref{lem:type_III_lambda}, the spectral distribution function of $\omega$ can be represented as 
\begin{align}
    f_\omega = \frac1{1-\lambda} \sum_{n\in\ZZ} \lambda^n \chi_{[\lambda^{n+1},\lambda^n)}
\end{align}
and, by \eqref{eq:hat_distance_with_f}, we have
\begin{align}
    \kappa(\omega) = \sup_s \norm{\hat\omega \circ\theta_s - \hat\omega} = \sup_s \int_\lambda^1 \left|\e^{-s}f_\omega(\e^{-s}t) -f_\omega(t)\right|\,dt.
\end{align}
Since $\lambda f_\omega(\lambda t) = f_\omega(t)$, we can choose $s$ so that $\lambda \leq e^{-s} \leq 1$. With this choice we find
\begin{align}
       \int_\lambda^1 \left|\e^{-s}f_\omega(e^{-s}t) f_\omega(t)\right|\,dt &= \frac{1}{1-\lambda}\left[\int_\lambda^{\exp(s)\lambda}\left(\frac{e^{-s}}{\lambda}-1\right)dt + \int_{\exp(s) \lambda}^1\left(1-e^{-s}\right)dt\right] \nonumber\\
    &= \frac{2}{1-\lambda}\left(1-e^{-s}\right)\left(1-e^s \lambda\right)\nonumber\\
    &= \frac{2}{1+\sqrt\lambda} \frac{1}{1-\sqrt{\lambda}}\left(1-e^{-s}\right)\left(1-e^s \lambda\right).
\end{align}
Choosing $\e^{-s} = \sqrt{\lambda}$ then yields $\kappa(\omega) \ge 2(1-\sqrt\lambda)/(1+\sqrt\lambda)$. 
\cref{cor:kappa_diameter} and \eqref{eq:diameter} show the reverse inequality so that \eqref{eq:kappa_powers} follows.

The claim that $\M$ admits embezzling states is seen as follows: By \cref{lem:f}, there exists a state $\omega\in S_*(\M)$ with $f_\omega$ equal to the right-hand side of \eqref{eq:f_embezzler} (clearly, the function is right-continuous, non-increasing and properly normalized) and by \cref{lem:type_III_lambda} this state is embezzling. 
\end{proof}

We also note the following behavior of $\kappa$ under tensor products:

\begin{lemma}\label{lem:kappa_TP}
    Let $\M$ and $\N$ be von Neumann algebras let $\omega$ and $\varphi$ be normal states on $\M$ and $\N$, respectively.
    Then
    \begin{equation}\label{eq:kappa_TP}
        \kappa(\omega\ox\varphi) \le \min\{\kappa(\omega),\kappa(\varphi)\}.
    \end{equation}
    Moreover, there exist examples where $\kappa(\omega\ox\varphi)=0$ but $\kappa(\omega),\kappa(\varphi)>0$.
\end{lemma}
\begin{proof}
    \eqref{eq:kappa_TP} is trivial from the definition of $\kappa$ (see \cref{eq:kappa_def}).
    For the last claim let $\lambda,\mu>0$ and consider $\III_\lambda$ and type $\III_\mu$ ITPFI factors $\R_\lambda$ and $\R_\mu$.
    If $\frac{\log\lambda}{\log\mu}\notin \QQ$ the tensor product $\R_\lambda\otimes\R_\mu$ is a type $\III_1$ factor \cite{araki1968factors}.
    Let $\omega$ and $\varphi$ be any non-embezzling states.
    Then $\kappa(\omega\ox\varphi)=0$ because every state on a type $\III_1$ factor is embezzling.
\end{proof}

\subsection{Universal embezzlers}\label{sec:universal}

\begin{definition}
    Let $\M$ be a von Neumann algebra. Then $\M$ is called a \emph{universal embezzler} (or a universal embezzling algebra) if all normal states $\omega\in S_*(\M)$ are embezzling.
\end{definition}

\begin{theorem}\label{thm:universal_embezzzeling_algebra}
    Let $\M$ be a von Neumann with separable predual. Then $\M$ is a universal embezzler if and only if it is a direct integral of type $\III_1$ factors.
\end{theorem}

\begin{corollary}
    A hyperfinite universal embezzler $\M$ is of the form $\M \cong Z(\M) \ox \R_\oo$ where $\R_\oo$ is the hyperfinite type $\III_1$ factor.
\end{corollary}

\begin{corollary}
    Up to isomorphism, the hyperfinite type $\III_1$ factor $\R_\oo$ is the unique hyperfinite universal embezzler.
\end{corollary}

\begin{proof}[Proof of \cref{thm:universal_embezzzeling_algebra}]
    By the compatibility of the flow of weights and direct integration (see \eqref{eq:FOW_direct_integral}), direct integrals of type $\III_1$ factors may be characterized as follows:
    A von Neumann algebra with separable predual is a direct integral of $\III_1$ factors if and only if its flow of weights $(Z(\N),\theta)$ is trivial $\theta=\id$ for all $s\in\RR$.
    Since a state $\omega\in S_*(\M)$ is embezzling if and only if $\hat\omega$ is $\theta$-invariant, every state on the flow of weights of a direct integral of type $\III_1$ factors is embezzling.
    Conversely, assume that $\hat\omega$ is invariant for all $\omega\in S_*(\M)$. 
    From \cref{prop:FOW_semifinite} and direct integration, it follows that $\M$ has no semifinite direct summand. Therefore, the main theorem of \cite{haagerup1990equivalence} implies that $\{\hat\omega:\omega\in S_*(\M)\}$ equals $\{\chi\in S_*(Z(\N)) : \chi\circ\theta_s\ge e^{-s}\chi\ \forall s>0\}$ which has dense span in $Z(\N)_*$ by
     \cite[Prop.~6.3]{haagerup1990equivalence}.
    Therefore, the flow of weights acts trivially on all normal states on $Z(\N)$ and, hence, must be trivial. Thus, $\M$ is a direct integral of type $\III_1$ factors.
\end{proof}

\begin{remark}
    One can show that a factor is a universal embezzler if and only if it has type $\III_1$ without using the flow of weights.
    The "only if" part is proved in \cref{cor:universal_mbz_spectrum}.
    The argument for the "if" part is based on the "homogeneity of the state space" of a type $\III_1$ factor, i.e., the fact that all normal states on a type $\III_1$ factor are approximately unitarily equivalent \cite{connes_homogeneity_1978}.
    Let $\omega$ a normal state on $\M$ and $\psi$ a state on $M_n$ and let $u\in M_{n,1}(\M)$ be a unitary.
    If $\M$ is type $\III_1$, the homogeneity of the state space implies that $\omega$ and $u^*\omega \ox\psi u$ are approximately unitarily equivalent.
    Multipliying these unitaries with the unitary $u\in M_{n,1}$ shows that for every $\eps>0$, there exists a unitary $v\in M_{n,1}(\M)$ such that $\norm{v\omega v^*- \omega\ox\psi}<\eps$. Since $\psi$ was arbitrary, $\omega$ is embezzling and since $\omega$ was arbitrary, $\M$ is universally embezzling.
\end{remark}

\subsection{Classification of type III factors via embezzlement}\label{sec:classification}

We introduce two algebraic invariants based on the quantifier $\kappa$, quantifying the worst and best embezzlement performance among all states: If $\M$ is a von Neumann algebra, we define
\begin{equation}
    \kappa_{\textit{min}}(\M) := \inf_{\omega\in S_*(\M)} \kappa(\omega) \qandq\kappa_{\textit{max}}(\M) := \sup_{\omega\in S_*(\M)} \kappa(\omega).
\end{equation}
In the following, we will give a detailed analysis of these invariants. We start with the second invariant $\kappa_{max}$, which turns out to be essentially equal to a well-known invariant in the case of factors:

\begin{theorem}\label{thm:max}
    Let $\M$ be a factor with separable predual, which is not of finite type $\I$. 
    \begin{equation}
        \diam(S_*(\M)/\!\sim) = \kappa_{\textit{max}}(\M)
    \end{equation}
    where $\omega_1\sim\omega_2$ if for all $\eps>0$ there exists a unitary $u\in\U(\M)$ such that $\norm{\omega_1-u\omega_2u^*}<\eps$ and where $S_*(\M)/\!\sim$ is equipped with the quotient metric $d([\omega_1],[\omega_2])=\inf_u \norm{\omega_1-u\omega_2u^*}$.
\end{theorem}

If $\M$ is a type $\I_n$ factor then $\diam(S_*(\M)/\!\sim)=2(1-\frac1n)$ while $\kappa_{\textit{max}}(\M)=2$ as we will see shortly.
The values of $\diam(S_*(\M)/\!\sim)$ are computed in terms of the classification of factors in \cite{connes1985diameters}. For semifinite factors that are not matrix algebras, the diameter is $2$, and for type $\III_\lambda$ factors with $0\le\lambda\le1$, it is $2\frac{1-\sqrt\lambda}{1+\sqrt\lambda}$ (which is a strictly monotonous function of $\lambda$).
Since $\kappa_{\textit{max}}$ computes the diameter by \cref{thm:max}, we can determine the subtype of a type $\III$ factor from its embezzlement performance: A type $\III$ factor $\M$ is type $\III_\lambda$, where $\lambda\in[0,1]$ is uniquely determined by
\begin{equation}
    \lambda = \left(\frac{1-\frac{1}{2}\kappa_{\textit{max}}(\M)}{1+\frac{1}{2}\kappa_{\textit{max}}(\M)}\right)^{2}.
\end{equation}

For semifinite factors, \cref{thm:max} follows from \cref{cor:kappa_semifinite} and for type $\III_\lambda$, $0<\lambda\le1$, it follows from \cref{thm:type_III_lambda}.
For the proof of the type $\III_0$ case, we need some preparatory results. The underlying idea for the following is that if $Z$ is an abelian von Neumann algebra and $\mc A\subset Z$ is an ultraweakly dense $C^*$-subalgebra, we can use the Gelfand-Naimark theorem to realize $\mc A$ as $C(X)$ for a compact Hausdorff space $X$, namely $X=\hat\A$ is the Gelfand spectrum of $\A$. If $Z$ admits a faithful normal state, it induces a full-support Borel probability measure $\mu$ on $X$ and $Z\cong L^\infty(X,\mu)$. 

If $(\theta_t)$ is a point-ultraweakly continuous one-parameter group of automorphisms on $Z$, we call $z\in Z$ a $\theta$-continuous element if $\norm{\theta_t(z) -z}\rightarrow 0$ as $t\rightarrow 0$. 
The set $\A$ of $\theta$-continuous elements is always an ultraweakly dense $C^*$-subalgebra of $Z$.\footnote{Density can be seen from the following: For any $z\in Z$ the sequence of $\theta$-continuous elements
\begin{align}
z_n = \sqrt{\frac{n}{\pi}}\int_{-\infty}^\infty \e^{-n t^2} \theta_t(z)\, dt
\end{align}
converges to $z$ ultraweakly.}
By definition, $(\A\cong C(X),\RR,\theta)$ is a $C^*$-dynamical system, which induces a flow $F:\RR\times X\rightarrow X$ via
\begin{align}
    \theta_t(a)(x) = a(F_{t}(x)),\qquad x\in X,\ a\in\A. 
\end{align}

The flow $F$ is jointly continuous: If $(t_i,x_i)\in \RR\times X$ is a net converging to $(t,x)$, we have for every $a \in \A$
\begin{align}
|a(F_{-{t_i}}(x_i)) - a(F_{-t}(x))| &\leq |a(F_{-t_i}(x_i)) - a(F_{-t}(x_i))| + |a(F_{-t}(x_i))-a(F_{-t}(x)| \nonumber\\
&\leq \norm{\theta_{t_i}(a)-\theta_t(a)} + |a(F_{-t}(x_i))-a(F_{-t}(x))| \rightarrow 0,
\end{align}
where we used that $\theta$ is norm-continuous on $\A$ and $a$ is a continuous function on $X$.

\begin{lemma}\label{lem:kappa_via_A}
    Let $Z$ be an abelian von Neumann algebra and let $(\theta_t)$ be a point-ultraweakly continuous one-parameter group of automorphisms.
    Let $\A\subset Z$ be the ultraweakly dense $C^*$-subalgebra of $\theta$-continuous elements. Then $\A$ is $\theta$-invariant and we have 
    \begin{equation}\label{eq:PandQ}
        P := \{\chi\in S_*(Z) : \chi\circ\theta_s \ge e^{-s}\chi,\ s\ge0\} \subseteq
        Q := \{\chi\in S(\A) : \chi\circ\theta_s \ge e^{-s}\chi,\ s\ge0\}.
    \end{equation}
    Note that $P \subset Q$ and that both are convex sets. The set $Q$ can be characterised as follows: Let $\chi\in S(\A)$, then
    \begin{equation}\label{eq:chi_omega}
        \chi\in Q \iff \exists\omega\in S(\A) \,:\, \chi = \int_{-\oo}^0 e^s \omega\circ \theta_s\,ds
    \end{equation}
    Furthermore, it holds that
    \begin{equation}\label{eq:PorQ}
        \kappa_{\textit{max}}(Z) := \sup_{\chi\in P} \sup_s \norm{\chi-\chi\circ\theta_s} = \sup_{\chi\in Q} \sup_s \norm{\chi-\chi\circ\theta_s}.
    \end{equation}
\end{lemma}
\begin{proof}
    All statements except for eq.~\eqref{eq:PorQ} are proved in \cite[Sec.~6]{haagerup1990equivalence} where it is also established that \eqref{eq:chi_omega} defines an affine bijection between the convex sets $Q$ and $S(\A)$. We denote the state $\chi\in Q$ is given by $\omega\in S(\A)$ as in \eqref{eq:chi_omega} by $\chi_\omega$.
    Denote the left and right-hand sides of \eqref{eq:PorQ} by $\kappa(P)$ and $\kappa(Q)$, respectively.
    Clearly, $P\subset Q$ implies $\kappa(P)\le \kappa(Q)$.
    Let $\eps>0$ and pick $\chi=\chi_\omega\in Q$ and $s>0$ such that $\norm{\chi-\chi\circ\theta_s}\ge \kappa(Q)-\eps$.
    Picking a net $\omega_\alpha \in S_*(Z)$ such that $\omega_\alpha\to \omega$ in the $w^*$-topology and setting $\chi_\alpha=\chi_{\omega_\alpha}$, gives us a net $\chi_\alpha\in P$ which $w^*$-converges to $\chi\in Q$.
    From $w^*$-lower semicontinuity of the norm on $\A^*$, we conclude
    \begin{align*}
        \kappa(P)-\eps \le \norm{\chi-\chi\circ\theta_s} \le \liminf_\alpha \norm{\chi_\alpha-\chi_\alpha\circ\theta_s} \le \lim_\alpha \kappa(Q) = \kappa(Q).
    \end{align*}
\end{proof}

\begin{lemma}\label{lem:period}
    Let $X$ be a compact Hausdorff space and let $F$ be a continuous flow.
    Then the period-function 
    \begin{equation}
        p:X\to[0,\oo],\quad p(x) = \inf\{t>0,\oo : F_t(x)=x\}
    \end{equation}
    is lower semicontinuous.
    If $\mu$ is a Borel measure of full support, which is quasi-invariant under the flow and such that $F$ is $\mu$-ergodic, then the period $p$ is constant almost everywhere.
\end{lemma}
\begin{proof}
    Since $p$ is a positive function, lower semicontinuity will follow if we show that all sublevel sets $p^{-1}([0,t])$ for any $t>0$ are closed. Note, that this also implies that $p^{-1}(\{0\})$ is closed since $p^{-1}(\{0\})=\bigcap_{t>0}p^{-1}([0,t])$.
Consider a net $x_n\in p^{-1}([0,t])$ with limit $x\in X$. Then, for all $n$, there exist $s_n\in(\tfrac{1}{2}t,t]$ such that $F_{s_n}(x_n)=x_n$. To see this, we note that by assumption, for all $n$, there exists $s'_{n}\in(0,t]$ with $F_{s'_n}(x_n)=x_n$. If, for some $n$, we have $s'_{n}\leq\tfrac{1}{2}t$, there exists an integer $N_{n}$ such that $N_{n}s'_{n}\in(\tfrac{1}{2}t,t]$ and $F_{N_{n}s'_{n}}(x_{n})=(F_{s'_{n}})^{N_{n}}(x_{n})=x_{n}$. By compactness of $[\tfrac{1}{2}t,t]$, there exists a subnet $s_{n(\alpha)}$ that converges to some $s\in [\tfrac{1}{2}t,t]$. Using continuity of $F$, we conclude $x\in p^{-1}([\tfrac{1}{2}t,t])$ from
    \begin{equation}
        x =  \lim_\alpha x_{n(\alpha)} = \lim_{\alpha} F_{s_{n(\alpha)}}(x_{n(\alpha)}) = F_{\lim_\alpha s_{n(\alpha)}}\Big(\lim_\alpha x_{n(\alpha)}\Big) = F_s(x).
    \end{equation}
    If $F$ is ergodic, $p$ must be constant almost everywhere because it is a measurable $F$-invariant function (to see this, observe that the preimage of every measurable subset of $[0,\oo]$ is $F$-invariant and measurable).
\end{proof}

\begin{lemma}\label{lem:kappa_if_aperiodic}
    Let $X$ be a compact Hausdorff space, let $F$ be a continuous flow on $X$ and let $\theta_s$ be the corresponding strongly continuous action on $C(X)$, i.e., $\theta_s(f)(x) = f(F_s(x))$.
    If $x\in X$ is an aperiodic point, i.e., $x\ne F_t(x)$ for all $0\neq t\in\RR$, then there exists a state $\chi$ on $C(X)$ such that
    \begin{equation}
        \chi \circ\theta_s \ge e^{-s} \chi,\quad s>0,\qandq
        \lim_{s\to\oo} \norm{\chi-\chi\circ\theta_s} =2.
    \end{equation}
\end{lemma}
\begin{proof}
    We identify states on $C(X)$ and Radon probability measures on $X$ and set $\chi = \int_{-\oo}^0 e^{t} \delta_{F_{t}(x)}dt$.
    Since $x$ is an aperiodic point, we have
    \begin{align*}
        \norm{\chi-\chi\circ\theta_s} 
        &= \norm{\int_{-\oo}^0e^t\delta_{F_t(x)}dt-\int_{-\oo}^s e^{t-s}\delta_{F_t(x)} dt}\\
        &= \norm{\int_{-\oo}^0 (1-e^{-s})e^t \delta_{F_t(x)} dt - \int_0^s e^{t-s} \delta_{F_t(x)} dt} \\
        &= (1-e^{-s})\norm{ \int_{-\oo}^0 e^t \delta_{F_t(x)} dt} + \norm{\int_0^s e^{t-s}\delta_{F_t(x)}dt}\\
        &= (1-e^{-s}) + e^{-s}\int_0^s e^{t} dt\\
        &= (1-e^{-s}) +e^{-s} (e^s-1) =2 - 2e^{-s} \to 2,
    \end{align*}
    where we used, from the second to the third line, that the two integrals define measures with orthogonal support.
\end{proof}

\begin{proof}[Proof of \cref{thm:max} for type $\III_0$ factors]
    Let $\M$ be a type $\III_0$ factor with separable predual and let $(Z(\N),\theta)$ be its flow of weights.
    Since $\M$ is a factor, $\theta$ is an ergodic flow.
    As discussed above, by considering the $C^*$-subalgebra $\A\cong C(X)$ of $\theta$-continuous elements, we obtain a continuous flow $F:\RR\times X\rightarrow X$.
    Let $\phi$ be a faithful normal state on $Z(\N)$ (such a state exists because $Z(\N)$ has separable predual). Restriction of this state to $\A$ gives a full-support Borel measure i$\mu$ on $X$ such that the isomorphism $A\cong C(X)$ extends to $L^\oo(X,\mu)\cong Z(\N)$.
    In the following, we make the identifications
    \begin{equation}
        \A = C(X) \subset L^\oo(X,\mu) = Z(\N),\qquad S_*(Z(\N))\subset S(\A) = M(X)_1^+,
    \end{equation}
    where $M(X)_1^+$ is the set of Radon probability measures on $X$.
    By the main theorem of \cite{haagerup1990equivalence}, the range of the map $S_*(\M)\ni\omega\mapsto\hat\omega\in S_*(Z(\N))$ is exactly the convex set $P$ in \eqref{eq:PandQ}.
    Therefore, \cref{lem:kappa_via_A} implies
    \begin{equation}\label{eq:typeIII_0_kappa}
        \kappa_{\textit{max}}(\M) = \sup_{\omega\in S_*(\M),\,s>0} \norm{\hat\omega-\hat\omega\circ\theta_s} = \sup_{\chi\in P,\, s>0} \norm{\chi-\chi\circ\theta_s}
        = \sup_{\chi\in Q,\, s>0} \norm{\chi-\chi\circ\theta_s}. 
    \end{equation}
    By \cref{lem:period}, almost all points of $X$ have the same period $T\in[0,\oo]$. If $T$ were finite, the flow of weights would be periodic, which contradicts the assumption that $\M$ is a type $\III_{0}$ factor (see \cref{sec:mbz_fow}).
    Thus, $T=\oo$. In particular, there exists a single point, which is aperiodic.
    Therefore, \cref{lem:kappa_if_aperiodic} implies that the right-hand side of \eqref{eq:typeIII_0_kappa} is equal to $2$, which is equal to the diameter of the state space \cite{connes1985diameters}.
\end{proof}

Let us turn our attention to the algebraic invariant $\kappa_{\textit{min}}$.
In stark contrast to $\kappa_{max}$ it is bivalent, only attaining the extremal values $0$ or $2$:
\begin{theorem}\label{thm:min}
    Let $\M$ be a von Neumann algebra with separable predual.
    Then exactly one of the following is true:
    \begin{enumerate}[(i)]
        \item There exist embezzling states and $\kappa_{\textit{min}}(\M)=0$.
        \item There exist no embezzling states and $\kappa_{\textit{min}}(\M)=2$.
    \end{enumerate}
    If $\M$ is semifinite, then $\kappa_{\textit{min}}(\M)=2$. If $\M$ is type $\III_\lambda$ factor with $0<\lambda\le1$, then $\kappa_{\textit{min}}(\M)=0$.
    There exist type $\III_0$ factors with $\kappa_{\textit{min}}(\M)=0$ and there exist $\III_0$ factors with $\kappa_{\textit{min}}(\M)=2$.
\end{theorem}
\cref{thm:min} follows from a corresponding dichotomy for ergodic systems, which we apply to the flow of weights. This was kindly pointed out to us by Amine Marrakchi.
\begin{lemma}\label{lem:amine}
    Consider a standard measure space $(X,\mu)$ with a nonsingular flow $F$, i.e., a one-parameter group of nonsingular Borel transformations $(F_t)$, that leaves $\mu$ quasi-invariant.
    The following are equivalent:
    \begin{enumerate}[(a)]
        \item\label{it:amine1} There is no nonsingular $F$-invariant probability measure on $X$.
        \item\label{it:amine2} $\sup_{t\in\RR} \norm{P-F_t(P)}=2$ for some faithful nonsingular probability measure $P$ on $X$. 
        \item\label{it:amine3} $\sup_{t\in\RR} \norm{P-F_t(P)}=2$ for all nonsingular probability measure $P$ on $X$.
        \end{enumerate}
\end{lemma}
\begin{proof}   
    The proof relies on \cite[Prop.~4.4,Prop.~4.5]{arano_ergodic_2021}.
    Consider the Koopman action on $L^2(X,\mu)$, i.e., the unitary one-parameter group $U_t\xi(x) = m_t(x)\xi(F_{-t}(x))$ where $m_t=\sqrt{dF_{t}(\mu)/d\mu}$.
    For any $\mu$-absolutely continuous probability measure $P$ on $X$, we write $\xi_P = \sqrt{dP/d\mu}$. Using standard techniques it follows that
    \begin{equation}\label{eq:L2_orth}
        \sup_{t\in\RR}\ \norm{P-F_t(P)} =2 \iff \inf_{t\in\RR}\ \ip{U_t\xi_P}{\xi_P}_{L^2(X,\mu)} =0.
    \end{equation}
    The equivalence between items \ref{it:amine1} and \ref{it:amine2} now follows from \cite[Prop.~4.5]{arano_ergodic_2021}. 
    Clearly \ref{it:amine3} implies \ref{it:amine2}. 
    For the converse, we observe that there exists a sequence of numbers $t_n\in\RR$ such that $\lim_n \ip{\xi_P}{U_{t_n}\xi_P}_{L^2(X,\mu)}=0$ by \eqref{eq:L2_orth}. 
    A closer look at the proof of \cite[Prop.~4.4]{arano_ergodic_2021}  shows that this implies $\lim_n \ip{\eta}{U_{t_n}\eta'}_{L^2(X,\mu)}=0$ for all $\eta,\eta'\in L^2(X,\mu)$.
    In particular, $\lim_n \ip{\xi_Q}{U_{t_n}\xi_Q}_{L^2(X,\mu)}=0$ for all nonsingular probability measures $Q$ on $X$. Thus, item \ref{it:amine3} follows from \eqref{eq:L2_orth}.
\end{proof}
\begin{proof}[Proof of \cref{thm:min}]
    The first assertion follows \cref{lem:amine} and the connection between $\kappa(\omega)$ and invariance of spectral states under the flow of weights (i.e., \cref{thm:kappa,lem:amine,cor:invariant_iff_mbz}).
    $\kappa_{\textit{min}}(\M)=2$ for semifinite factors is a consequence of \cref{cor:kappa_semifinite}.
    $\kappa_{\textit{min}}(\M)=0$ for type $\III_\lambda$ factors with $0<\lambda\le1$ follows from \cref{thm:type_III_lambda}.
    It remains to be shown that both $0$ and $2$ are attained among type $\III_0$ factors.
    Since every properly ergodic flow arises as the flow of weights of some type $\III_0$ factor \cite[Thm.~XVIII.2.1]{takesaki3}, it suffices to show the existence of properly ergodic flows that do and that do not satisfy the equivalent properties of \cref{lem:amine}.
    Properly ergodic flows that admit an invariant probability measure are easy to construct. For instance, we can take the irrational line flow on the torus, i.e., the flow induced by $(x,y)\mapsto (x+\alpha t,y+\alpha t)$ on $\RR^2/\ZZ^2$ with $\alpha \in[0,1]$ irrational.
    The existence of properly ergodic flows that do not admit invariant probability measures is also well-known.
    In \cref{rem:stefaan}, we explain a construction that was communicated to us by Stefaan Vaes, which exhibits the relevant property that $\sup_t \norm{P-F_t(P)}=2$ for all nonsingular probability measures $P$.
\end{proof}

\begin{remark}\label{rem:stefaan}
    In the following we describe a construction of a properly ergodic flow $(X,\mu,F)$ for which one sees that
    \begin{equation}\label{eq:dichotomy}
        \sup_t \norm{P-F_t(P)}=2
    \end{equation}
    holds for all nonsingular probability measures $P$ on $X$, without invoking \cref{lem:amine}.
    Since every properly ergodic flow arises as the flow of weights of a type $\III_0$ factor \cite[Thm.~XVIII.2.1]{takesaki3}, this yields the construction of a type $\III_0$ factors with  $\kappa_{\textit{min}}(\M)=2$.
    The construction uses a 2-odometer, which is defined as follows: Consider the standard Borel space $Y=\{0,1\}^\NN$ of infinite bit strings, which can be viewed as positive integers with infinite digits written in opposite order. The odometer transformation $T:Y\to Y$ is defined as addition of $1$ with carry-over, e.g.,
    \begin{align}
        T(0100\cdots) = 1100\cdots,\qquad T(1100\cdots) = 0010\cdots,
    \end{align}
    and more generally $T(0.y) = 1.y$ and $T(1.y) = 0.T(y)$ with $T(111\cdots) = 000\cdots$. Here $z.y$ denotes the concatenation of a finite bitstring $z$ with the infinite bitstring $y$. 
    Let $(a_n)$ be a sequence of numbers $0<a_n<1$, such that $\sum_n a_n = +\infty = \sum_n(1-a_n)$, and $a_n\rightarrow 1$. 
    We define a Borel probability measure $\nu$ on $Y$ by choosing for each $n\in\NN$ the probability measure $\nu_n$ on $\{0,1\}$ with $\nu_n(\{0\})=a_n \in (0,1)$ and setting $\nu = \prod_n \nu_n$. 
    The resulting measure $\nu$ is non-atomic and it is well known that the odometer transformation $T$ is ergodic.
    
    Now consider an event $E\subset Y$ such that $T(E)\subset Y\setminus E$. It follows for any probability measure $P$ on $X$ that
    \begin{align}
        \norm{P - T(P)} \geq 2- 4P(Y\setminus E),
    \end{align}
    where $T(P)$ denotes the Borel measure $T(P)(A)=P(T^{-1}(A))$.
    To see this, first note that $\norm{P - P|_E} = P(Y\setminus E)$, where $P|_E(A) := P(E\cap A)$. 
    Then, by the triangle inequality, we have 
    \begin{align}
        \norm{P-T(P)} \geq \norm{P|_E - T(P|_E)} - 2P(Y\setminus E).
    \end{align}
    But since $T(P|_E) = T(P)|_{T(E)}$ and since $T(E)$ and $E$ are disjoint by assumption, we find 
    \begin{align}
        \norm{P-T(P)} \geq P(E) + T(P)(T(E)) - 2P(Y\setminus E) = 2- 4P(Y\setminus E).
    \end{align}
    Now for any $P\ll \nu$ the events $Y_n = \{y\in Y: y_n = 0\}$ for $n\in\NN$ have the property that $P(Y\setminus Y_n)\rightarrow 0$ while $T^{2^n}(Y_n) \subset Y\setminus Y_n$.\footnote{If the density $f=dP/d\nu$ is $\nu$-essentially bounded, this follows from the estimate $P(A)\le \norm{f}_{L^\oo(\nu)} \nu(A)$. For $f$ unbounded, one can approximate $P$ in total variation by measures $P_k$ with bounded density and use the estimate $P(A) \le P_k(A) + \norm{P-P_k}$.}
    Repeating the above argument for $T^{2^n}$ in place of $T$, we thus find $\sup_{n\in\NN} \norm{P - T^{2^n}(P)} = 2$ for any $P\ll\nu$.
    The existence of an ergodic flow with the property \eqref{eq:dichotomy} holds for all nonsingular probability measures now follows as the flow under a ceiling function induced by $(Y,\nu,T)$, see, e.g., \ \cite{takesaki3}.
\end{remark}
Taken together, \cref{thm:max,thm:min} yield the table in \cref{introthm:tablekappa} in the Introduction.

It is interesting to ask which factors admit embezzling states. So far, we have answered this in terms of the flow of weights: A factor admits an embezzling state if and only if the flow of weights admits an invariant probability measure.
However, since the flow of weights is often hard to make explicit, a direct algebraic characterization is desirable.
In this respect, Haagerup and Musat obtained necessary and sufficient conditions for the existence of an invariant probability on the flow of weights in \cite{haagerup_classification_2009}.
Their result yields the following characterization of factors admitting embezzling states:

\begin{proposition}[Haagerup-Musat \cite{haagerup_classification_2009}]
    Let $\M$ be a type $\III$ factor with separable predual. 
    The following are equivalent:
    \begin{enumerate}[(a)]
        \item $\M$ admits a normal state that is embezzling,
        \item for each $\lambda\in(0,1)$, $\M$ admits an embedding of the hyperfinite type $\III_\lambda$ factor with normal conditional expectation,
        \item $\M$ admits an embedding of the hyperfinite type $\III_1$ factor with normal conditional expectation.
    \end{enumerate}
\end{proposition}

\subsection{Embezzling infinite systems}

The main result of this subsection is the following:

\begin{theorem}\label{thm:kappa_bound}
    Let $\M$ be a von Neumann algebras and let $\P$ be a hyperfinite factor.
    Let $\omega$ be a normal state on $\M$ and let $\psi,\phi$ be normal states on $\P$.
    Then 
    \begin{equation}\label{eq:kappa_bound}
        \inf_{u\in \U(\M\ox\P)} \norm{u(\omega\ox\psi)u^*-\omega\ox\phi} \le \kappa(\omega).
    \end{equation}
\end{theorem}

Even if we exclude $\P= M_n$ and maximize over pairs of normal states $\psi,\phi\in S_*(\P)$, we cannot expect equality.
Indeed, if $\M=\B(\H)$ and $\P$ is type $\III_1$, the left-hand side is zero for all $\psi,\phi$ while the right-hand side is equal to $2$.
However, if $\P$ is a semifinite factor, we do get equality:

\begin{proposition}\label{thm:typeII_mbz}
    Let $\M$ be a von Neumann algebra and let $\P$ be a semifinite factor which is not of finite type $\I$.
    For every normal state $\omega$ on $\M$ it holds that
    \begin{equation}\label{eq:typeII_mbz}
        \kappa(\omega) = \sup_{\psi,\phi} \inf_{u} \norm{u(\omega\ox\psi)u^*-\omega\ox\phi}
    \end{equation}
    where the supremum is over all normal states $\psi, \phi$ on $\P$ and the infimum is over all unitaries $u$ in $\M\ox\P$.
\end{proposition}

Note that \cref{thm:typeII_mbz}, implies that \cref{thm:kappa_bound} holds for semifinite $\P$.
We will first prove \cref{thm:typeII_mbz} and then use it to deduce \cref{thm:kappa_bound}:

\begin{proof}[Proof of \cref{thm:typeII_mbz}]
For the proof, we denote the right-hand side of \eqref{eq:typeII_mbz} by $\nu(\omega)$.
By \cref{thm:distance_of_unitary_orbits}, $\nu(\omega)$ is equal to the supremum of $\norm{(\omega\ox\psi)^\wedge-(\omega\ox\phi)^\wedge}$ over all $\psi,\phi$.
Combining this with \cref{prop:dual_state}, shows 
\begin{align}
    \nu(\omega)
    &=\sup_{\psi,\phi\in S_*(\P)} \bigg\| \int_0^\oo \lambda_\psi(t)\,\hat\omega\circ\theta_{\log\lambda_\psi(t)}\,dt - \int_0^\oo \lambda_\phi(s)\,\hat\omega\circ\theta_{\log\lambda_\phi(s)} ds\bigg\| \nonumber\\
    &=\sup_{\psi,\phi\in S_*(\P)} \bigg\| \int_0^\oo\!\!\!\int_0^\oo\lambda_\psi(t)\lambda_\phi(s)\Big(\hat\omega\circ\theta_{\log\lambda_\psi(t)}-\hat\omega\circ\theta_{\log\lambda_\phi(s)}\Big) \,dsdt\bigg\|
\end{align}
Therefore,
\begin{align*}
    \nu(\omega) 
    &\le \sup_{\psi,\phi\in S_*(\P)} \int_0^\oo\!\!\!\int_0^\oo \lambda_\psi(t)\lambda_\phi(s) \Big\| \hat\omega\circ\theta_{\log\lambda_\psi(t)}-\hat\omega\circ\theta_{\log\lambda_\phi(s)}\Big\| \,dsdt\\
    &= \sup_{\psi,\phi\in S_*(\P)} \int_0^\oo\!\!\!\int_0^\oo \lambda_\psi(t)\lambda_\phi(s) \Big\| \hat\omega\circ\theta_{\log\frac{\lambda_\psi(t)}{\lambda_\phi(s)}}-\hat\omega\Big\| \,dsdt\\
    &\le \kappa(\omega) \sup_{\psi,\phi\in S_*(\P)} \int_0^\oo\!\!\!\int_0^\oo \lambda_\psi(t)\lambda_\phi(s)\,dsdt = \kappa(\omega)
\end{align*}
where we used \cref{thm:kappa} in the last line.
To show the converse inequality, we consider states $\pi_i = (\tr p_i)^{-1}\tr[p_i(\placeholder)]$ defined by finite projections $p_1,p_2\in\proj(\P)$.
By \cref{cor:projections}, it holds that 
\begin{equation}
    \nu(\omega)\ge \norm{(\omega\ox\pi_1)^\wedge-(\omega\ox\pi_2)^\wedge} = \norm{\hat\omega\circ\theta_s-\hat\omega},\qquad s=\log \frac{\tr p_1}{\tr p_2}.
\end{equation}
If $\P$ is type $\I$, the set $S= \{\log\frac{\tr p_1}{\tr p_2} : p_1,p_2 \in \proj_{\textit{fin}}(\P)\}$ is $\log \mathbb Q^+$ and if $\P$ is type $\II$, then $S=\RR$. 
In both cases, $S$ is dense, and, since $\theta_s$ is continuous, \cref{thm:kappa} implies that we get $\kappa(\omega)$ if we take the supremum of $\norm{\hat\omega\circ\theta_s -\hat\omega}$ over all $s\in S$.
This proves $\kappa(\omega)=\nu(\omega)$ and finishes the proof.
\end{proof}

We now turn to the proof of \cref{thm:kappa_bound}.
Note that the case where $\P$ is a semifinite factor follows from \cref{thm:typeII_mbz}.
To reduce \cref{thm:kappa_bound} to the semifinite case, we show the following martingale Lemma:

\begin{lemma}\label{lem:martingale}
    Let $\P$ be a factor. Let $(\mc I,\le)$ be a directed set and let $\P_\alpha$, $\alpha\in\mc I$, be an increasing net of semifinite subfactors with faithful normal conditional expectations $E_\alpha:\P\to\P_\alpha$ such that
    \begin{enumerate}
        \item $E_\alpha\circ E_\beta= E_\alpha$ for $\alpha\ge\beta\in \mc I$,
        \item for all normal states $\psi$ on $\P$, $\lim_\alpha \norm{\psi-\psi\circ E_\alpha}=0$.
    \end{enumerate}
    Let $\M$ be a von Neumann algebra with a normal state $\omega$. If \cref{thm:kappa_bound} holds for all $\P_\alpha$ then it holds for $\P$.
\end{lemma}
\begin{proof}
    For a state $\psi$ on $\P$ denote by $\psi_\alpha\in S_*(\P)$ the restriction to $\P_\alpha$.
    We can apply the arguments in \cite[Sec.~7]{haagerup1990equivalence} to the system of conditional expectations $\id\ox E_\alpha:\M\ox\P\to\M\ox\P_\alpha$.
    It is proved in \cite[223]{haagerup1990equivalence}, that
    \begin{equation}
        \norm{\hat\mu-\hat\nu} =\lim_\alpha \ \norm{(\mu_{|\M\ox\P_\alpha})^\wedge-(\nu_{|\M\ox\P_\alpha})^\wedge},\qquad \mu,\nu\in S_*(\M\ox\P).
    \end{equation}
    With this, we find
    \begin{align*}
        \inf_{u\in\U(\M\ox\P)} \norm{u(\omega\ox\psi)u^*-\omega\ox\phi}
        &=\norm{(\omega\ox\psi)^\wedge-(\omega\ox\phi)^\wedge}\\
        &= \lim_\alpha \norm{(\omega\ox\psi_\alpha)^\wedge-(\omega\ox\phi_\alpha)^\wedge}\\
        &= \lim_\alpha \inf_{u\in\U(\M\ox\P_\alpha)} \norm{u(\omega\ox\psi_\alpha)u^*-\omega\ox\phi_\alpha} \le \kappa(\omega).
    \end{align*}
\end{proof}

\begin{proof}[Proof of \cref{thm:kappa_bound}]
    By \cref{lem:martingale}, we only have to show that all hyperfinite factors $\P$ can be approximated with semifinite factors $\P_\alpha$ as in \cref{lem:martingale}.
    For $\III_0$ factors, this is shown in \cite[Prop.~8.3]{haagerup1990equivalence}. For ITPFI factors, the approximating semifinite algebras are finite type $\I$ algebras, and the conditional expectations are the slice maps (see \cref{sec:itpfi}).
    Since all hyperfinite factors are ITPFI or type $\III_0$ factors (or both), this finishes the proof.
\end{proof}

\begin{remark}
    Let us comment on the assumption of hyperfiniteness in \cref{thm:kappa_bound}.
    While our technique does not work for general factors $\P$, we expect the statement to hold in full generality.
    To prove the general case, one needs to relate the flow of weights of $\P$ and $\M$ to the flow of weights of $\M\ox\P$ and relate $(\omega\ox\psi)^\wedge$ to $\hat\omega$ and $\hat\psi$ (generalizing \cref{prop:dual_state}).
    We leave it to future work to settle this.
\end{remark}

\section{Embezzlement and infinite tensor products}
\label{sec:itpfi}
Historically, the construction of factors of type $\II$ and $\III$ from infinite tensor products of factors of finite type $\I$ (ITPFI) played a crucial role in the classification program \cite{araki1968factors}. ITPFI factors belong to the more general class of hyperfinite (or approximately finite-dimensional) factors, i.e., those admitting an ultraweakly dense filtration by finite type $\I$ factors \cite{elliott1976afd, blackadar_operator_2006}. But, the only separable, hyperfinite factors that are non-ITPFI factors are of type $\III_{0}$ \cite{connes1980_non_itpfi}.\footnote{ITPFI factors are precisely those hyperfinite factors that have an approximatively transitive flow of weights \cite{connes1985_approximately_transitive}.}

In physics, ITPFI and hyperfinite factors appear naturally as well, most evidently in the context of spin chains \cite{bratteli1997oa2}, which also motivated the original construction of a family of type $\III_{\lambda}$ factors, $0<\lambda<1$, known as Powers factors \cite{powers1967uhf}, but also in the context of quantum field theory \cite{buchholz_universal_1987, yngvason_role_2005}.

In this section, we provide an elementary proof of the fact that a universally embezzling ITPFI factor $\M$ must have a homogeneous state space, i.e., $\diam(S_*(\M)/\!\sim)=0$. As the latter implies that $\M$ is of type $\III_{1}$ \cite{connes_homogeneity_1978}, this, together with the uniqueness of the hyperfinite type $\III_{1}$ factor \cite{haagerup_uniqueness_1987}, implies that $\M\cong\R_{\infty}$, known as the Araki-Woods factor.

We briefly recall some essentials about ITPFI factors \cite{araki1968factors} (cf.~\cite[Ch.~XIV]{takesaki3}):
Let us consider a sequence, $\{M_{j}\}$, of finite type $\I_{n_{j}}$ factors and an associated family of (faithful) normal states $\{\varphi_{j}\}_{}$. Then, we can form the infinite tensor product $\M_{0}=\bigotimes_{j=1}^{\infty}\M_{j}$ as an inductive limit via the diagonal inclusions $\M_{\leq n}=\bigotimes_{j=1}^{n}\M_{j}\mapsto\M_{\leq n+1}=\bigotimes_{j=1}^{n+1}\M_{j}$ together with the infinite-tensor-product state $\varphi = \otimes_{j=1}^{\infty}\varphi_{j}$. The corresponding ITPFI factor $\M$ (relative to $\{\varphi_j\}$) is given by
\begin{align}\label{eq:itpfi}
\M & = \pi_{\varphi}(\M_{0})'' = \bigotimes_{j=1}^{\infty}(\M_{j},\varphi_{j}),
\end{align}
where $\pi_{\varphi}:\M_{0}\rightarrow\B(\H_{\varphi})$ is the GNS representation of $\M_{0}$ given by $\varphi$. It follows that $\M$ can equivalently be understood as the algebra generated by the natural representation of $\M_{0}$ on the (incomplete) infinite tensor product Hilbert space $(\H_{\varphi},\Omega_{\varphi}) = \bigotimes_{j=1}^{\infty}(\H_{\varphi_{j}},\Omega_{\varphi_{j}})$ \cite{vonneumann1939itp}. Moreover, since each $\M_{j}$ is assumed to be a factor, $\M$ is also a factor, and $\varphi$ is faithful if each $\varphi_{j}$ is faithful. In particular, we are in standard form in \cref{eq:itpfi} if each $\varphi_{j}$ is faithful.
It is also evident from this construction that ITPFI factors are hyperfinite since $\bigcup_{n}\M_{\leq n}$ is ultraweakly dense in $\M$.

The following result due to Størmer \cite{stoermer1971hyperfinite1} intrinsically characterizes ITPFI factors among hyperfinite factors (see also \cite{effros1977tensor}). To state the result, let us introduce two useful concepts:

First, consider a factor $\M$ together with a type $\I_{n}$ subfactor $\M_{n}\subset\M$. Since $\M_{n}$ admits a system of matrix units $\{e^{(n)}_{ij}\}$ (see \cite[Ch.~IV, Def.~1.7]{takesaki1}), we know that $\M$ can be factorized as \cite[Ch.~IV, Prop.~1.8]{takesaki1}
\begin{align}\label{eq:hyperfinite_factorization}
\M & \cong \M_{n}\ox\M_{n}^{c},
\end{align}
where $\M_{n}^{c}=e^{(n)}_{11}\M e^{(n)}_{11}\cong \M_{n}'\cap\M$ can be identified with the relative commutant of $\M_{n}$ in $\M$. The isomorphism in \cref{eq:hyperfinite_factorization} is given by the unitary $u_{n}:\CC^{n}\ox e^{(n)}_{11}\H\rightarrow\H$ such that
\begin{align*}
u_{n}(\ket{j}\ox\Phi)=e^{(n)}_{j1}\Phi, \qquad \Phi\in e^{(n)}_{11}\H.
\end{align*}
Given a normal state $\omega$ on $\M$, we consider the restrictions $\omega_{n} = \omega_{|\M_{n}}$ and $\omega^{c}_{n} = \omega_{|\M_{n}^{c}}$ from which we can form the product state $\omega_{n}\ox\omega_{n}^{c}$ on $\M_{n}\ox\M_{n}^{c}$.

\begin{definition}[cf.~\cite{KadisonRingrose2}]\label{def:normal_factorization}
The product state $\omega^{\times}_{n}$ on $\M$ uniquely determined by $\omega_{n}\ox\omega_{n}^{c}$ via \cref{eq:hyperfinite_factorization} is called the \emph{normal factorization of $\omega$ relative to $\M_{n}$}.
\end{definition}

Second, we need the definition of a state to be asymptotically a product.

\begin{definition}[\cite{stoermer1971hyperfinite1}]\label{def:asymptotic_product_state}
Let $\M$ be a factor. A normal state $\omega$ on $\M$ is said to be \emph{asymptotically a product state} if, given $\eps>0$ and a type $\I_{m}$ factor $\M_{m}\subset\M$, there exists a type $\I_{n}$ factor $\M_{n}$ such that $\M_{m}\subset\M_{n}\subset\M$ and
\begin{align*}
\norm{\omega - \omega^{\times}_{n}} & < \eps.
\end{align*}
\end{definition}

The fundamental result concerning hyperfinite product factors is:
\begin{theorem}[\cite{stoermer1971hyperfinite1}]\label{thm:hyperfinite_product_factor}
Let $\M$ be a hyperfinite factor. Then, the following are equivalent:
\begin{itemize}
    \item[(a)] every normal state on $\M$ is asymptotically a product state.
    \item[(b)] $\M$ admits a normal state $\omega$ that is asymptotically a product state.
    \item[(c)] $\M$ is isomorphic to an ITPFI factor.
\end{itemize}
\end{theorem}

Clearly, an infinite tensor product state on an ITPFI factor has a naturally associated sequence of normal factorizations.
\begin{remark}\label{rem:state_factorization_approximation}
Let $\M$ be an ITPFI factor constructed from a sequence $\{\M_{j}\}$ of finite type $\I$ factors and faithful normal states $\varphi_{j}$. Then, the normal factorizations $\varphi^{\times}_{\leq n}$ relative to $\M_{\leq n}=\bigotimes_{j=1}^{n}\M_{j}$ agree with $\varphi = \otimes_{j=1}^{\infty}\varphi_{j}$. 
\end{remark}

\begin{remark}\label{rem:normal_factorization}
It can be shown that a hyperfinite factor $\M$ is an ITPFI factor if and only if it admits an ultraweakly dense filtration, $\{\M_{n}\}$, by factors of type $\I$ such that the associated normal factorizations $\{\omega_{n}^{\times}\}$ are convergent in norm.
\end{remark}

\subsection{Tensor-product approximations of normal states}
\label{sec:tp_approximation}

We continue by analyzing universal embezzlement for hyperfinite factors. To this end, we prove an approximation result for normal states on $\B(\H)$, i.e., those states given by density matrices $\rho$. Informally, we show that any density matrix $\rho$ on an infinite-dimensional, separable Hilbert space $\H$ can be approximated arbitrarily well as a tensor product $\rho_{\textup{fin}}\ox\bigotimes_{j}\varphi_{j}$ of a finite-dimensional density matrix $\rho_{\textup{fin}}$ and an infinite tensor product of a fixed sequence of reference states $\{\varphi_{j}\}$. 

To make this statement precise, we represent $\H$ as an infinite tensor product, $\H = \bigotimes_{j=1}^{\infty}(\H_{j},\Phi_{j})$, as discussed below \cref{eq:itpfi}.\footnote{For example, we may choose a spin chain representation relative to the all-down state: $\H_{j}=\CC^{2}$ and $\Phi_{j} = |0\rangle$.} This means, that $\H$ is the inductive limit $\H = \varinjlim_{n}\H_{\leq n}$ of the sequence of Hilbert spaces $\H_{\leq n}$ which are connected by the isometries $V_{n+1,n}\Psi_{n} = \Psi_{n}\ox\Phi_{n+1}$ for all $\Psi_{n}\in\H_{\leq n}$ relative to the sequence $\{\Phi_{j}\}$. 
In addition, we have a compatible factorization of $\B(\H)$ as an infinite tensor product of von Neumann algebras
\begin{align}
\B(\H) & = \bigotimes_{j=1}^{\infty}(\B(\H_{j}),\varphi_{j}),
\end{align}
where $\varphi_{j}=\ip{\Phi_{j}}{(\placeholder)\Phi_{j}}$ (see \cite{KadisonRingrose2, takesaki3} for further details).

Let us denote the asymptotic isometric embeddings of $\H_{\leq n}$ into $\H$ by $V_{n}$, i.e., $V_{n}\Psi_{n} = \Psi_{n}\ox\Phi_{>n}$ with $\Phi_{>n} = \ox_{j=n+1}^{\oo}\Phi_{j}$. The induced sequence of projections,
\begin{align*}
P_{n} & = V_{n}V_{n}^{*}:\H \rightarrow \H,
\end{align*}
converges strongly to the identity on $\H$ (see \cite[Lem.~3.1]{araki1966complete_boolean} for a similar statement):
\begin{align}\label{eq:projection_strong_limit}
\lim_{n\rightarrow\infty}\norm{P_{n}\Psi-\Psi} & = 0, \qquad \forall\Psi\in\H,
\end{align}
which follows directly from $P_{n|\H_{\leq m}} = 1$ and $P_{n}P_{m}=P_{m}$ for $n\ge m$.
The key observation is that we can approximate normal functionals on $\B(\H)$ by truncating with $P_{n}$:
\begin{lemma}
\label{lem:normapprox}
Given any normal functional $\omega$ on $\B(\H)$, i.e., $\omega\in\B(\H)_{*}$ such that $\omega = \tr_{\H}(T_{\omega}\placeholder)$ for some positive, normalized trace-class operator $T_{\omega}\in\B(\H)$, we have:
\begin{align}
\label{eq:normapprox}
\lim_{n\rightarrow\infty}\norm{P_{n}\omega P_{n} - \omega} & = 0.
\end{align}
\end{lemma}
\begin{proof}
Since $P_{n}$ converges strongly to the identity, we can choose, for an arbitrary pair $0\neq\Psi,\Psi'\in\H$ and $\eps>0$, an $n\in\NN$ such that $\norm{\Psi-P_{n}\Psi}<\tfrac{\eps}{\norm{\Psi}+\norm{\Psi'}}$ and $\norm{\Psi'-P_{n}\Psi'}<\tfrac{\eps}{\norm{\Psi}+\norm{\Psi'}}$. This implies the following estimate for finite-rank functionals $\ip{\Psi}{(\placeholder)\Psi'}$:
\begin{align}\label{eq:finiterankapprox}
\norm{P_{n}\ip{\Psi}{(\placeholder)\Psi'}P_{n} - \ip{\Psi}{(\placeholder)\Psi'}} & = \norm{\ip{\Psi}{P_{n}(\placeholder)P_{n}\Psi'}-\ip{\Psi}{(\placeholder)\Psi'}} \nonumber \\
& \leq \norm{\Psi-P_{n}\Psi}\norm{\Psi'}+\norm{P_{n}\Psi}\norm{\Psi'-P_{n}\Psi'} \nonumber \\
& < \eps
\end{align}
As the finite-rank functionals are norm dense in $\B(\H)_{*}$, and the truncation with $P_{n}$ is a normal completely positive map, the claim follows.
\end{proof}

An immediate consequence of this lemma is that we can approximate any normal state $\omega$ in norm via truncating with $P_{n}$:

\begin{corollary}
\label{cor:normapproxstate}
Let $\omega$ be a normal state on $\B(\H)$. Then, the sequence of normalized truncations approximates $\omega$ converges in norm:
\begin{align}
\label{eq:normapproxstate}
\lim_{n\rightarrow\infty}\norm{\tfrac{1}{\omega(P_{n})} P_{n}\omega P_{n} - \omega} & = 0,
\end{align}
where we consider the approximating sequence only for sufficiently large $n$ such that $\omega(P_{n})>0$.
\end{corollary}
\noindent For convenience, let us clarify the form of the dual action of truncating with $P_{n}$ on trace-class operators $T\in \B(\H)$:

It follows from the definition of $P_{n}$ that $P_{n}TP_{n} = V_{n}^{*}TV_{n}\ox\ketbra{\Phi_{>n}}{\Phi_{>n}}$. Thus, we can express the approximating sequence in \cref{eq:normapproxstate} as 
\begin{align*}
\tfrac{1}{\omega_{T}(P_{n})}P_{n}\omega_{T}P_{n} & = \omega_{T_{n}},
\end{align*}
where $\omega_{T} = \tr(T\placeholder)$ and $T_{n} = \tfrac{1}{\tr(P_{n}TP_{n})}P_{n}TP_{n}$. Clearly, the density matrix $T_{n}$ has the desired tensor-product form alluded to at the beginning of this section.\\

We will now argue that \cref{lem:normapprox} and, thus, \cref{cor:normapproxstate} apply to ITPFI factors as well.

Consider an ITPFI factor $\M$ given by a sequence $\M_{j}$ of type $\I_{n_{j}}$ and a family of (faithful) normal states $\varphi_{j}$. By the above, we have:
\begin{align}
\label{eq:vnaitp}
\M = \bigotimes_{j=1}^{\infty}(\M_{j},\varphi_{j}) \subset \bigotimes_{j=1}^{\infty}(\B(\H_{\varphi_{j}}),\omega_{\varphi_{j}}) = \B(\H_{\varphi}), \qquad \H_{\varphi}  = \bigotimes_{j=1}^{\infty}(\H_{\varphi_{j}},\Omega_{\varphi_{j}}),
\end{align}
where $\omega_{\varphi_{j}} = \ip{\Omega_{\varphi_{j}}}{(\placeholder)\Omega_{\varphi_{j}}}$ denotes the vector state implementing $\varphi_{j}$ in its GNS representation.

The space of normal functionals $\M_{*}\subset\B(\H_{\varphi})_{*}$ is a norm-closed subspace, and the finite-rank functionals form a norm-dense subspace (with respect to the norm $\norm{\placeholder}_{\M_{*}}$) \cite{takesaki1}. 
Thus, it is a direct consequence of \cref{lem:normapprox} that truncating with $P_{n}$, which we can intrinsically define on $\M$ because of \cref{rem:state_factorization_approximation}, allows for approximations in norm within $\M_{*}$, i.e., for any $\omega\in\M_{*}$, we have:
\begin{align}
\label{eq:normapproxvna}
\lim_{n\rightarrow\infty}\norm{P_{n}\omega P_{n} - \omega}_{\M^{*}} & = 0.
\end{align}
This can be easily verified by observing that $\omega$ is induced by a normalized trace-class operator $T\in\B(\H_{\varphi})$, and that $\norm{\placeholder}_{\M^{*}}\leq\norm{\placeholder}_{\B(\H_{\varphi})^{*}}$, which entails:
\begin{align}
\label{eq:normapproxvna_estimate}
\norm{P_{n}\omega P_{n} - \omega}_{\M^{*}} & \leq \norm{\omega_{P_{n}TP_{n}}-\omega_{T}}_{\B(\H_{\varphi})^{*}} = \norm{P_nTP_n-T}_{1},
\end{align}
where $\norm{\placeholder}$ denotes the trace norm.

Thus, we obtain the analog of \cref{cor:normapproxstate} for normal states on $\M$ (using the notation of \cref{rem:state_factorization_approximation}).
\begin{corollary}
\label{cor:matrixapprox}
For any normal state $\omega$ on $\M$ and $\eps>0$ there is an $n\in\NN$ and density matrix $\rho_{n}
\in\ox_{j=1}^{n}\B(\H_{\varphi_{j}}) = \B(\H_{\varphi_{\leq n}})$ such that
\begin{align}
\label{eq:matrixapprox}
\norm{\omega_{\rho_{n}}\ox\varphi_{\leq n}^{c} - \omega}_{\M^{*}} & < \eps,
\end{align}
where we use the natural tensor-product splitting $\M=\M_{\leq n}\ox\M_{\leq n}^{c}$ associated with the tensor-product splitting $\H_{\varphi} = \H_{\varphi_{\leq n}}\ox\H_{\varphi_{\leq n}^{c}}$. In particular, we can choose $\rho_{n} = \tfrac{1}{\tr(P_{n}\rho P_{n})}V_{n}^{*}\rho V_{n}$ for any representation $\omega = \tr(\rho\placeholder)$ provided $n$ is sufficiently large.
\end{corollary}

\begin{remark}
\label{rem:normapproxvna}
Since $\M$ is in standard form, any normal state $\omega$ is implemented by a vector $\Psi_\omega\in\H$. This means that we can choose $\rho_{n} = \ketbra{\Psi_{n}}{\Psi_{n}}$ with $\Psi_{n} = \tfrac{1}{\norm{P_{n}\Psi_{\omega}}}V_{n}^{*}\Psi_{\omega}$, and we find:
\begin{align}
\label{eq:normapproxpure_estimate}
\norm{\tfrac{1}{\omega(P_{n})}P_{n}\omega P_{n}-\omega}_{\M^{*}} &\leq \norm{|\Psi_{n}\rangle\langle\Psi_{n}|\ox|\Omega_{\varphi_{\leq n}}\rangle\langle\Omega_{\varphi_{\leq n}}|-|\Psi_{\omega}\rangle\langle\Psi_{\omega}|}_{\B(\H_{\varphi})^{*}} \\ \nonumber
& \leq 2 (\norm{\Psi_{\omega}\!-\!P_{n}\Psi_{\omega}} + (1\!-\!\norm{P_{n}\Psi_{\omega}})).
\end{align}
Thus, we obtain a vector state tensor-product approximation $\ip{\Psi_{n}\ox\Omega_{\phi_{>n}}}{(\placeholder)\Psi_{n}\ox\Omega_{\phi_{>n}}}$ of $\omega$.
\end{remark}

\begin{remark}\label{rem:normal_factorization_approx}
An analog of \cref{cor:matrixapprox} holds relative to a convergent sequence of normal factorizations $\{\omega_{n}^{\times}\}$ of faithful normal state $\omega$ on $\M$. But, because of \cref{rem:normal_factorization}, this does not enlarge the class of factors beyond ITPFI factors to which such a result applies.
\end{remark}

\subsection{The unique universally embezzling hyperfinite factor}
\label{sec:universal_mbz_hyperfinite}

We can apply \cref{cor:matrixapprox} to prove that there is a unique universally embezzling ITPFI factor $\M$, i.e., it is isomorphic to the unique hyperfinite factor of type $\III_{1}$ -- the Araki-Woods factor $\R_{\infty}$ \cite{araki1968factors}.\\

Assume that $\M$ is a universally embezzling ITPFI factor. Given a normal state $\omega\in\M_{*}$ and an $\eps>0$, by \cref{cor:matrixapprox}, we find an $n\in\NN$ and a density matrix $\rho_{n}\in \B(\H_{\varphi_{\leq n}})\cong M_{d_{n}}\ox M_{d_{n}}$ such that \cref{eq:matrixapprox} holds, where $d_{n} = \prod_{j=1}^{n}\dim(\M_{j})^{\frac{1}{2}}$. In this case, we can interpret $\omega_{\rho_{n}}\ox\varphi_{\leq n}^{c}$ as a state on $M_{d_{n}}\ox\M_{\leq n}^{c}$ because $\pi_{\phi_{\leq n}}(\M_{\leq n})\cong M_{d_{n}}$ (with the latter given in standard form). Thus, we obtain:

\begin{theorem}
\label{thm:mbz_to_typeIII}
Let $\M$ be an ITPFI factor (with the same notation as above). If for all $n\in\NN_{0}$ the state $\varphi_{\leq n}^{c}$ (with $\varphi_{0}^{c}=\varphi$) is embezzling, then $\M$ has vanishing state space diameter, i.e., $\diam(S_{*}(\M)/\!\sim)=0$. In particular, if $\M$ is universally embezzling, then $\diam(S_{*}(\M)/\!\sim)=0$.
\end{theorem}

\begin{proof}
Since $\M$ has a faithful embezzling state (namely $\varphi$), we know by \cref{thm:partite} that $\M$ is properly infinite. Therefore, we know that there is a spatial isomorphism $\M\cong\M_{\leq n}^{c}$, which is implemented by a unitary $u_{0,n}\in M_{d_{n},1}(\M_{\leq n}^{c})$, because $\M\cong M_{d_{n}}\ox\M_{\leq n}^{c}$ and, thus, there are $d_{n}$ mutually orthogonal, minimal projections $p_{j}\in M_{d_{n}}$ with $p_{j}\sim 1$ and $\sum_{j=1}^{d_{n}}p_{j} = 1$.

Now, given any two normal states $\omega_{1},\omega_{2}$ and $\eps>0$, we choose $n$ and $\rho_{n,1}, \rho_{n,2}$ according to \cref{cor:matrixapprox}, and because $\varphi_{\leq n}^{c}$ is embezzling for $\M_{\leq n}^{c}$, we find a unitary $u_{n}\in\M \cong M_{d_{n}}(\M_{\leq n}^{c})$ such that
\begin{align*}
\norm{u_{n}(\omega_{\rho_{1,n}}\ox\varphi_{\leq n}^{c})u_{n}^{*}-\omega_{\rho_{2,n}}\ox\varphi_{\leq n}^{c}}_{\M^{*}} & < \eps.
\end{align*}
This implies that
\begin{align*}
\norm{u_{n}\omega_{1}u_{n}^{*}-\omega_{2}}_{\M^{*}} & < 3\eps,
\end{align*}
and the first claim follows. For the second claim, we note that $\M$ is universally embezzling if and only if $\M_{\leq n}^{c}$ is universally embezzling.
\end{proof}
As a consequence of the homogeneity of the state space of a universally embezzling ITPFI factor, we have by \cite{connes_homogeneity_1978} and Haagerup's uniqueness result for hyperfinite type $\III_{1}$ factors \cite{haagerup_uniqueness_1987}:
\begin{corollary}\label{cor:itpfi-III1}
    Let $\M$ be an ITPFI factor. $\M$ is universally embezzling if and only if $\M$ is the unique hyperfinite factor of type $\III_{1}$, i.e., $\M\cong\R_{\infty}$.
\end{corollary}

The uniqueness and ITPFI construction of the hyperfinite type $\III_1$ factor allows us to extend the embezzling property from bipartite pure state to \emph{all} bipartite mixed states:
\begin{corollary}
\label{cor:itpfi-III1-mixed}
    Consider a bipartite system $(\H,\M,\M')$ where $\M$ (and hence $\M'$) are hyperfinite type $\III_1$ factors.
    Then every density operator $\rho$ on $\H$ is embezzling in the sense that for every unit vector $\Psi\in\CC^n\ox\CC^n$ and every $\eps>0$, there exist unitaries $u \in \M\ox M_n\ox 1$ and $u'\in\M\ox1\ox M_n$ such that
    \begin{equation}\label{eq:rho_mbz}
        \norm{uu' (\rho\ox \kettbra{11}) u^*u'^* - \rho \ox \kettbra\Psi}_1 <\eps.
    \end{equation}
\end{corollary}
\begin{proof}
    Let $\rho= \sum_{i=1}^r p_i \ketbra{\Omega_i}{\Omega_i}$ be the spectral decomposition of $\rho$. We may assume the rank $r$ to be finite.
    As explained above, we can assume the Hilbert space to be an infinite tensor product $\H = \bigotimes_{j=1}^\oo (\H_j,\Phi_j)$ of finite-dimensional Hilbert spaces $\H_j$ with unit vectors $\Phi_j\in\H_j$, such that $\M = \bigotimes_{j=1}^\oo (\M_j,\varphi_j)$ and $\M' = \bigotimes_{j=1}^\oo (\M_j,\varphi_j)$ where each $\M_j$ acts on $\H_j$.
    For every $N$, $(\H_{>N},\M_{>N},\M_{>N}')$ is a bipartite system with type $\III_1$ factors and the unit vector $\Phi_{>N}$ is embezzling.
    Therefore every density operator on $\H$ that is of the form $\sigma\ox \kettbra{\Phi_{>N}}$ (for some density operator $\sigma$ on $\H_{\le N}$) is embezzling in the sense of \eqref{eq:rho_mbz}.
    We can pick unit vectors $\Omega_j\up N \in \H_{\le N}$ such that
    \begin{equation*}
        \norm{\Omega_j - \Omega_{j}\up N \ox \Phi_{>N}} \to 0.
    \end{equation*}
    Setting $\rho\up N = \sum_j p_j\kettbra{\Omega_j\up N}\ox \kettbra{\Phi_{<N}}$, we find a sequence of embezzling density operators that converges to $\rho$ in trace norm. Hence, the same holds for $\rho$.
\end{proof}

\section{Embezzling entanglement from quantum fields}\label{sec:qft}

Entanglement in quantum field theory and the problem of its quantification is a topic that is drawing an increasing amount of attention, not least because of the growing interest in quantum information theory (see \cite{hollands2018entanglement_measures} for a recent discussion). Specifically, characterizing the entanglement structure of the vacuum is a fundamental question, with many results indicating that the vacuum should be understood as exhibiting a "maximal" amount of entanglement (see \cite{summers2011yet_more_ado}). Our findings on the possibility of embezzlement of entanglement and the structure of the involved von Neumann algebras allow for a further and, in particular, operational characterization of said maximality.

Conceptually, it is noteworthy that the modular theory of von Neumann algebras features prominently in the operator-algebraic approach to quantum field theory \cite{borchers2000on_revolutionizing, borchers2000modular_groups}, e.g., via the Bisognano-Wichmann theorem \cite{bisognano1976duality}, the determination of the type of the local observable algebras \cite{longo1982algebraic_and_modular}, modular nuclearity \cite{buchholz1986causal_independence, buchholz1990nuclear_maps1,buchholz1990nuclear_maps2}, or the construction of models \cite{brunetti2002wigner_particles,morinelli2021covariant_homogeneous,neeb2021nets_lie_groups}, and our results add to the list of applications.

\subsection{Local algebras as universal embezzlers}\label{subsec:qft}
\label{sec:local_alg}

Let us briefly recall the essential structures of algebraic quantum field theory (AQFT) in the vacuum sector mainly following \cite{baumgaertel1995oam} (see also \cite{haag_local_1996, halvorson2007aqft, buchholz2000current_trends, weiner2011algebraic_haag} for additional discussions):

The basic object of AQFT is a map,
\begin{align}\label{eq:loc_net}
    \O & \mapsto \M(\O),
\end{align}
that associates to each bounded region\footnote{For simplicity and to avoid pathological situations, one can assume that $\O$ is a diamond region (or double cone), i.e., $\O$ is the intersection of the (open) forward lightcone $V^{+}_{x}$ and the (open) backward lightcone $V^{-}_{y}$ of two points in $x,y\in\MM$ in Minkowski spacetime. Then, the algebras of more complicated regions can be built by appealing to \emph{additivity}, see \cref{eq:additivity}.} $\O\subset\MM$ of Minkowski spacetime $\MM$ a von Neumann algebra $\M(\O)$ (all acting on the same Hilbert space $\H$ sharing a common unit). The map $\M$ is referred to as a \emph{net of observable algebras}. This way, the von Neumann algebra $\M(\O)$ is thought of as the observables localized in the spacetime region $\O$. To justify the interpretation of the net $\M$ as encoding a quantum field theory, further assumptions are required:

\begin{definition}\label{def:loc_alg}
A net of observable algebras $\M$ is called {\bf local} if it satisfies the following conditions:
\begin{itemize}
    \item[(a)] \emph{Isotony}: The observables of larger spacetime regions include those of smaller spacetime regions contained in them, i.e.,
    \begin{align}\label{eq:isotony}
    \M(\O_{1})\subset\M(O_{2}) \qquad \textup{if} \qquad \O_{1}\subset\O_{2},
    \end{align}
    \item[(b)] \emph{Causality}: Observables in causally disconnected (or spacelike separated) spacetime regions commute in accordance with Einstein causality, i.e.,
    \begin{align}\label{eq:causal}
        \M(\O_{1})\subset\M(\O_{2})' \qquad \textup{if} \qquad \O_{1}\subset\O_{2}'.
    \end{align}
    Here, $\O'$ denotes the {\bf causal complement} consisting of the interior of the set of those points in $\MM$ that are spacelike to all points of $\O$.
    \item[(c)] \emph{Relativistic covariance}: The proper orthochronous Poincaré group, $\P^{\uparrow}_{+}$, considered as the symmetry group of Minkowski spacetime $\MM$, acts geometrically on the net of observables, i.e., the exists a strongly continuous, (projective) unitary representation $U:\P^{\uparrow}_{+}\rightarrow\U(\H)$ such that
    \begin{align}\label{eq:rel_cov}
        U_{g}\M(\O)U_{g}^{*} & = \M(g\O), \qquad g\in\P^{\uparrow}_{+}.
    \end{align}
\end{itemize}
\end{definition}

Depending on the physical situation under consideration, it is natural to impose further conditions on a local net $\M$, e.g., in the vacuum sector.

\begin{definition}\label{def:vacuum_net}
A local net $\M$ is said to be in or equivalently called a {\bf vacuum representation} if it satifies the following conditions:
\begin{itemize}
    \item[(d)] \emph{Completeness}: The quasi-local algebra $\fA$ given by the uniform closure of the *-algebra $\bigcup_{\O}\M(\O)$ acts irreducibly on $\H$.
    \item[(e)] \emph{Additivity}: The observables associated with a finite collection of spacetime regions $\{\O_{j}\}_{j=1}^{n}$ generate the observables associated with the joint region $\bigcup_{j=1}^{n}\O_{j}$\footnote{This conditions is sometimes relaxed to weak additivity by requiring only an inclusion of the left-hand side in the right-hand side of \cref{eq:additivity}.}, i.e.,
    \begin{align}\label{eq:additivity}
        \M\bigg(\bigcup_{j=1}^{n}\O_{j}\bigg) & = \bigvee_{j=1}^{n}\M(\O_{j}).
    \end{align}
    \item[(f)] \emph{Positive energy}: The joint spectrum of the generators $\{P^{\mu}\}_{\mu}$ of the subgroup of spacetime translations of $\P^{\uparrow}_{+}$ is contained in the forward lightcone $\overline{V}^{+}\subset\MM$.
    \item[(g)] \emph{Uniqueness of the vacuum}: There exists an (up to phase) unique vacuum vector $\Omega\in\H$, i.e., $U(g)\Omega =\Omega$ for all $g\in\P^{\uparrow}_{+}$. In particular, $0\in\overline{V}^{+}$ is a non-degenerate joint eigenvalue of the generators of spacetime translations $\{P^{\mu}\}_{\mu}$.
\end{itemize}
\end{definition}

We note that, by extending additivity \cref{eq:additivity}, a local net $\M$ allows for the construction of observable algebras associated with unbounded open regions, e.g.,
\begin{align}\label{eq:wedge_algebra}
    \M(\W) & = \bigvee_{\O\subset\W}\M(\O),
\end{align}
where $\W = \{x\in\MM\ :\ |x^{0}|<x^{1}\}$ is the standard wedge region (other wedge regions are obtained by relativistic covariance \cref{eq:rel_cov}).
Additivity serves as an essential ingredient for proving the \emph{Reeh-Schlieder property} of vacuum representations, i.e., the property that the vacuum vector $\Omega$ is cyclic for each $\M(\O)$ \cite[Thm.~1.3.2]{baumgaertel1995oam}.
In \cref{eq:additivity,eq:wedge_algebra}, $\bigvee_{j}\M_j$ denotes the von Neumann algebra generated by a collection of von Neumann algebras $\{\M_j\}$, i.e., $\bigvee_j\M_j = (\bigcup_j\M_j)''$.

Another property that is essential for our discussion of embezzlement in quantum field theory is \emph{Haag duality}. In the context of AQFT, this property is defined as follows (see also the discussion below \cref{def:bipartite_system}):

\begin{definition}\label{def:net_haag_duality}
A local net $\M$ satisfies {\bf Haag duality} if the local observables algebras of causally closed regions $\O\subset\MM$, i.e., $\O=\O''$,\footnote{For example, diamond regions are causally closed.} satisfy 
\begin{itemize}
    \item[(h)] \emph{Haag duality}: The observables commuting with $\M(\O)$ are precisely the observables localized in  the causal complement $\O'$:
    \begin{align}\label{eq:haag_duality}
        \M(\O)' & = \M(\O').
    \end{align}
\end{itemize}
If $\M$ satisfies Haag duality for all wedge regions, i.e., Poincaré transformations of the standard wedge region $\W$ in \cref{eq:wedge_algebra}, then it is said to satisfy {\bf essential duality}.
\end{definition}

\begin{remark}\label{rem:net_haag_duality}
Demanding Haag duality for all bounded open regions $\O$, not only for causally closed ones, enforces a particularly strong form of causal completeness or determinacy of a local net $\M$ as it entails
\begin{align}\label{eq:causal_completeness}
   \M(\O) & = \M(\O)'' = \M(\O')' = \M(\O''),
\end{align}
i.e., the dynamics associated with $\M$ uniquely determines the observables localized in the causal closure $\O''$ from those localized in $\O$ (see \cite[Sec.~III.4]{haag_local_1996} and \cite[Sec.~1.14]{baumgaertel1995oam} for further discussion).
In the theory of superselection sectors \cite{kastler1990superselection_sectors}, it is common to assume Haag duality but only for diamond regions (double cones), which are causally closed.
\end{remark}

In view of \cref{sec:mbz_fow}, it is crucial to know the types of von Neumann algebras that appear as local observable algebras of a local net $\M$ to decide whether and to which extent embezzlement is possible in quantum field theory (see \cite{longo1982algebraic_and_modular, baumgaertel1995oam,yngvason_role_2005,halvorson2007aqft} for a general overview):

Due to an early result of Kadison, it is known that the local observable algebras $\M(\O)$ are properly infinite for all regions $\O$ with non-empty interior \cite{kadison1963remarks_on_type}, while an observation of Borchers \cite{borchers1967remark_on} almost establishes the type $\III$ property in vacuum representations, i.e., for a nontrivial projection $p\in\M(\O_{1})$, we have $p\sim1$ in $\M(\O_{2})$ if $\overline{\O}_{1}\subset\O_{2}$\footnote{Recall that in a type $\III$ factor any nontrivial projection is equivalent to $1$, and if a von Neumann algebra is not finite and each nontrivial projection is equivalent to $1$, then it is of type $\III$.}.

In the specific case of the vacuum representation of the scalar free field (of any mass), Araki proved that the local observable algebras $\M(\O)$ are isomorphic to $\R_{\infty}$ \cite{araki1964type_of_free,araki1964vna_free_field} (see also \cite{dellantonio1968structure_of_some}), i.e., the unique hyperfinite type $\III_{1}$ factors by Haagerup's result \cite{haagerup_uniqueness_1987}. In two dimensions, combining Araki's result with the construction of $P(\Phi)_{2}$-models by Glimm, Jaffe, and others \cite{glimm1985qft_expositions} implies that the local observable algebras of said models are hyperfinite type $\III_{1}$ factors as well because of local unitary equivalence with the local observable algebras of the free field.

To generally conclude that the local observable algebras $\M(\O)$ in vacuum representations are of type $\III_{1}$, it is necessary to impose further conditions on the local net $\M$ such as the existence of certain natural scaling limits \cite{driessler1977type,buchholz_universal_1987,buchholz1995scaling_algebras1}.

Apart from the case of bounded open regions, it is known that the observable algebras $\M(\W)$ of wedge regions $\W$ are of type $\III_{1}$ in vacuum representations (see \cite[Cor.~1.10.9]{baumgaertel1995oam}). This was originally proved by Bisognano, Wichmann, and Kastler \cite{bisognano1975duality_hermitian,bisognano1976duality, longo1982algebraic_and_modular} for nets generated by Wightman fields and by Driesler and Longo \cite{driessler1975lightlike, longo1979notes} in the setting of vacuum representations of local nets. It is interesting to contrast the latter situation with that for the observable algebra $\M(V^{+})$ of the forward light cone $V^{+}$ (as the timelike analog of $\W$) is indefinite -- it can be of type $\III_{1}$, e.g., for massless free fields \cite{buchholz1977on_the_structure,longo1979notes}, but also of type $\I_{\infty}$, e.g., for gapped theories \cite{sadowski1971total_sets}.

In addition to the type $\III$ property, it is expected that the local observable algebras are hyperfinite (if one extrapolates boldly from the case of free fields and constructed models).
In the general setting of local nets, it is possible to ensure hyperfiniteness by approximation properties related to additivity \cref{eq:additivity} and conditions that control the relative size of local algebras.
\begin{definition}\label{eq:split_continuity}
    Let $\M$ be a local net. Then, $\M$ has the {\bf split (or funnel) property} if the inclusion $\M(\O_{1})\subset\M(\O_{2})$ for any pair of bounded regions such that $\overline{\O_{1}}\subset\O_{2}$ is a split inclusion \cite{doplicher1984split}, i.e., there exists a type $\I$ factor $\N$ such that
    \begin{align}\label{eq:split_inclusion}
        \M(\O_{1}) & \subset \N \subset \M(\O_{2}).
    \end{align}
    $\M$ is called {\bf inner continuous}, if local observable algebras $\M(\O)$ can be approximated from within\footnote{Similarly, the notion of outer continuity of a net $\M$ can be defined by considering decreasing collections $\{O_{j}\}_{j}$ with $\cap_{j}\O_{j}=\O$.}, i.e., given an arbitrary, increasing collection of open bounded regions $\{\O_{j}\}_{j}$ with $\cup_{j}\O_{j} = \O$, then
    \begin{align}\label{eq:outer_cont}
        \M(\O) & = \bigvee_{j}\M(\O_{j}).
    \end{align}
\end{definition}
It follows that inner continuous local nets $\M$ with the split property have hyperfinite local algebras $\M(\O)$ \cite{buchholz_universal_1987} (see also \cite[Prop.~2.28]{halvorson2007aqft}). A physically motivated derivation of the split property of a local net $\M$ is possible by nuclearity assumptions involving the Hamiltonian $H=P^{0}$ or the (local) modular operators $\Delta_{\O}$, thereby essentially limiting the number of local degrees of freedom \cite{buchholz1986causal_independence,buchholz_universal_1987,buchholz1990nuclear_maps1,buchholz1990nuclear_maps2}.

We conclude from the preceding discussion and the results presented in \cref{sec:mbz_fow,sec:itpfi}:
\begin{theorem}\label{thm:wedge_universal_mbz}
    Let $\M$ be a local net in a vacuum representation (with Hilbert space $\H$). Then, $\M(\W)$ as defined in \cref{eq:wedge_algebra} is a universal embezzler.
    
    \noindent
    In particular, if we consider the standard bipartite system $(\H,\M(\W),\M(\W'))$, then the vacuum vector $\Omega\in\H$ is embezzling.
    
    \noindent
    In addition, if $\M$ is inner continuous and has the split property, then $\M(\W)$ is isomorphic to the unique hyperfinite embezzler $\R_{\infty}$. 
\end{theorem}

\begin{theorem}\label{thm:local_universal_mbz}
    Let $\M$ be a local net in a vacuum representation. If $\M$ admits a nontrivial scaling limit in the sense of Buchholz and Verch \cite{buchholz1995scaling_algebras1}, satisfying essential duality, then the local observable algebras $\M(\O)$ of diamond regions (double cones) $\O$ are universal embezzlers.
    
    \noindent
    In particular, the restriction $\omega_{\O} = \omega_{|\M(\O)}$ of the vacuum state $\omega = \bra{\Omega}\placeholder\ket{\Omega}$ is embezzling.
    
    \noindent
    In addition, if $\M$ is inner continuous and has the split property, then each local observable algebra $\M(\O)$ is isomorphic to $\R_{\infty}\ox Z(\M(\O))$. 
\end{theorem}

\begin{remark}
    In the context of (algebraic) conformal field theory, it is possible to deduce that local observable algebras $\M(\O)$ in vacuum representations (subject to an appropriate modification of \cref{def:vacuum_net}) are of type $\III_{1}$ on general grounds \cite{brunetti1993modular_structure,gabbiani1993oa_cft}.
\end{remark}

Thus, we may loosely summarize the implications of results on embezzlement of entanglement by the statement:

\begin{center}
    \textit{Relativistic quantum fields are universal embezzlers.}
\end{center}
It is this statement, together with \cref{thm:wedge_universal_mbz,thm:local_universal_mbz}, that makes precise to which extent the vacuum of quantum field theory possesses the maximally possible amount of entanglement.

Due to the operational interpretation of embezzlement, our findings give a precise meaning to the infinite amount of entanglement present in quantum field theories. However, this interpretation needs to be taken with a grain of salt, as the status of "local operations" in quantum field theory is not fully settled (see \cite{fewster2020local_measurements} for a comprehensive discussion of operational foundations of AQFT), and it is not clear whether every unitary $u\in\M(\O)$ (or $\M(\W)$) may be interpreted as a viable operation localized in $\O$ (or $\W$). This issue might be adequately addressed by appealing to recent proposals basing local nets $\M$ on Bogoliubov's local $S$-matrices, which may be interpreted as prototypes of local operations in quantum field theory \cite{buchholz2020interacting_qft}. Interestingly, certain models, e.g., the Sine-Gordon model, have been constructed within this framework \cite{bahns2018quantum_sine_gordon,bahns2021local_nets,bahns2023equilibrium_states_lorentz}.
It is an interesting open problem to explicitly determine the unitaries required for embezzlement, for example in free quantum field theories, and relate them to the discussion of local operations in quantum field theory.

Concerning the interpretation of \cref{thm:local_universal_mbz}, we point out that in contrast with \cref{thm:wedge_universal_mbz} the bipartite interpretation of embezzlement is rather asymmetric in this case because one party, e.g., Alice, has access to the local observable algebra $\M(\O)$ while the other party, Bob, is required, according to \cref{def:bipartite_system}, to have access to the observable algebra $\M(\O)'$, which contains $\M(\O')$ by causality \cref{eq:causal} and, thus, is associated with the unbounded region $\O'$.
Another interpretation of \cref{thm:local_universal_mbz} is suggested by the monopartite setting: Since the local observable algebras are universal embezzlers, we infer from \cref{thm:kappa_bound,thm:typeII_mbz} that any state on a (locally) coupled system, either described by an arbitrary hyperfinite factor (even of type $\III$) or an arbitrary semifinite factor (type $\I$ or $\II$), can be locally prepared to arbitrary precision. This feature of the local observable algebras is in accordance with related results concerning the local preparability of states in AQFT \cite{buchholz1986on_noether,werner1987local_prep,summers1990on_independence} (see also \cite{yngvason_role_2005}).

In addition, we illustrate in the following subsection that, specifically, \cref{thm:wedge_universal_mbz} provides a simple explanation for the classic result that the vacuum of relativistic quantum fields allows for a maximal violation of Bell's inequalities.

\subsection{Violation of Bell's inequalities}
\label{subsec:bell}
Consider two parties in spacelike separated labs, Alice and Bob, and imagine that both have access to a measurement device with two different measurement settings (say $\pm$), each of which yields two possible outcomes "yes" or "no". We denote by $A_\pm$ ($B_\pm$) the events that Alice's (Bob's) measurement yields "yes" with setting $\pm$ and denote the joint distribution by $P$. 

The correlation experiment has a local hidden variable model if the probabilities can be explained by a classical probabilistic model under the assumption that the measurement setting of one party does not influence the probability distribution for the outcomes of the other party. 
A Bell inequality is an inequality that any local hidden variable model must fulfill. 
In particular, the CHSH inequality \cite{clauser_proposed_1969} states that
\begin{align}
    0\leq P(A_+) + P(B_+) + P(A_- \wedge B_-) - P(A_-\wedge B_+) - P(A_+\wedge B_-) - P(A_+\wedge B_+) \leq 1.
\end{align}
Equivalently, we can consider random variables $a_\pm$ that take values $+1$ if the outcome is "yes" and $-1$ if the outcome is "no". Then
\begin{align}
    -2 \leq \mathbb E[a_+b_++a_+b_-+a_-b_+-a_-b_-] \leq 2.
\end{align}
In the quantum mechanical setting we can model the situation by a bipartite system $(\H,\M_A,\M_B)$, a state vector $\Omega\in\H$ and self-adjoint operators $a_\pm\in \M_A$ and $b_\pm \in \M_B$ such that $-1\leq a_\pm,b_\pm\leq 1$.
Optimizing over $a_\pm,b_\pm$, we define the Bell correlation coefficient:
\begin{equation} 
    \beta(\Omega;\M_A,\M_B) = \sup_{a_\pm,b_\pm} \ip{\Omega}{(a_+b_++a_+b_-+a_-b_+-a_-b_-)\Omega}
\end{equation}
where the supremum is over all self-adjoint operators $-1\le a_\pm,b_\pm \le 1$ with $a_\pm\in\M_A$ and $b_\pm\in\M_B$.
It is well known that $\beta(\Omega;\M_A,\M_B)\leq 2\sqrt 2$ and that the value $2\sqrt{2}$ may be achieved using the bipartite system $(\CC^2\ox\CC^2,M_2,M_2)$ with state $\Omega = \frac{1}{\sqrt 2}(\ket{00}+\ket{11})$. As seen above,  $\beta(\Omega;\M_A,\M_B)>2$ shows that the correlation experiment does not admit a local hidden model explanation.
If $\M$ is a von Neumann algebra in standard form, we set 
\begin{equation}
    \beta(\omega;\M) = \beta(\Omega_\omega;\M,\M'),\qquad \omega\in S_*(\M),
\end{equation}
where $\Omega_\omega$ is the unique vector in the positive cone corresponding to $\omega$.

\begin{proposition}\label{prop:beta}
    Let $\M$ be a von Neumann algebra. Then
    \begin{equation}
        \beta(\omega;\M) \ge 2\sqrt 2 - 8\sqrt{\kappa(\omega)}.
    \end{equation}
\end{proposition}
Thus, $\omega$ is guaranteed to violate a Bell inequality if $\kappa(\omega)<\frac1{100}$.
As a direct consequence, we get a new proof of the following classic result due to S.\ J.\ Summers and the third author:

\begin{theorem}[{\cite{summers1985vacuum,summers1987,summers_bells_1995}}]
    Let $\M$ be a local net that satisfies essential duality and assume that the observable algebras $\M(\W)$ of wedge regions $\W$ (as in \cref{eq:wedge_algebra}) are of type $\III_1$. Let $\Omega\in\H$ be any unit vector, then
    \begin{equation}
        \beta(\Omega;\M(\mc W),\M(\mc W')) = 2\sqrt2.
    \end{equation}
\end{theorem}

\begin{lemma}
    Let $\M$ be a von Neumann algebra, $\omega$ a normal state on $\M$, and $\psi$ a pure state on $M_n$. Then $\beta(\omega\ox\psi)=\beta(\omega)$.
\end{lemma}
\begin{proof}[Proof of the Lemma]
    Without loss of generality, we set $\psi=\bra1\placeholder\ket1$ so that $\Omega_{\omega\ox\psi}=\Omega_\omega\ox\ket{11}\in \H\ox\CC^n\ox\CC^n$.
    Clearly $\beta(\omega\ox\psi)\ge \beta(\omega)$. 
    If $b_\pm\in M_n(\M)$, $b_\pm'\in M_n(\M')$ and if $a_\pm\in\M$ and $a_\pm'\in\M'$ denote their $(11)$-matrix entries, then
    \begin{equation}
        \ip{\Omega_{\omega\ox\psi}}{(b_+b_+'+b_+b_-'+b_-b_+'-b_-b_-')\Omega_{\omega\ox\psi}}
        = \ip{\Omega_\omega}{(a_+a_+'+a_+a_-'+a_-a_+'-a_-a_-')\Omega_\omega}.
    \end{equation}
\end{proof}

\begin{proof}[Proof of \cref{prop:beta}]
For two states $\omega,\omega'$ on $\M$ and a self-adjoint operator $x\in\B(\H)$, we have
\begin{align*}
    \ip{\Omega_{\omega'}}{x\Omega_{\omega'}}
    &=\ip{\Omega_{\omega'}}{x (\Omega_{\omega'}-\Omega_{\omega})} + \ip{\Omega_{\omega}-\Omega_{\omega'}}{x\Omega_{\omega}} + \ip{\Omega_\omega}{x\Omega_\omega}\\
    &\le 2\sqrt{\norm{\omega-\omega'}}\norm x + \ip{\Omega_\omega}{x\Omega_\omega}.
\end{align*}
Putting $x=a_+a_+' + ...$ and taking the supremum, we get
\begin{equation}
    \beta(\omega)\ge \beta(\omega')- 8 \sqrt{\norm{\omega-\omega'}}.
\end{equation}
Now consider the state $\omega_1 = \omega \ox \bra1\placeholder\ket1$ on $M_2(\M)$. By the previous Lemma, $\beta(\omega)=\beta(\omega_1)$.
Pick a unitary $u\in M_{2}(\M)$ and put $\omega'_1=u(\omega\ox\frac12\tr)u^*$.
Clearly, $\beta(\omega'_1)=2\sqrt 2$ and minimizing over the unitaries $u\in M_{2}(\M)$ shows the claim.
\end{proof}

\printbibliography
\null
\vspace{11pt}

\hrule
\null

\noindent\textsc{Email addresses:}\\[0.2cm]
\textsc{L.~van Luijk:} lauritz.vanluijk@itp.uni-hannover.de\\[0.1cm]
\textsc{A.~Stottmeister:} alexander.stottmeister@itp.uni-hannover.de \\[0.1cm]
\textsc{R.~F.~Werner:} reinhard.werner@itp.uni-hannover.de \\[0.1cm]
\textsc{H.~Wilming:} henrik.wilming@itp.uni-hannover.de
\clearpage

\end{document}